%% file: paper.tex
\documentclass[dvipsnames]{amsart}

\usepackage{amsfonts, amssymb, amsmath, amsthm}
\usepackage{mathrsfs}                                                         
\usepackage{setspace}                                                         
\usepackage{geometry}                                                         
\usepackage{slashed}
\usepackage{xcolor}                                                           
\usepackage{graphicx}                                                         
\usepackage{cancel}
\usepackage{comment}                                                          

\newtheorem{theorem}{Theorem}
\newtheorem{corollary}[theorem]{Corollary}

\newtheorem{proposition}[theorem]{Proposition}
\newtheorem{definition}[theorem]{Definition}
\newtheorem{remark}[theorem]{Remark}

\newtheorem{problem}[theorem]{Problem}

\numberwithin{equation}{section}
\numberwithin{theorem}{section}

\newcommand{\mf}[1]{\mathfrak{#1}}                                            
\newcommand{\mc}[1]{\mathcal{#1}}                                             
\newcommand{\ms}[1]{\mathsf{#1}}                                              
\newcommand{\mi}[1]{\mathscr{#1}}                                             

\newcommand{\N}{\mathbb{N}}                                                   
\newcommand{\Z}{\mathbb{Z}}                                                   
\newcommand{\R}{\mathbb{R}}                                                   
\newcommand{\Sph}{\mathbb{S}}                                                 

\newcommand{\gv}{\mathsf{g}}                                                  
\newcommand{\Lv}{\mathsf{m}}                                                  
\newcommand{\kv}{\mathsf{k}}                                                  
\newcommand{\Rv}{\mathsf{R}}                                                  
\newcommand{\wv}{\mathsf{w}}                                                  
\newcommand{\Dv}{\mathsf{D}}                                                  
\newcommand{\Gammav}{\mathsf{\Gamma}}                                         

\newcommand{\gm}{\mf{g}}                                                      
\newcommand{\Dm}{\mf{D}}                                                      
\newcommand{\Rm}{\mf{R}}                                                      
\newcommand{\Rcm}{\mf{Rc}}                                                    
\newcommand{\Rsm}{\mf{Rs}}                                                    

\newcommand{\gb}[1]{\gm^{\scriptscriptstyle (#1)}}                            
\newcommand{\gbc}[1]{\check{\gm}^{\scriptscriptstyle (#1)}}                   
\newcommand{\gs}{\bar{\gm}}                                                   

\newcommand{\nablam}{\bar{\nabla}}                                            
\newcommand{\Dvm}{\bar{\Dv}}                                                  
\newcommand{\Boxm}{\bar{\Box}}                                                

\newcommand{\sch}[1]{\mi{S} (#1)}                                             
\newcommand{\oo}[2]{\mi{O}_{#1} (#2)}                                         
\newcommand{\ix}[1]{\bar{#1}}                                                 
\newcommand{\ixd}[2]{\hat{#1}_{#2}}                                           
\newcommand{\ixr}[3]{\hat{#1}_{#2}\![#3]}                                     

\allowdisplaybreaks
\begin{document}

\title[Bulk-boundary correspondence]{The bulk-boundary correspondence for the Einstein equations in asymptotically Anti-de Sitter spacetimes}

\author{Gustav Holzegel}
\address{Mathematisches Institut \\
WWU M\"unster \\
Einsteinstrasse 62 \\
48149 M\"unster \\ Germany}
\email{gholzege@uni-muenster.de}
\address{Department of Mathematics\\
Imperial College\\
South Kensington Campus\\
London SW7 2AZ\\ United Kingdom}
\email{gholzege@imperial.ac.uk}

\author{Arick Shao}
\address{School of Mathematical Sciences\\
Queen Mary University of London\\
London E1 4NS\\
United Kingdom}
\email{a.shao@qmul.ac.uk}

\begin{abstract}
In this paper, we consider vacuum asymptotically anti-de Sitter spacetimes $( \mi{M}, g )$ with conformal boundary $( \mi{I}, \gm )$. We establish a correspondence, near $\mi{I}$, between such spacetimes and their conformal boundary data on $\mi{I}$.
More specifically, given a domain $\mi{D} \subset \mi{I}$, we prove that the coefficients $\gb{0} = \gm$ and $\gb{n}$ (the \emph{undetermined term}, or \emph{stress energy tensor}) in a Fefferman-Graham expansion of the metric $g$ from the boundary uniquely determine $g$ near $\mi{D}$, provided $\mi{D}$ satisfies a \emph{generalised null convexity condition} (GNCC).
The GNCC is a conformally invariant criterion on $\mi{D}$, first identified by Chatzikaleas and the second author, that ensures a foliation of pseudoconvex hypersurfaces in $\mi{M}$ near $\mi{D}$, and with the pseudoconvexity degenerating in the limit at $\mi{D}$.
As a corollary of this result, we deduce that conformal symmetries of $( \gb{0}, \gb{n} )$ on domains $\mi{D} \subset \mi{I}$ satisfying the GNCC extend to spacetime symmetries near $\mi{D}$.
The proof, which does not require any analyticity assumptions, relies on three key ingredients:\ (1) a calculus of vertical tensor-fields developed for this setting; (2) a novel system of transport and wave equations for differences of metric and curvature quantities; and (3) recently established Carleman estimates for tensorial wave equations near the conformal boundary.
\end{abstract}

\maketitle


\section{Introduction} \label{sec.intro}

\input{intro}

\subsection*{Data Availability Statement}

Data sharing is not applicable to this article as no datasets were generated or analysed during the current study.

\subsection*{Statement on Competing Interests}

The authors have no financial or proprietary interests in any material discussed in this article.

\section{Preliminaries} \label{sec.aads}

\input{aads}

\section{The Wave-Transport System} \label{sec.sys}

\input{system}

\section{The Carleman Estimates} \label{sec.carleman}

\input{carleman}

\section{Unique Continuation} \label{sec.proof}

\input{proof}

\section{The Main Results} \label{sec.app}

\input{app}

\section{Additional Details and Computations} \label{sec.extra}

\input{details}

\raggedright
\raggedbottom
\bibliographystyle{amsplain}
\bibliography{ads4}

\end{document}

%% file: intro.tex
\emph{Asymptotically anti-de Sitter} (abbreviated \emph{aAdS}) solutions to the $(n+1)$-dimensional Einstein-vacuum equations with negative cosmological constant, 
\begin{equation}
\label{evec} \operatorname{Ric}[g] = -n g \text{,}
\end{equation}
are spacetimes whose asymptotic geometry models the maximally symmetric solution of \eqref{evec}, \emph{anti-de Sitter} (AdS) space.
Recall AdS spacetime can be globally represented as
\begin{equation}
\label{adsglobal} ( \R^4, g_\textrm{AdS} ) \text{,} \qquad g_\textrm{AdS} := - ( 1+r^2 ) dt^2+ ( 1+r^2 )^{-1} dr^2 + r^2 \mathring{\gamma}_{n-1} \text{,}
\end{equation}
where $g_\textrm{AdS}$ is expressed in polar coordinates, and where $\mathring{\gamma}_{n-1}$ is the unit round metric on $\Sph^{n-1}$. 

The distinguishing feature of aAdS spacetimes, in contrast to asymptotically flat settings, is the existence of a timelike conformal boundary at infinity.
This undermines global hyperbolicity, requiring the specification of suitable boundary conditions in addition to Cauchy data for a well-posed dynamical formulation of \eqref{evec}; see \cite{Enciso, Friedrich}.
Globally, this leads to very rich behaviour and requires understanding an entire range of novel phenomena, such as superradiant instabilities \cite{CardosoDias} and stable trapping \cite{HolSmul} in the case of aAdS black holes.
In particular, the nonlinear (in)stability properties of AdS spacetime and the Kerr-AdS family of black holes are still not known, although considerable progress has been made on various model problems, most notably the recent breakthrough \cite{Moschidis}.

Asymptotically AdS spacetimes have also seen a resurgence of interest in the physics literature, in view of the AdS/CFT conjecture \cite{gubs_kleban_polyak:ads_cft, Malda, Witten}, which, roughly, posits a correspondence between the gravitational dynamics in the aAdS spacetime interior and a conformal field theory on the boundary.
Despite its prominence in physics, there are relatively few rigorous mathematical statements pertaining to the AdS/CFT correspondence, especially in dynamical settings.
In fact, almost all known rigorous results have been in stationary or static contexts; see, e.g., \cite{alex_baleh_nach:inv_area, and_herz:uc_ricci, and_herz:uc_ricci_err, Biquard, witt_yau:ads_cft}.

In this paper, we formulate and prove a purely classical version of this correspondence, relating the geometry of the conformal boundary to the interior geometry near the boundary.
In particular, this serves as the first such rigorous result for dynamical (time-dependent) aAdS spacetimes.

\subsection{Fefferman-Graham Expansions} \label{sec:introFG}

As our main interests lie near the conformal boundary, it will be useful to express aAdS metrics in a form that centres the boundary geometry.
In the case of AdS spacetime, one convenient method for achieving this is to apply the change of coordinate
\[
4 r := \rho^{-1} ( 2 + \rho ) ( 2 - \rho ) \text{,} \qquad \rho \in ( 0, 2 ] \text{,}
\]
which transforms \eqref{adsglobal} into the so-called \emph{Fefferman-Graham gauge}:\footnote{A different possibility is the more well-known conformal embedding of AdS spacetime into half the Einstein cylinder. However, the Fefferman-Graham gauge is a more convenient form for studying general aAdS metrics.}
\begin{equation}
\label{eq.intro_ads_fg} g_\textrm{AdS} = \rho^{-2} \left[ d \rho^2 + ( - dt^2 + \mathring{\gamma}_{n-1} ) - \frac{1}{2} \rho^2 ( dt^2 + \mathring{\gamma}_{n-1} ) + \frac{1}{16} \rho^4 ( - dt^2 + \mathring{\gamma}_{n-1} ) \right] \text{.}
\end{equation}

For general aAdS geometries, one can apply a similar transformation into a Fefferman-Graham (\emph{FG}) gauge, characterized by a boundary defining function $\rho$ that is both normalised and fully decoupled from the other components.\footnote{See \cite{gra:vol_renorm} for a treatment of asymptotically hyperbolic manifolds---the Riemannian analogue of our setting. In fact, the transformation to FG gauges in \cite{gra:vol_renorm} extends directly to Lorentzian, aAdS settings.}
As a result, in this paper, we will \emph{define} the \emph{aAdS} spacetimes that we consider in terms of such FG gauges.
We refer to these as \emph{FG-aAdS segments}, representing an appropriate near-boundary spacetime patch along with adapted coordinates:

\begin{definition} \label{def:basic1}
Let $( \mi{I}, \gm )$ be a smooth $n$-dimensional Lorentzian manifold, and let $\rho_0 > 0$.
We say that $( \mi{M} := ( 0, \rho_0 ] \times \mi{I}, g )$ is a \emph{vacuum FG-aAdS segment}, with \emph{conformal infinity} $( \mi{I}, \gm )$, if $g$ satisfies the Einstein-vacuum equations \eqref{evec}, and it can be expressed in the FG gauge,
\begin{equation}
\label{FGform} g = \rho^{-2} [ d \rho^2 + \gv (\rho) ] \text{,}
\end{equation}
where $\gv (\rho)$, $\rho \in ( 0, \rho_0 ]$ is a smooth family of Lorentzian metrics on $\mi{I}$ (i.e.\ a \emph{vertical metric}) that also extends continuously as a Lorentzian metric to $\{0\} \times \mi{I}$, and with $\gv (0) = \gm$.
\end{definition}

The reader is referred to Section \ref{sec.aads_fg} for a more detailed development of FG-aAdS segments, as well as for precise definitions.
In particular, observe from \eqref{eq.intro_ads_fg} that (time strips of) AdS spacetime can itself be expressed as a vacuum FG-aAdS segment, with the standard conformal infinity
\begin{equation}
\label{adsi} ( \mi{I}_\textrm{AdS}, \gm_\textrm{AdS} ) := ( ( T_-, T_+ ) \times \Sph^{n-1}, -dt^2 + \mathring{\gamma}_{n-1} ) \text{,} \qquad T_- < T_+ \text{.}
\end{equation}
More generally, a large class of vacuum FG-aAdS segments with conformal infinity \eqref{adsi} arises by solving a boundary-initial value problem for the Einstein-vacuum equations; see \cite{Enciso, Friedrich}.

If $( \mi{M}, g )$ is a vacuum FG-aAdS segment, with conformal infinity $( \mi{I}, \gm )$, then the Einstein-vacuum equations imply the following formal series expansion for $g$ near $\rho=0$:
\begin{align}
\label{eq.aadsi} \gv ( \rho ) &= \begin{cases} \gb{0} + \gb{2} \rho^2 + \dots + \gb{n-1} \rho^{n-1} + \gb{n} \rho^n + \dots & n \text{ odd,} \\ \gb{0} + \gb{2} \rho^2 + \dots + \gb{n-2} \rho^{n-2} + \gb{\star} \rho^n \log \rho + \gb{n} \rho^n + \dots & n \text{ even,} \\ \end{cases}
\end{align}
where the $\gb{k}$'s and $\gb{\star}$ are tensor fields on $\mi{I}$.
Note that the leading coefficient $\gb{0} = \gm$ is simply the boundary metric.\footnote{We will use the notations $\gm$ and $\gb{0}$ interchangeably, depending on context.}
Furthermore, the Einstein-vacuum equations imply that all coefficients $\gb{k}$ for $0 < k < n$, as well as $\gb{\star}$ when $n$ is even, are determined locally by $\gm$ and its derivatives.
In particular, for $n\geq 3$, $-\gb{2}$ is precisely the Schouten tensor of $\gm$, namely,
\[
\mc{P} [ \gm ] := \frac{1}{n-2} \left( \operatorname{Ric} [\gm] - \frac{1}{2(n-1)} \operatorname{R} [\gm] \cdot \gm \right) \text{.}
\]

For the coefficient $\gb{n}$, the Einstein-vacuum equations imply that there exist universal functions $\mc{F}, \mc{G}$---depending only on the boundary dimension $n$---such that\footnote{Moreover, both $\mc{F}$ and $\mc{G}$ are identically zero if $n$ is odd.}
\begin{equation}
\label{constraints} \operatorname{div}_{ \gm } \gb{n} = \mc{F} ( \gm, \partial \gm, \dots, \partial^{n+1} \gm ) \text{,} \qquad \operatorname{tr}_{ \gm } \gb{n} = \mc{G} ( \gm, \partial \gm, \dots, \partial^n \gm ) \text{,}
\end{equation}
that is, the divergence and the trace of $\gb{n}$ are determined by $\gm$.
On the other hand, the remaining components of $\gb{n}$ are free---they are not formally determined by the Einstein-vacuum equations.
Moreover, assuming sufficient regularity for $\gv ( \rho )$, the expansion \eqref{eq.aadsi} can be continued beyond $\gb{n}$, with all subsequent coefficients formally determined by the pair $( \gb{0}, \gb{n} )$ alone.\footnote{When $n$ is even, the expansion remains polyhomogeneous beyond $\gb{n}$.}

Thus, we henceforth refer to $( \mi{I}, \gb{0}, \gb{n} )$ as \emph{holographic data}, or a \emph{boundary triple}, if $\gb{0}$ is a Lorentzian metric on $\mi{I}$ and $\gb{n}$ is a symmetric $2$-tensor on $\mi{I}$ satisfying \eqref{constraints}.\footnote{Our terminology arises from the common description of the AdS/CFT correspondence in physics as ``holographic". This is due to the difference in dimension between the aAdS spacetime and its conformal boundary.}

\begin{remark}
The interpretation of $\gb{0}$, as describing the geometry of the conformal boundary, is clear.
In addition, in the physics literature, $\gb{n}$ is closely connected to the stress-energy tensor for the boundary conformal field theory; see \cite{deharo_sken_solod:holog_adscft, sken:aads_set} for further discussions.
\end{remark}

The expansions \eqref{eq.aadsi}, which are widely used in the physics literature, can be formally derived by adapting the seminal works \cite{MR837196} of Fefferman and Graham to aAdS settings.
For \emph{real-analytic} holographic data $( \mi{I}, \gb{0}, \gb{n} )$, one can employ Fuchsian techniques to show \cite{Kichenassamy} that the infinite expansion \eqref{eq.aadsi} converges near $\mi{I}$ to a vacuum aAdS metric.\footnote{However, note this does not a priori prevent the possibility that there exist other vacuum metrics realising the boundary data whose Fefferman-Graham expansions do not converge.}

For generic (non-analytic) settings, where the full expansion \eqref{eq.aadsi} needs not converge, \cite{shao:aads_fg} showed rigorously that a vacuum FG-aAdS segment must still satisfy a \emph{partial FG expansion}.
More specifically, $\gv ( \rho )$ retains the form \eqref{eq.aadsi}, but only up to $n$-th order.
Nonetheless, the view of $( \mi{I}, \gb{0}, \gb{n} )$ as free boundary data (with the constraint \eqref{constraints}) for vacuum aAdS spacetimes persists.
A summary of the precise results of \cite{shao:aads_fg} can be found in Theorem \ref{thm.aads_fg} and Corollary \ref{thm.aads_fg_exp} below.

\subsubsection{Gauge Covariance}

The term conformal infinity arises from a special gauge covariance inherent to aAdS spacetimes.
Here, one can transform the boundary defining function $\rho$ in a manner that preserves the FG gauge condition \eqref{FGform} but alters the corresponding FG expansion \eqref{eq.aadsi}.
One can show that the boundary metric $\gb{0}$ then undergoes a conformal transformation,
\begin{equation}
\label{eq.intro_gauge_0} \gb{0} \mapsto \gbc{0} = e^{ 2 \mf{a} } \gb{0} \text{.}
\end{equation}
Thus, another way of phrasing this is that only the conformal class $[ \gm ]$ of the induced boundary metric can be invariantly associated with a given aAdS spacetime.

The other coefficients in \eqref{eq.aadsi} are also transformed via changes of FG gauge (see \cite{deharo_sken_solod:holog_adscft, imbim_schwim_theis_yanki:diffeo_holog}), though the formulas quickly become rather complicated.
In particular, there is a known, and in principle explicitly computable, function $\mc{H}$---depending on $\mf{a}$, $\gb{0}$, and $\gb{n}$---such that $\gb{n}$ transforms as
\begin{equation}
\label{eq.intro_gauge_n} \gb{n} \mapsto \gbc{n} = \mc{H} ( \partial^{ \leq n } \mf{a}, \partial^{ \leq n } \gb{0}, \gb{n} ) \text{.}
\end{equation}

As a result, we refer to pairs $( \gb{0}, \gb{n} )$ and $( \gbc{0}, \gbc{n} )$ as \emph{gauge-equivalent} when they are related via the formulas \eqref{eq.intro_gauge_0} and \eqref{eq.intro_gauge_n}.
The physical significance is that gauge-equivalent pairs should be viewed as ``the same", since they arise from the same aAdS spacetime.

\begin{remark}
The most general formulation of gauge equivalence can be expressed as two boundary data triples $( \mi{I}, \gb{0}, \gb{n} )$ and $( \check{\mi{I}}, \gbc{0}, \gbc{n} )$ satisfying \eqref{eq.intro_gauge_0}, \eqref{eq.intro_gauge_n} after pulling back through some boundary diffeomorphism $\smash{ \phi: \mi{I} \leftrightarrow \check{\mi{I}} }$.
However, for convenience, we will always restrict, without any loss of generality, to the case when $\phi$ is the identity map.
\end{remark}

\subsection{The Main Results}

While the above discussion shows that any vacuum FG-aAdS segment induces some holographic data $( \mi{I}, \gb{0}, \gb{n} )$, it is also natural to ask the converse---in what sense does the holographic data $( \mi{I}, \gb{0}, \gb{n} )$ determine an Einstein-vacuum metric that realises this data.
In view of the timelike nature of the boundary and the hyperbolicity of the Einstein-vacuum equations, this is generally an ill-posed problem, and hence one cannot expect existence and continuous dependence of the infilling geometry on the boundary quantities.

Instead, the appropriate mathematical framework is that of \emph{unique continuation} for the Einstein-vacuum equations, leading us to the following more precise questions:

\begin{problem} \label{prb.intro_main}
Given holographic data $( \mi{I}, \gb{0}, \gb{n} )$---up to gauge equivalence for $( \gb{0}, \gb{n} )$---\emph{and} a vacuum FG-aAdS segment $( \mi{M}, g )$ that realises this data:
\begin{enumerate}
\item Is $( \mi{M}, g )$ unique, that is, is this the only aAdS solution realising this holographic data?

\item Does $( \mi{M}, g )$ necessarily inherit the symmetries of $( \mi{I}, \gb{0}, \gb{n} )$?
\end{enumerate}
\end{problem}

Note in particular that (1) in the above can be interpreted as asking \emph{whether there is a one-to-one correspondence between vacuum aAdS spacetimes (gravity) and some appropriate space of holographic data on the conformal boundary (conformal field theory)}.

Our paper provides an \emph{affirmative answer to both questions in Problem \ref{prb.intro_main}}, provided \emph{the conformal boundary also satisfies a gauge-invariant geometric condition}---which we call the \emph{generalised null convexity criterion}, or \emph{GNCC}, first identified in \cite{ChatzikaleasShao}.
This GNCC will be defined and discussed in Section \ref{sec.intro_GNCC} below (Definition \ref{def:admissibleright}), but let us first state informal versions of our main results.

The following theorem answers question (1) of Problem \ref{prb.intro_main}:

\begin{theorem}[Bulk-boundary correspondence, informal version] \label{theo:adscft}
Let $n > 2$, and consider vacuum FG-aAdS segments $( \mi{M}, g )$ and $( \check{\mi{M}}, \check{g} )$, inducing holographic data $( \mi{I}, \gb{0}, \gb{n} )$ and $( \mi{I}, \gbc{0}, \gbc{n} )$, respectively.
Also, let $\mi{D} \subset \mi{I}$ such that $( \mi{D}, \gb{0} )$ satisfies the GNCC.
If $( \gb{0}, \gb{n} )$ and $( \gbc{0}, \gbc{n} )$ are gauge-equivalent on $\mi{D}$, then $( \mi{M}, g )$ and $( \check{\mi{M}}, \check{g} )$ must be isometric near $\mi{D}$.
\end{theorem}

The precise version of Theorem \ref{theo:adscft} that we will prove is stated as Theorem \ref{thm.adscft_gauge} further below.
Furthermore, the special case in which $( \gb{0}, \gb{n} ) = ( \gbc{0}, \gbc{n} )$, which forms the heart of the unique continuation analysis, is treated separately in Theorem \ref{thm.adscft}.

\begin{remark} \label{rmk.intro_adscft_mfld}
Since $( \mi{M}, g )$ and $( \check{\mi{M}}, \check{g} )$ have the same boundary manifold $\mi{I}$, by Definition \ref{def:basic1}, both $\mi{M}$ and $\check{\mi{M}}$ are products of an interval with $\mi{I}$.
Therefore, we can, without loss of generality, assume $\mi{M} = \check{\mi{M}}$; we make this simplification in the statements of Theorems \ref{thm.adscft} and \ref{thm.adscft_gauge}.
\end{remark}

\subsubsection{Extension of Symmetries}

An important application of Theorem \ref{theo:adscft} toward proving extension of symmetry results on aAdS spacetimes---namely, point (2) from Problem \ref{prb.intro_main}.

\begin{theorem}[Extension of Killing fields, informal version] \label{theo:killing}
Let $n > 2$, consider a vacuum FG-aAdS segment $( \mi{M}, g )$ with holographic data $( \mi{I}, \gb{0}, \gb{n} )$, and let $\mi{D} \subset \mi{I}$ such that $( \mi{D}, \gb{0} )$ satisfies the GNCC.
If $\mf{K}$ is a vector field on $\mi{I}$ that is \emph{holographic Killing} on $\mi{D}$, that is,
\begin{equation}
\label{eq.killing_ass} ( \mi{L}_{ \mf{K} } \gb{0}, \mi{L}_{ \mf{K} } \gb{n} ) |_{ \mi{D} } = 0 \text{,}
\end{equation}
then $\mf{K}$ extends to a ($g$-)Killing field $K$ near $\mi{D}$ in $\mi{M}$.
\end{theorem}

See Theorem \ref{thm.app_killing} for the precise statement of this result.
Moreover, the conclusion of Theorem \ref{theo:killing} remains valid if $\mf{K}$ is merely \emph{holographic conformal Killing}, that is, \eqref{eq.killing_ass} holds instead for data $( \gbc{0}, \gbc{n} )$ that is gauge-equivalent to $( \gb{0}, \gb{n} )$; see again Theorem \ref{thm.app_killing}.

One immediate consequence of Theorem \ref{theo:killing} and the classical Birkhoff theorem is the following rigidity result for the Schwarzschild-AdS family of spacetimes:

\begin{corollary}[Rigidity of Schwarzschild-AdS] \label{theo:killing_schwarzschild}
Let $n > 2$, let $( \mi{M}, g )$ denote a vacuum FG-aAdS segment with holographic data $( \mi{I} := ( T_-, T_+ ) \times \Sph^{n-1}, \gb{0}, \gb{n} )$, and let $\mi{D} \subset \mi{I}$ such that $( \mi{D}, \gb{0} )$ satisfies the GNCC.
If $\gb{0}$ and $\gb{n}$ are both spherically symmetric on $\mi{D}$, then $( \mi{M}, g )$ must be isometric to a domain of the Schwarzschild-AdS spacetime near $\mi{D}$.
\end{corollary}

\begin{remark}
The statement of Corollary \ref{theo:killing_schwarzschild} can be refined by noting that if $\gb{0} = \gm_\textrm{AdS}$, then the GNCC holds for a time slab $\mi{D} := \{ t_- < t < t_+ \}$ whenever $t_+ - t_- > \pi$ (see Proposition \ref{thm.intro_GNCC_ads}).
\end{remark}

Next, Theorem \ref{theo:killing} can in fact be viewed as a special case of a more general result:

\begin{theorem}[Extension of symmetries, informal version] \label{theo:symmetries}
Let $n > 2$, consider a vacuum FG-aAdS segment $( \mi{M}, g )$ with holographic data $( \mi{I}, \gb{0}, \gb{n} )$, and let $\mi{D} \subset \mi{I}$ be such that $( \mi{D}, \gb{0} )$ satisfies the GNCC.
If $\phi$ is a boundary diffeomorphism such that $( \phi_\ast \gb{0}, \phi_\ast \gb{n} )$ and $( \gb{0}, \gb{n} )$ are gauge-equivalent on $\mi{D}$, \footnote{Informally, $\phi$ is a \emph{holographic conformal symmetry} on $\mi{D}$.} then $\phi$ extends to an isometry of $( \mi{M}, g )$ near $\mi{D}$.
\end{theorem}

See Theorem \ref{thm.app_symm} for the precise statement and proof of this result.
In particular, Theorem \ref{theo:symmetries} also applies to discrete symmetries that are not generated by Killing vector fields.
One immediate consequence of Theorem \ref{theo:symmetries}---which cannot be inferred directly from Theorem \ref{theo:killing}---is that time periodicity of the conformal boundary is inherited by the bulk spacetime:

\begin{corollary}[Extension of time periodicity] \label{theo:symmetries_periodic}
Let $n > 2$, let $( \mi{M}, g )$ be a vacuum FG-aAdS segment with holographic data $( \mi{I}, \gb{0}, \gb{n} )$, and let $\mi{D} \subset \mi{I}$ such that $( \mi{D}, \gb{0} )$ satisfies the GNCC.
If $\gb{0}$ and $\gb{n}$ are both time-periodic on $\mi{D}$, then $( \mi{M}, g )$ must be time-periodic near $\mi{D}$.
\end{corollary}

\subsubsection{Previous and Related Work}

The Riemannian analogue of Theorem \ref{theo:adscft} was proven by Biquard \cite{Biquard} using a Carleman estimate of Mazzeo \cite{mazzeo} for asymptotically hyperbolic manifolds; see also \cite{and_herz:uc_ricci, and_herz:uc_ricci_err}.
The work of Biquard was then generalised by Chrusciel and Delay \cite{ChruscDelay} to an analogue of Theorem \ref{theo:adscft}, under the restriction that the spacetimes are \emph{stationary}.
Also, \cite[Theorem 1.6]{ChruscDelay} is an analogue of our Theorem \ref{theo:symmetries}, again assuming a priori that the spacetime is stationary.

We note that a fundamental ingredient in \cite{and_herz:uc_ricci, and_herz:uc_ricci_err, Biquard, ChruscDelay} is that the key equations are elliptic in nature.
In contrast, our main theorems, which are centred around hyperbolic equations, constitute the first correspondence and symmetry extension results in general dynamical settings.

We also recall, again in the Riemannian context, the well-known result of Graham and Lee \cite{gra_lee:fg_riem_ball}, which shows (for $n \geq 4$) existence of asymptotically hyperbolic Einstein metrics on the Poincar\'e ball $\mathbb{B}^{n}$ with prescribed conformal infinity on the boundary, provided the boundary metric is sufficiently close (in the $C^{2,\alpha}$-norm) to the round metric on $\Sph^{n-1}$.
Note this corresponds to solving an elliptic Dirichlet problem, which has no analogue for hyperbolic equations.

In the Lorentzian context, we first mention the programme of Anderson \cite{andpro, anderson}.
In \cite{anderson}, a \emph{conditional} global symmetry extension result for stationary Killing vectors was established under global a priori assumptions on $( \mi{M}, g )$ (including convergence to stationarity as $t \rightarrow \pm \infty$), assuming that a unique continuation property holds from $\mi{I}$ for the linearized Einstein equations.

Moreover, extension results for Killing fields have seen several applications in general relativity.
For instance, (unconditional) Killing extension theorems have been established in the contexts of black hole rigidity \cite{alex_io_kl:unique_bh, alex_io_kl:rigid_bh, gior:killing_em, ionescusk}, cosmic censorship \cite{peters:cpt_cauchy}, and non-existence of time-periodicity \cite{alex_schl:time_periodic}.
The proofs of these results revolve around proving unique continuation for a system of tensorial wave and transport equations that is similar to the system studied in this paper; see Section \ref{sec.intro_ik}.

Returning to the aAdS setting, \cite{ChatzikaleasShao, hol_shao:uc_ads, hol_shao:uc_ads_ns, mcgill_shao:psc_aads} established the first unique continuation results for (scalar and tensorial) wave equations, from the conformal boundaries of general dynamical aAdS spacetimes.
In particular, the Carleman estimates developed in \cite{ChatzikaleasShao, hol_shao:uc_ads, hol_shao:uc_ads_ns, mcgill_shao:psc_aads} form a key ingredient for proving the main results of this paper; see Section \ref{sec.intro_summary} for further discussions in this direction.

Finally, the recent work of McGill \cite{mcgill:loc_ads}, which characterized locally AdS spacetimes in terms of its holographic data, can be seen as a precursor to our results and as a special case of Theorem \ref{theo:adscft}.
More specifically, \cite{mcgill:loc_ads} showed that (assuming the GNCC) a vacuum FG-aAdS segment $( \mi{M}, g )$ is locally AdS if and only if both $\gb{0}$ is conformally flat and $\gb{n} = 0$.
The key step in the proof of this result is a more straightforward analogue of the process in this paper; in particular, \cite{mcgill:loc_ads} applies directly the unique continuation results of \cite{mcgill_shao:psc_aads} to the tensorial wave equations satisfied by the spacetime curvature on a single aAdS spacetime.

\subsection{The Generalised Null Convexity Criterion} \label{sec.intro_GNCC}

We now turn our attention toward the key geometric assumption required for Theorems \ref{theo:adscft}, \ref{theo:killing}, and \ref{theo:symmetries}---the GNCC of \cite{ChatzikaleasShao}.
First, we give a rough statement of the GNCC, in the special case of vacuum aAdS spacetimes treated here:

\begin{definition} \label{def:admissibleright}
Let $( \mi{M}, g )$ be a vacuum FG-aAdS segment, with conformal boundary $( \mi{I}, \gm )$, and consider an open subset $\mi{D} \subset \mi{I}$ with compact closure.
We say $( \mi{D}, \gm )$ satisfies the \emph{generalised null convexity criterion} (or \emph{GNCC}) iff there is a $C^4$-function $\eta$ on a neighbourhood of $\bar{\mi{D}}$ such that:
\begin{itemize}
\item $\eta > 0$ on $\mi{D}$, and $\eta = 0$ on the boundary $\partial \mi{D}$ of $\mi{D}$.

\item The following bilinear form is uniformly positive-definite on $\mi{D}$ along all $\mf{g}$-null directions,
\begin{equation}
\label{conr} \eta^{-1} \Dm^2_{ \gm } \eta + \mc{P} [ \gm ] \text{,}
\end{equation}
where $\Dm^2_{ \gm }$ and $\mc{P} [ \gm ]$ are the Hessian and Schouten tensor with respect to $\gm$, respectively.
\end{itemize}
\end{definition}

\begin{remark} \label{rmk.intro_GNCC_conformal}
One important feature of the GNCC is that it is conformally invariant.
In particular, \cite[Proposition 3.6]{ChatzikaleasShao} showed that if $( \mi{D}, \gm )$ satisfies the GNCC with $\eta$, then the conformally related $( \mi{D}, e^{ 2 \varphi } \gm )$ also satisfies the GNCC, with $\eta' := e^\varphi \eta$.
\end{remark}

\begin{remark}
Observe that $\mc{P} [ \gm ]$ can be replaced by $\frac{1}{n-2} \operatorname{Ric} [ \gm ]$ in \eqref{conr}, since their difference is proportional to $\gm$ and hence vanishes along all null directions.
\end{remark}

\begin{remark} \label{rem:coinv}
One can also show \cite[Proposition 3.4]{mcgill_shao:psc_aads} that $( \mi{D}, \gm )$ satisfies the GNCC if and only if there exists $\eta$ as in Definition \ref{def:admissibleright} and a smooth function $\zeta: \mi{I} \rightarrow \R$ such that the following bilinear form is uniformly positive-definite on $\mi{D}$ along \emph{all} directions:\footnote{This can be directly checked when $n=2$. For $n \geq 3$, this follows from the fact that two bilinear forms that do not vanish simultaneously (except at zero) can be simultaneously diagonalised \cite{Greub}.}
\begin{equation}
\label{refo} \eta^{-1} \Dm^2_{ \gm } \eta + \mc{P} [ \gm ] - \zeta \, \gm
\end{equation}
\end{remark}

See Definition \ref{def.carleman_gncc} or \cite{ChatzikaleasShao} for a more precise description of the GNCC.
Roughly, one can interpret the GNCC as stating that the domain $\mi{D}$ is ``large enough" with respect to the geometry of $( \mi{I}, \gm )$.
Its main significance, demonstrated in \cite{ChatzikaleasShao}, is that it precisely captures the conditions on the conformal boundary that lead to pseudoconvexity of the near-boundary geometry.
More specifically, it ensures the level hypersurfaces of $\smash{f := \frac{\rho}{\eta}}$ are pseudoconvex in a small region of $\mi{M}$ near $\mi{D}$.
This observation was a crucial ingredient in the Carleman estimates of \cite{ChatzikaleasShao}; see the discussions in Section \ref{sec.intro_summary}.

\subsubsection{Special Cases}

To further flesh out Definition \ref{def:admissibleright}, let us now consider the special case of the AdS conformal boundary $( \mi{I}_\textrm{AdS}, \gm_\textrm{AdS} )$ of \eqref{adsi}, which satisfies
\begin{equation}
\label{eq.intro_ads_P} \gm_\textrm{AdS} = - dt^2 + \mathring{\gamma}_{n-1} \text{,} \qquad \mc{P} [ \gm_\textrm{AdS} ] = \frac{1}{2} ( dt^2 + \mathring{\gamma}_{n-1} ) \text{.}
\end{equation}
In addition, we take $\mi{D} := \mi{D}_0$ to be the time slab
\begin{equation}
\label{eq.intro_ads_D} \mi{D}_0 := ( t_-, t_+ ) \times \Sph^{n-1} \text{,} \qquad \partial \mi{D}_0 = \{ t_-, t_+ \} \times \Sph^{n-1} \text{,} \qquad T_- < t_- < t_+ < T_+ \text{.}
\end{equation}

\begin{proposition}[\cite{ChatzikaleasShao}, Corollary 3.14] \label{thm.intro_GNCC_ads}
$( \mi{D}_0, \gm_\textrm{AdS} )$ satisfies the GNCC if and only if $t_+ - t_- > \pi$.
\end{proposition}

The key observation here is that if we assume $\eta$ to depend only on $t$, then \eqref{eq.intro_ads_P} yields
\begin{equation}
\label{eq.intro_ads_GNCC} ( \eta^{-1} \Dm^2_{ \gm } \eta + \mc{P} [ \gm ] ) ( \mf{Z}, \mf{Z} ) = \eta^{-1} ( \mf{Z} t )^2 \, ( \ddot{\eta} + \eta ) \text{.}
\end{equation}
Then, one can directly check that Definition \ref{def:admissibleright} is satisfied by the function
\begin{equation}
\label{eq.intro_ads_eta} \eta := \sin \left( \pi \cdot \frac{ t - t_- }{ t_+ - t_- } \right)
\end{equation}
whenever $t_+ - t_- > \pi$.\footnote{In particular, the condition $t_+ - t_- > \pi$ is required for the right-hand side of \eqref{eq.intro_ads_GNCC} to be positive.}
(Conversely, if $t_+ - t_- \leq \pi$, then a contradiction argument using Sturm comparison yields that the GNCC cannot hold for $( \mi{D}_{ 0}, \gm_\textrm{AdS} )$; see \cite[Lemma 3.7]{ChatzikaleasShao}.)

\begin{remark}
Note in particular that Proposition \ref{thm.intro_GNCC_ads} applies to every Kerr-AdS spacetime, since these all induce the AdS conformal boundary.
\end{remark}

\begin{remark}
The key consequence of Proposition \ref{thm.intro_GNCC_ads}---that unique continuation for wave equations holds from $\mi{D}_0$ when $t_+ - t_- > \pi$---was first proved as a special case of the results of \cite{hol_shao:uc_ads}.
\end{remark}

Next, we move to more general boundary domains that are foliated by a time function $t$,
\begin{equation}
\label{eq.intro_time} \mi{I}_\ast := ( T_-, T_+ ) \times \mc{S} \text{,} \qquad \mi{D}_\ast := ( t_-, t_+ ) \times \mc{S} \text{,}
\end{equation}
with $\mc{S}$ being a compact manifold of dimension $n - 1$.
Previous unique continuation results for linear wave equations were developed in this setting \eqref{eq.intro_time}, and these can also be viewed as special cases of the GNCC.
First, \cite{hol_shao:uc_ads} developed an analogue of the GNCC for static $\gm$.
This was extended to non-static $\gm$ in \cite{hol_shao:uc_ads_ns}, and then to a wider class of metrics $\gm$ and time foliations in \cite{mcgill_shao:psc_aads}.

Let us focus on the key criterion of \cite{mcgill_shao:psc_aads}, as well as its relation to the GNCC:

\begin{proposition}[\cite{ChatzikaleasShao}, Proposition 3.13] \label{thm.intro_NCC}
Assume the setting of \eqref{eq.intro_time}, and suppose there exist constants $0 \leq B < C$ such that the following holds for any $\mf{g}$-null vector field $\mf{Z}$:\footnote{Observe that the second condition of \eqref{eq.intro_NCC} can be viewed as a bound on the non-stationarity of $\gm$, since $\mf{D}_{ \gm }^2 t$ is proportional to the Lie derivative of $\gm$ along the gradient of $t$.}
\begin{equation}
\label{eq.intro_NCC} \mc{P} [ \gm ] ( \mf{Z}, \mf{Z} ) \geq C^2 \cdot ( \mf{Z} t )^2 \text{,} \qquad | \mf{D}_{ \gm }^2 t ( \mf{Z}, \mf{Z} ) | \leq 2 B \cdot ( \mf{Z} t )^2 \text{.}
\end{equation}
Then, \emph{$( \mi{D}_\ast, \gm )$ satisfies the GNCC} as long as $t_+ - t_-$ is large enough (depending on $B$ and $C$).
\end{proposition}

The proof of Proposition \ref{thm.intro_NCC} is similar to that of Proposition \ref{thm.intro_GNCC_ads}, except one now chooses $\eta$ (still depending only on $t$) to roughly solve a damped harmonic oscillator:
\begin{equation}
\label{odel} \ddot{\eta} - 2 b | \dot{\eta} | + c^2 \eta = 0 \text{,} \qquad B \leq b < c < C \text{.}
\end{equation}

\begin{remark}
The connection between damped harmonic oscillators \eqref{odel} and unique continuation from $\mi{D}_\ast$ was first illuminated in \cite{hol_shao:uc_ads_ns}; see the discussions therein.
\end{remark}

\begin{remark}
Proposition \ref{thm.intro_NCC} can be used to generalize the conclusions of Proposition \ref{thm.intro_GNCC_ads}:
\begin{itemize}
\item For instance, if $t_+ - t_- > \pi$, then Proposition \ref{thm.intro_NCC} implies that $( \mi{D}_0, \gm )$ satisfies the GNCC whenever $\gm$ is a sufficiently small perturbation of $\gm_\textrm{AdS}$.

\item If $\gm$ is static with respect to $t$, and if the cross-sections $\mc{S}$ have positive Ricci curvature, then $( \mi{D}_\ast, \gm )$ satisfies the GNCC for sufficiently large $t_+ - t_-$; see \cite[Proposition B.2]{hol_shao:uc_ads_ns}.
\end{itemize}
\end{remark}

The conditions \eqref{eq.intro_NCC} were first identified in \cite{mcgill_shao:psc_aads} and were named the \emph{null convexity criterion} (or \emph{NCC}).
Proposition \ref{thm.intro_NCC} shows that the GNCC indeed generalizes the NCC, both removing the need for a predetermined time function and allowing for a larger class of boundary domains $\mi{D}$.

One advantage of the NCC \eqref{eq.intro_NCC} is that it is easier to check than the rather abstract GNCC.
On the other hand, one shortcoming of \eqref{eq.intro_NCC} is that it fails to be conformally invariant, as a conformal transformation of $\gm$ can cause \eqref{eq.intro_NCC} to no longer hold.
This makes the NCC undesirable for the main results of this paper and provides a key motivation for developing the GNCC.

\subsubsection{Geodesic Return}

The necessity of \emph{some} geometric condition in Theorem \ref{theo:adscft} was already conjectured in \cite{hol_shao:uc_ads, hol_shao:uc_ads_ns}, due to the special properties of AdS geometry near its conformal boundary.

On AdS spacetime, there exist null geodesics which propagate arbitrarily close to the conformal boundary $\mi{I}_\textrm{AdS}$, but only intersect $\mi{I}_\textrm{AdS}$ at two points that are time $\pi$ apart; see \cite[Section 1.2]{hol_shao:uc_ads}.\footnote{In terms of the standard embedding of AdS spacetime into the Einstein cylinder $\R \times \Sph^n_+$, these geodesics move forward in time and along great circles in the spatial component, both with constant speed.}
One can then construct, via the geometric optics methods of Alinhac and Baouendi \cite{alin_baou:non_unique}, solutions to linear wave equations that are concentrated along such a family of geodesics.\footnote{One important caveat is that the methods of \cite{alin_baou:non_unique} only apply directly to wave operators $\Box_g + \sigma$ with the conformal mass $4 \sigma := n^2 - 1$. The case of general $\sigma$ will be treated in the upcoming work of Guisset \cite{guiss:non_unique}.}
These solutions yield, for AdS spacetime, counterexamples to unique continuation for various linear wave equations when the data on the conformal boundary is imposed on a timespan of less than $\pi$ (the return time of these null geodesics), between the start and end times of the geodesics.

\begin{remark}
We note that not every wave equation can have such counterexamples to unique continuation.
By Holmgren's theorem \cite{hor:lpdo2}, if all the coefficients of the wave equation (including the principal part $g$) are real-analytic, then the above counterexamples cannot exist.
\end{remark}

One can in fact view the GNCC as a generalization of the above intuitions for AdS spacetime to aAdS settings.
This observation was given in \cite[Theorem 4.1]{ChatzikaleasShao}, which connected the GNCC to the trajectories of null geodesics near the conformal boundary.
In particular, given a spacetime null geodesic $\Lambda$ that is sufficiently close to the conformal boundary and that travels over $\mi{D}$ satisfying the GNCC, \cite[Theorem 4.1]{ChatzikaleasShao} established that $\Lambda$ must intersect the conformal boundary within $\mi{D}$ (in either the future or past direction).
In other words, there cannot exist any near-boundary null geodesics that travel over $\mi{D}$ but do not terminate at $\mi{D}$ itself.
From this, one concludes that the Alinhac-Baouendi counterexamples of \cite{alin_baou:non_unique} in AdS cannot be constructed over $\mi{D}$.

\begin{remark}
In addition, \cite[Theorem 4.5]{mcgill_shao:psc_aads} established an analogue of \cite[Theorem 4.1]{ChatzikaleasShao} for the NCC \eqref{eq.intro_NCC}.
This was the first result in general aAdS settings that directly connected criteria for unique continuation to near-boundary null geodesic trajectories.
\end{remark}

Consequently, the above discussions give us two justifications for the GNCC being the crucial condition for unique continuation of wave equations from the conformal boundary:
\begin{itemize}
\item The GNCC rules out the known counterexamples to unique continuation for waves.

\item The GNCC implies pseudoconvexity, allowing for unique continuation results to be proved.
\end{itemize}

\begin{remark}
Though the GNCC is crucial to our proof of Theorem \ref{theo:adscft}, it is not known whether the methods of \cite{alin_baou:non_unique} extend to the Einstein-vacuum equations.
The construction of counterexamples to unique continuation in Theorem \ref{theo:adscft} when the GNCC is violated is a challenging problem.
\end{remark}

Finally, we note that this connection between the GNCC and null geodesics can be used to show that \emph{no subdomain $\mi{D}$ of the planar AdS or toric AdS conformal boundaries can satisfy the GNCC}; see \cite[Corollary 3.10]{ChatzikaleasShao}.\footnote{These are analogues of AdS, but with the spheres $\Sph^{n-1}$ are replaced with $\R^{n-1}$ or the flat torus $\mathbb{T}^{n-1}$.}
In particular, on both planar and toric AdS spacetimes, there exist null geodesics that remain arbitrarily close to but never intersect the conformal boundary for all times.

\subsection{Proof Overview of Theorem \ref{theo:adscft}} \label{sec.intro_summary}

In this subsection, we provide an outline of the proof of Theorem \ref{theo:adscft}, our key result.
First, via an appropriate gauge transformation, we can assume
\begin{equation}
\label{eq.intro_proof_ass} ( \gb{0}, \gb{n} ) = ( \gbc{0}, \gbc{n} )
\end{equation}
on $\mi{D}$, without any loss of generality; for details of this process, see Section \ref{sec.app_gauge}.

Furthermore, since we are only concerned with the near-boundary region, we can assume (see Remark \ref{rmk.intro_adscft_mfld}) that $\smash{ \check{\mi{M}} = \mi{M} }$, so that the two aAdS metrics $g$ and $\check{g}$ take the forms.
\begin{equation}
\label{eq.intro_proof_g} g = \rho^{-2} [ d \rho^2 + \gv ( \rho ) ] \text{,} \qquad \check{g} = \rho^{-2} [ d \rho^2 + \check{\gv} ( \rho ) ] \text{.}
\end{equation}
In light of \eqref{eq.intro_proof_g}, it suffices to show that
\begin{equation}
\label{eq.intro_proof_goal} \gv - \check{\gv} = 0 \text{.}
\end{equation}
Below, we discuss each of the three key components of the proof of Theorem \ref{theo:adscft}.

\subsubsection{The Vertical Tensor Calculus}

The main objects of analysis in the proof of Theorem \ref{theo:adscft} are so-called \emph{vertical tensor fields}.
These can be thought of as tensor fields on $\mi{M}$ that are everywhere tangent to the level sets of $\rho$; an equivalent way to view vertical tensor fields is as $\rho$-parametrized families of tensor fields on $\mi{I}$.
See Section \ref{sec.aads_fg} for a more detailed development.

The simplest examples of vertical tensor fields are the \emph{vertical metrics} $\gv$ and $\check{\gv}$.
As these define Lorentzian metrics on each level set of $\rho$, one can also define corresponding vertical connections $\Dv$ and $\check{\Dv}$ on $\mi{M}$, respectively.
Other vertical tensor fields are obtained by appropriate decompositions of spacetime quantities, such as the Weyl curvature $W$ associated with $g$:\footnote{See Section \ref{sec.aads_fg} for precise coordinate conventions. Roughly, Latin letters $a, b, \dots$ denote vertical components.}
\[
\wv^0_{ a b c d } := \rho^2 \, W_{ a b c d } \text{,} \qquad \wv^1_{ a b c } := \rho^2 \, W_{ \rho a b c } \text{,} \qquad \wv^2_{ a b } := \rho^2 \, W_{ \rho a \rho b } \text{.}
\]

One reason for formulating our main quantities as vertical tensor fields is that these, when viewed as $\rho$-parametrised tensor fields on $\mi{I}$, have a natural notion of limits at the conformal boundary---as $\rho \rightarrow 0$.
(For instance, the boundary limit of $\gv$ is the boundary metric $\gm = \gb{0}$.)
This allows one to easily connect quantities in the bulk spacetime with those on the conformal boundary.

Analogues of vertical tensor fields have been widely used in mathematical relativity,\footnote{These are usually formulated as \emph{horizontal tensors} that are everywhere tangent to a foliation of spacelike submanifolds. Common examples include the connection and curvature components in a double null foliation.} but here we also extend these ideas beyond the standard uses.
In particular, since tensorial wave equations play a key role in the proof of Theorem \ref{theo:adscft}, we want to make sense of a spacetime wave operator $\Boxm$ applied to vertical tensor fields.
Furthermore, we aim to do this in a covariant manner, so that the usual operations of geometric analysis---such as Leibniz rules and integrations by parts---continue to hold.
As a result of this, we can present the analysis of vertical tensor fields in almost the exact same manner as corresponding analyses of scalar fields.

The difficulty in defining $\Boxm$ covariantly lies in making proper sense of \emph{second, spacetime} derivatives of vertical tensor fields.\footnote{First spacetime derivatives can be straightforwardly defined by projecting spacetime covariant derivatives.}
To get around this, we extend our calculus to \emph{mixed tensor fields}---those that contain both spacetime and vertical components.
This allows us to make sense of the spacetime Hessian as adding spacetime components to a mixed field; see Section \ref{sec.aads_mixed} for details.

Mixed tensor fields and extended wave operators $\Boxm$ originated from \cite{shao:ksp} and have been applied in aAdS contexts in \cite{ChatzikaleasShao, hol_shao:uc_ads, hol_shao:uc_ads_ns, mcgill_shao:psc_aads}.\footnote{Similar notions were independently developed and used in \cite{keir:weak_null}.}
The full vertical (and mixed) tensor calculus, in the form shown in this paper, was first constructed in \cite{mcgill_shao:psc_aads, shao:aads_fg} and was also adopted in \cite{ChatzikaleasShao}.

\begin{remark}
An alternative approach is to decompose our quantities into scalar fields and derive an analogue of the wave-transport system used in \cite{alex_io_kl:hawking_anal}.
One disadvantage is that the unknowns are only locally defined, while we have to work with all of $\mi{D}$ simultaneously.
In contrast, the vertical formalism allows us to present our arguments in a geometric and frame-independent manner.
\end{remark}

\subsubsection{The Wave-Transport System}

The strategy for obtaining \eqref{eq.intro_proof_goal} is to formulate $\gv - \check{\gv}$ as an unknown in a closed system of (vertical) tensorial transport and wave equations, and to then apply the requisite unique continuation results to this system.

From the Gauss-Codazzi equations on level sets of $\rho$ and from \eqref{evec}, one derives
\begin{equation}
\label{basiceq} \partial_\rho \gv_{ab} =: \Lv_{ab} \text{,} \qquad \partial_\rho \Lv_{ab} = -2 \wv^2_{ab} + \rho^{-1} \Lv_{ab} + \frac{1}{2} \gv^{cd} \Lv_{ad} \Lv_{bc} \text{,} \qquad \Dv_b \Lv_{ a c } - \Dv_a \Lv_{ b c } = 2 \wv^1_{ c a b } \text{;}
\end{equation}
analogous formulas also hold with respect to $\check{\gv}$.
We wish to couple the transport equations \eqref{basiceq} to wave equations satisfied by the Weyl curvature $W$ (see Proposition \ref{thm.aads_weyl_pre} for precise formulas):
\[
\Box_g W + 2n W = W \cdot W \text{.}
\]
Decomposing $W$ into vertical components as before, we derive, for $\ms{U} \in \{ \wv^\star, \wv^1, \wv^2 \}$,
\begin{equation}
\label{wavew} \Boxm \ms{U} + c_\ms{U} \ms{U} = \ms{NL} ( \gv, \Lv, \Dv \Lv, \wv^\star, \wv^1, \wv^2, \Dv \wv^\star, \Dv \wv^1, \Dv \wv^2 ) \text{,}
\end{equation}
where $c_\ms{U} \in \Z$ depends on the component $\ms{U}$ considered.
Moreover, $\ms{NL} (\cdot)$ represents terms involving (contractions of) the listed quantities that decay sufficiently quickly toward the conformal boundary, while $\wv^\star$ in \eqref{wavew} is a renormalization of $\wv^0$; see \eqref{eq.aads_wstar} and \eqref{eq.aads_wave} for precise formulas.

\begin{remark}
That different masses $c_{\ms{U}}$ appear in \eqref{wavew}, at least when $n > 3$, is because $\wv^\star$, $\wv^1$, and $\wv^2$ have different asymptotics (in powers of $\rho$) at the conformal boundary.\footnote{The case $n = 3$ is an exception, as $c_{ \wv^1 } = c_{ \wv^2 } = 2$, and $\wv^\star$ is fully determined by $\wv^1$ and $\wv^2$.}
One consequence of this is that we must treat the components $\wv^\star$, $\wv^1$, $\wv^2$ separately in our analysis.
\end{remark}

Subtracting \eqref{basiceq}-\eqref{wavew} from their counterparts for $\check{\gv}$ yields a closed wave-transport system for the quantities $\gv - \check{\gv}$,  $\Lv - \check{\Lv}$, $\wv^\star - \check{\wv}^\star$, $\wv^1 - \check{\wv}^1$, and $\wv^2 - \check{\wv}^2$.
However this system \emph{fails to close for the purpose of applying our Carleman estimates}.
In particular, the wave equation will only allow us to control up to one derivative of $\wv^\star - \check{\wv}^\star$, $\wv^1 - \check{\wv}^1$, and $\wv^2 - \check{\wv}^2$, which, in turn, allows us to control only one derivative of $\Lv - \check{\Lv}$ and $\gv - \check{\gv}$.
On the other hand, when we take a difference of the wave equations \eqref{wavew}, we obtain a term involving the difference of $\Boxm$ and $\smash{\check{\Boxm}}$, which contains second derivatives of $\gv - \check{\gv}$ (since $\ms{U}$ is a tensor) that we a priori cannot handle.\footnote{In principle, one may try to obtain estimates for two vertical derivatives of the metric from the structure equation involving $\ms{w}^0$ (see the second equation in \eqref{eq.aads_connection}) and the vertical Riemann curvature. However, this introduces other difficulties related to finding an appropriate gauge on the vertical slices.}

The resolution, inspired by the symmetry extension result \cite{ionescusk} of Ionescu and Klainerman, is to apply a careful renormalization of the system that eliminates the troublesome quantities.
(See also Section \ref{sec.intro_ik} below, where we compare our wave-transport system with that of \cite{ionescusk}.)

The first crucial observation is that while $\Dv^2 ( \gv - \check{\gv} )$ is off limits, we can obtain improved control if one derivative is a curl.
In particular, the first and third equations in \eqref{basiceq} yield, roughly,
\begin{align*}
\partial_\rho [ \Dv_{db} ( \gv - \check{\gv} )_{ac} - \Dv_{da} ( \gv - \check{\gv} )_{bc} ] &\sim \Dv_d [ ( \Dv_b \Lv_{ a c } - \Dv_a \Lv_{ b c } ) - ( \check{\Dv}_b \check{\Lv}_{ a c } - \check{\Dv}_a \check{\Lv}_{ b c } ) ] \\
&\sim \Dv_d ( \wv^1 - \check{\wv}^1 )_{cab} \text{.}
\end{align*}
The above still does not quite suffice, and we need one more renormalization---this is due to terms involving $\Dv ( \check{\Dv} - \Dv )$, which again contain the undesirable $\Dv^2 ( \gv - \check{\gv} )$.
All this leads us to define the auxiliary quantities (see \eqref{Qdef} and \eqref{Bdef} for precise formulas)
\begin{align}
\label{eq.intro_proof_B} \ms{B}_{cab} &:= \Dv_c ( \gv - \check{\gv} )_{ab} - \Dv_a ( \gv - \check{\gv} )_{cb} - \Dv_b \ms{Q}_{ca} \text{,} \\
\notag \partial_\rho \ms{Q}_{ca} &:= \gv^{de} \Lv_{ce} ( \gv - \check{\gv} )_{ad} - \gv^{de} \Lv_{ae} ( \gv - \check{\gv} )_{cd} \text{,}
\end{align}
with $\ms{Q} \rightarrow 0$ as $\rho \rightarrow 0$.
We then show that $\Dv \ms{B}$ can indeed be adequately controlled by $\Dv ( \wv^1 - \check{\wv}^1 )$.

The second crucial observation comes from a detailed examination of the difference $\smash{ \Boxm - \check{\Boxm} }$.
To appreciate this, we consider the wave equation for just $\wv^2 - \check{\wv}^2$ for concreteness:
\begin{equation}
\label{eq.intro_proof_boxdiff1} \Boxm ( \wv^2 - \check{\wv}^2 ) = \Boxm \wv^2 - \check{\Boxm} \check{\wv}^2 - ( \Boxm - \check{\Boxm} ) \check{\wv}^2 \text{.}
\end{equation}
The dangerous terms arise from the following (rather long) computation,
\begin{equation}
\label{eq.intro_proof_boxdiff2} ( \Boxm - \check{\Boxm} ) \check{\wv}^2_{ab} = - \frac{1}{2} \rho^2 \gv^{cd} \gv^{ef} \check{\wv}^2_{eb} \Dv_c \ms{B}_{afd} - \frac{1}{2} \gv^{ef} \check{\wv}^2_{eb} \Boxm ( \gv - \check{\gv} + \ms{Q} )_{af} + \{ a \leftrightarrow b \} + \ms{Err} \text{,}
\end{equation}
where $\{ a \leftrightarrow b \}$ denotes the preceding terms repeated but with $a$ and $b$ interchanged, and where $\ms{Err}$ consists of (many) terms containing only difference quantities that we can control.

The key point is that the only instances of $\Dv^2 ( \gv - \check{\gv} )$ appear either as $\Dv \ms{B}$, which we can control, or as $\Boxm$ applied to difference quantities.
This leads us to the renormalized curvature difference
\begin{equation}
\label{eq.intro_proof_wdiff} \ms{W}^2_{ab} := \wv^2_{ab} - \check{\wv}^2_{ab} + \frac{1}{2} \gv^{de} \check{\wv}^2_{ad} ( \gv - \check{\gv} + \ms{Q} )_{be} + \frac{1}{2} \gv^{de} \check{\wv}^2_{db} ( \gv - \check{\gv} + \ms{Q} )_{ae} \text{,}
\end{equation}
which in essence shifts the $\Boxm$-terms from the right-hand side of \eqref{eq.intro_proof_boxdiff2} into the left; one can also define the remaining $\ms{W}^1$ and $\ms{W}^\star$ similarly.
In light of \eqref{eq.intro_proof_boxdiff1} and \eqref{eq.intro_proof_boxdiff2}, we obtain that $\ms{W}^\star$, $\ms{W}^1$, $\ms{W}^2$ satisfy wave equations that do not contain $\Dv^2 ( \gv - \check{\gv} )$ as sources.

Finally, the renormalized wave-transport system is obtained by treating the quantities
\begin{equation}
\label{eq.intro_proof_unknown} \gv - \check{\gv} \text{,} \qquad \ms{Q} \text{,} \qquad \Lv - \check{\Lv} \text{,} \qquad \ms{B} \text{,} \qquad \ms{W}^\star \text{,} \qquad \ms{W}^1 \text{,} \qquad \ms{W}^2
\end{equation}
as unknowns.
In particular, from the above discussions, and from various asymptotic properties of geometric quantities, we arrive at the (schematic) transport equations
\begin{align}
\label{eq.intro_proof_transport} \partial_\rho ( \gv - \check{\gv} ) &= \Lv - \check{\Lv} \text{,} \\
\notag \partial_\rho \ms{Q} &= \mc{O} (\rho) \, ( \gv - \check{\gv}, \ms{Q} ) \text{,} \\
\notag \partial_\rho \ms{B} &= 2 ( \wv^1 - \check{\wv}^1 ) + \mc{O} (\rho) \, ( \gv - \check{\gv}, \ms{Q}, \ms{B} ) + \mc{O} (1) \, ( \Lv - \check{\Lv} ) \text{,} \\
\notag \partial_\rho [ \rho^{-1} ( \Lv - \check{\Lv} ) ] &= - 2 \rho^{-1} \ms{W}^2 + \mc{O} (1) \, ( \gv - \check{\gv}, \Lv - \check{\Lv}, \ms{Q} ) \text{,}
\end{align}
coupled to the following (schematic) wave equations for any $\ms{W} \in \{ \ms{W}^\star, \ms{W}^1, \ms{W}^2 \}$:
\begin{align}
\label{eq.intro_proof_wave} \Boxm \ms{W} + c_\ms{W} \ms{W} &= \sum_{ \ms{V} \in \{ {\ms{W}^\star}, \ms{W}^1, \ms{W}^2 \} } \left[ \mc{O} (\rho^2) \ms{V} + \mc{O} (\rho^3) \Dv \ms{V}  \right] + \mc{O} (\rho) \, ( \Lv - \check{\Lv} ) \\
\notag &\qquad + \mc{O} (\rho^2) \, ( \gv - \check{\gv}, \ms{Q}, \Dv ( \gv - \check{\gv} ), \Dv \ms{Q}, \Dv ( \Lv - \check{\Lv} ), \Dv \ms{B} ) \text{,}
\end{align}
The $\mc{O} ( \cdot )$'s in \eqref{eq.intro_proof_transport}--\eqref{eq.intro_proof_wave} indicate the asymptotics of various coefficients as $\rho \rightarrow 0$.

For more precise formulas, see Propositions \ref{thm.sys_transport} and \ref{thm.sys_wave}.
In particular, the wave-transport system \eqref{eq.intro_proof_transport}--\eqref{eq.intro_proof_wave} indeed closes from the point of view of derivatives.\footnote{The system \eqref{eq.intro_proof_transport}--\eqref{eq.intro_proof_wave} could also be used to derive general unique continuation results for the Einstein equations near general timelike hypersurfaces, providing an alternate approach to that of \cite{alex:uc_vacuum}.}

\subsubsection{The Carleman Estimate}

The technical workhorse in the proof of Theorem \ref{theo:adscft}, connecting the system \eqref{eq.intro_proof_transport}--\eqref{eq.intro_proof_wave} with unique continuation, is a Carleman estimate for wave equations that are satisfied by vertical tensor fields near aAdS conformal boundaries.

The role of Carleman estimates in unique continuation theory has an extensive history, tracing back to the seminal \cite{cald:unique_cauchy, carl:uc_strong} for elliptic problems.
Classical results for wave equations---see \cite{hor:lpdo4, ler_robb:unique}---highlight \emph{pseudoconvexity} as the crucial condition needed for Carleman estimates, and hence unique continuation results, to hold across a given hypersurface.
The novelty in aAdS settings is that the conformal boundary is \emph{zero-pseudoconvex}, so the classical results no longer apply.\footnote{In particular, the null geodesics with respect to $\rho^2 g$ asymptote toward being tangent to the boundary $\rho = 0$. As a result, the conformal boundary just barely fails to be pseudoconvex.}

These difficulties were overcome in a series of results \cite{ChatzikaleasShao, hol_shao:uc_ads, hol_shao:uc_ads_ns, mcgill_shao:psc_aads} by the authors, Chatzikaleas, and McGill, leading to Carleman estimates and unique continuation results on FG-aAdS segments, under the assumption of the GNCC.\footnote{The machinery for deriving Carleman estimates in zero-pseudoconvex settings originated from works of the second author with Alexakis and Schlue \cite{alex_schl_shao:uc_inf, alex_shao:uc_global}. See also \cite{peters:cpt_cauchy}, which independently studied zero-pseudoconvex settings.}
As mentioned before, the GNCC ensures the existence of a foliation of pseudoconvex hypersurfaces near the conformal boundary.
(See the above references for further discussions of the ideas leading to the Carleman estimates.)

We now give a rough statement of the wave Carleman estimate used in this article:

\begin{theorem}[Carleman estimate for wave equations, \cite{ChatzikaleasShao}] \label{prop:keycarleman}
Let $( \mi{M}, g )$ be a vacuum FG-aAdS segment, and suppose its conformal infinity $( \mi{I}, \gm )$ has a subdomain $\mi{D} \subset \mi{I}$ such that $( \mi{D}, \gm )$ satisfies the GNCC.
Also, fix $\sigma \in \R$, set $\smash{ f := \frac{\rho}{\eta} }$ (with $\eta$ as in Definition \ref{def:admissibleright}), and define the region
\[
\Omega_{ f_\star } := \{ f < f_\star \} \subset \mi{M} \text{,} \qquad f_\star > 0 \text{.}
\]
Then, the following holds for any vertical tensor field $\ms{\Phi}$ on $\mi{M}$ with $\ms{\Phi}$, $\nabla \ms{\Phi}$ vanishing on $f = f_\star$,
\begin{align}
\label{eq.intro_carleman} &\int_{ \Omega_{ f_\star } } e^{ -\frac{\lambda f^p}{p} } f^{ n-2-p-2\kappa } | ( \Boxm + \sigma ) \ms{\Phi} |^2 \, dg \\
\notag &\qquad + \lambda^3 \limsup_{ \rho_\star \searrow 0 } \int_{ \{ \rho = \rho_\star \} } ( | \partial_\rho ( \rho^{-\kappa} \ms{\Phi} ) |^2 + | \Dv ( \rho^{-\kappa} \ms{\Phi} ) |^2 + | \rho^{-\kappa-1} \ms{\Phi} |^2 ) \, d \gv \\
\notag &\quad \gtrsim \lambda \int_{ \Omega_{ f_\star } } e^{ -\frac{\lambda f^p}{p} } f^{ n-2-2\kappa } ( f \rho^3 | \Dv \ms{\Phi} |^2 + f^{2p} | \ms{\Phi} |^2 ) \, dg \text{,}
\end{align}
provided $\kappa$ and $\lambda$ are sufficiently large, $f_\star$ is sufficiently small, and $0 < p < \frac{1}{2}$.
\end{theorem}

\begin{remark}
The norm $| \cdot |$ can be defined relative to a given Riemannian metric on the space of vertical tensors.
Moreover, the admissible values of $\kappa$, $\lambda$, $f_\star$ depend on $\gv$, $\mi{D}$, $\sigma$, and the rank of $\Phi$.
\end{remark}

See Theorem \ref{thm.carleman} below for a precise statement of the Carleman estimate.
The region $\Omega_{ f_\star }$, on which the Carleman estimate holds, is illustrated in Figure \ref{regionsfig}.

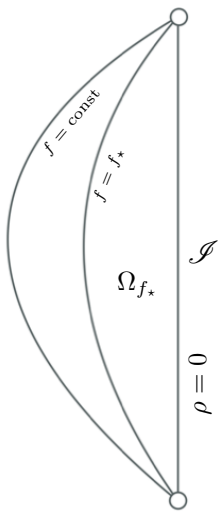
\begin{figure}
\includegraphics[width=3.2cm]{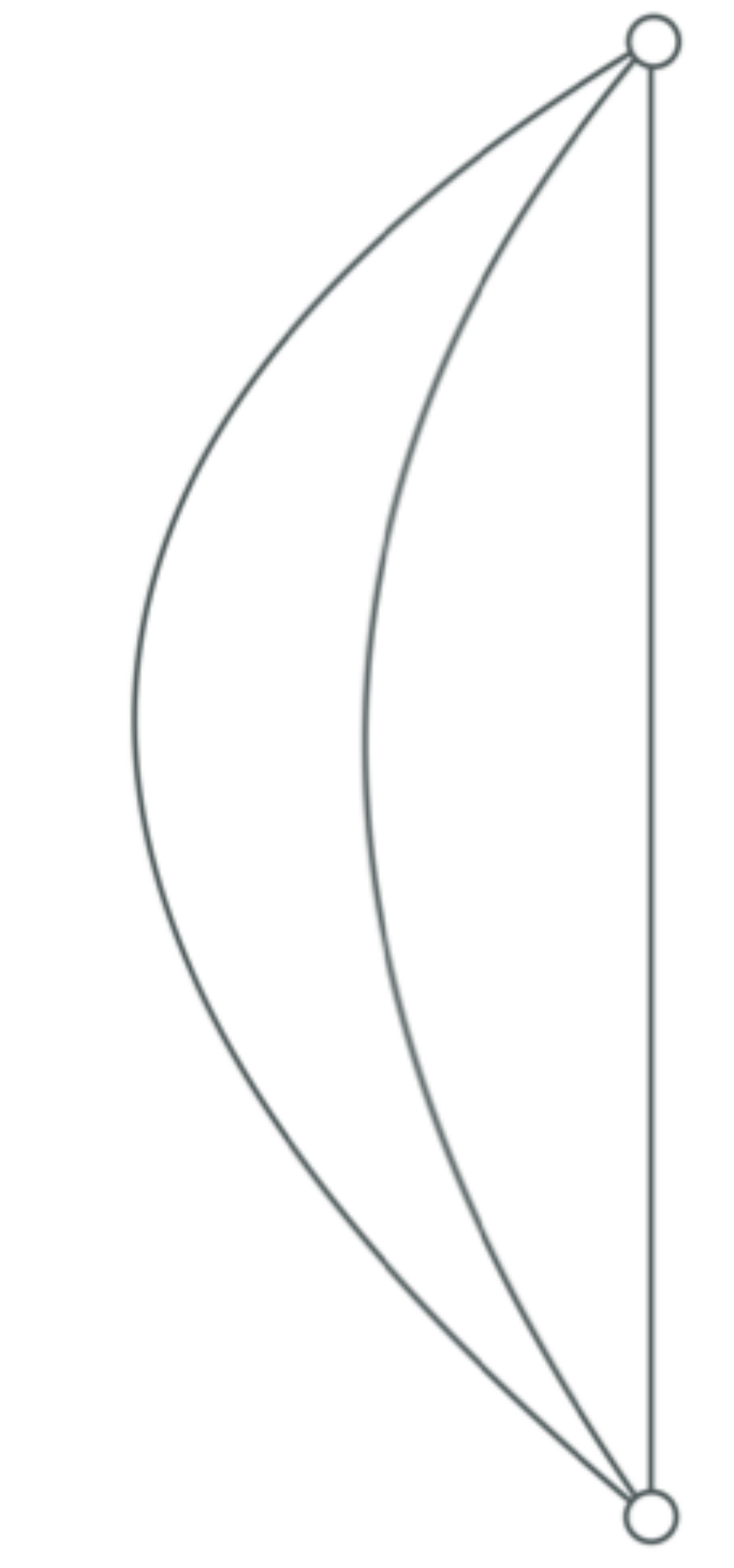}
\setlength{\unitlength}{1657sp}%
\begingroup\makeatletter\ifx\SetFigFont\undefined%
\gdef\SetFigFont#1#2#3#4#5{%
  \reset@font\fontsize{#1}{#2pt}%
  \fontfamily{#3}\fontseries{#4}\fontshape{#5}%
  \selectfont}%
\fi\endgroup%
\begin{picture}(2502,6478)(4339,-7800)
\put(1871,-2400){\rotatebox{50.0}{\makebox(0,0)[lb]{\smash{{\SetFigFont{6}{7.2}{\rmdefault}{\mddefault}{\updefault}{\color[rgb]{0,0,0}$f=\textrm{const}$}%
}}}}}
\put(3950,-4000){\rotatebox{0.0}{\makebox(0,0)[lb]{\smash{{\SetFigFont{12}{7.2}{\rmdefault}{\mddefault}{\updefault}{\color[rgb]{0,0,0}$\mi{I}$}%
}}}}}
\put(2670,-3100){\rotatebox{68.0}{\makebox(0,0)[lb]{\smash{{\SetFigFont{6}{7.2}{\rmdefault}{\mddefault}{\updefault}{\color[rgb]{0,0,0}$f=f_\star$}%
}}}}}
\put(2900,-4400){\makebox(0,0)[lb]{\smash{{\SetFigFont{10}{12.0}{\rmdefault}{\mddefault}{\updefault}{\color[rgb]{0,0,0}$\Omega_{f_\star}$}%
}}}}
\put(4231,-6271){\rotatebox{90.0}{\makebox(0,0)[lb]{\smash{{\SetFigFont{10}{12.0}{\rmdefault}{\mddefault}{\updefault}{\color[rgb]{0,0,0}$\rho=0$}%
}}}}}
\end{picture}%

\caption{The slices of constant $f$ and the region $\Omega_{f_\star}$. $\phantom{XX}$}
\label{regionsfig}
\end{figure}

The zero-pseudoconvexity of the conformal boundary leads to several complications in both the statement and the proof of Theorem \ref{prop:keycarleman}.
For instance, one consequence of this is that in contrast to classical results, which apply in small neighbourhoods of a single point, the estimate \eqref{eq.intro_carleman} holds only near sufficiently large domains $\mi{D}$ in the conformal boundary that satisfy the GNCC.
This is a feature that is exclusive to zero-pseudoconvex settings.

A second complication, arising from the degeneration of the pseudoconvexity of the level sets of $f$ toward the conformal boundary, is the presence of decaying weights in \eqref{eq.intro_carleman}---$f \rho^3$ and $f^{2p}$ in the right-hand side.
This makes absorption arguments in the proof of Theorem \ref{prop:keycarleman} far more delicate, and it restricts the class of wave equations for which one can prove unique continuation---namely, to equations with similarly decaying lower-order coefficients.

\begin{remark}
There do exist stronger unique continuation results, for which pseudoconvexity is not needed at all; see \cite{io_kl:unique_ip, ler:unique_tchar, robb_zuil:uc_interp, tata:uc_interp}.
However, these results require either additional symmetries or partial analyticity, neither of which is available in our setting.
\end{remark}

Theorem \ref{prop:keycarleman} is supplemented by a Carleman estimate---with the same region $\Omega_{ f_\ast }$ and Carleman weight---for transport equations.
In contrast to Theorem \ref{prop:keycarleman}, the proof of the transport Carleman estimate is straightforward; see Proposition \ref{thm.carleman_transport} below for the precise statement and for details.

\subsubsection{Unique Continuation}

The last step is to apply the wave and transport Carleman estimates to our system \eqref{eq.intro_proof_transport}--\eqref{eq.intro_proof_wave} to derive unique continuation---in particular \eqref{eq.intro_proof_goal}.
First, we claim \emph{all the unknowns \eqref{eq.intro_proof_unknown} vanish to arbitrarily high order at the conformal boundary}, so there are no boundary terms present in the Carleman estimates.
This follows from two key observations:
\begin{itemize}
\item From \eqref{eq.intro_proof_ass}, one can also derive that $\gb{k} = \gbc{k}$ for all $k < n$, and that $\gb{\star} = \gbc{\star}$; see \cite{shao:aads_fg} for details.
This leads to high-order vanishing for all the unknowns \eqref{eq.intro_proof_unknown}.

\item Further orders of vanishing can then be derived from transport and Bianchi equations.\footnote{This can be viewed as coefficients of the two FG expansions matching beyond $\gb{n}$ and $\gbc{n}$.}
\end{itemize}
See Section \ref{sec:addvanish} for further details on these steps.

From here, the process is mostly standard.
We apply the wave Carleman estimate \eqref{eq.intro_carleman} to $\chi \ms{W}^\star$, $\chi \ms{W}^1$, and $\chi \ms{W}^2$---for an appropriate cutoff $\chi := \chi (f)$---and recall \eqref{eq.intro_proof_wave} in order to obtain
\begin{align}
\label{eq.intro_proof_1} &\lambda \int_{ \Omega_i } w_\lambda (f) \sum_{ \ms{V} \in \{ \ms{W}^\star, \ms{W}^1, \ms{W}^2 \} } ( \rho^{2p} | \ms{V} |^2 + \rho^4 | \ms{D} \ms{V} |^2 ) \, d g \\
\notag &\quad \lesssim \int_{ \Omega_i } w_\lambda (f) \sum_{ \ms{U} \in \{ \gv - \check{\gv}, \ms{Q}, \Lv - \check{\Lv}, \ms{B} \} } ( \rho^{ 2 - p } | \ms{U} |^2 + \rho^{ 4 - p } | \Dv \ms{U} |^2 ) \, d g \\
\notag &\quad\qquad + \int_{ \Omega_i } w_\lambda (f) \sum_{ \ms{V} \in \{ \ms{W}^\star, \ms{W}^1, \ms{W}^2 \} } ( \rho^{ 4 - p } | \ms{V} |^2 + \rho^{ 6 - p } | \Dv \ms{V} |^2 ) \, d g + \int_{ \Omega_e } w_\lambda (f) \, ( \dots ) \, d g \text{.}
\end{align}
where $\smash{ w_\lambda (f) := e^{ -\lambda p^{-1} f^p } f^{n-2-2\kappa} }$ is the Carleman weight, and where\footnote{See \eqref{eq.proof_chi}--\eqref{eq.proof_cutoff} for the definition of $\chi$ and its relation to $f_i$, $f_e$; in particular, $\chi$ is non-constant precisely on $\Omega_e$.}
\[
\Omega_i := \{ 0 < f_i \} \text{,} \qquad \Omega_e := \{ f_i \leq f < f_e \} \text{,} \qquad 0 < f_i < f_e < f_\star \text{.}
\]
(The ``$\dots$" in the $\Omega_e$-integral depends on the unknowns \eqref{eq.intro_proof_unknown}, various weights in $\rho$ and $f$, and the cutoff $\chi$; however, its precise contents are irrelevant, as we only require that this integral is finite.)
Similar (but easier) applications of the transport Carleman estimate and \eqref{eq.intro_proof_transport} yield
\begin{align}
\label{eq.intro_proof_2} &\lambda \int_{ \Omega_i } w_\lambda (f) \sum_{ \ms{U} \in \{ \gv - \check{\gv}, \ms{Q}, \Lv - \check{\Lv}, \ms{B} \} } ( | \ms{U} |^2 + \rho^3 | \Dv \ms{U} |^2 ) \, d g \\
\notag &\quad \lesssim \int_{ \Omega_i } w_\lambda (f) \sum_{ \ms{U} \in \{ \gv - \check{\gv}, \ms{Q}, \Lv - \check{\Lv}, \ms{B} \} } ( \rho^{ 2 - p } | \ms{U} |^2 + \rho^{ 5 - p } | \Dv \ms{U} |^2 ) \, d g \\
\notag &\quad \qquad + \int_{ \Omega_i } w_\lambda (f) \sum_{ \ms{V} \in \{ \ms{W}^\star, \ms{W}^1, \ms{W}^2 \} } ( \rho^{ 2 - p } | \ms{V} |^2 + \rho^{ 5 - p } | \Dv \ms{V} |^2 ) \, d g \text{.}
\end{align}

The key point here is that the $\rho$-weights on the right-hand sides of \eqref{eq.intro_proof_1} and \eqref{eq.intro_proof_2} come from the $\mc{O} ( \cdot )$-coefficients in \eqref{eq.intro_proof_transport}--\eqref{eq.intro_proof_wave}.
The final crucial feature of our system is these $\rho$-weights are strong enough that, after summing \eqref{eq.intro_proof_1} and \eqref{eq.intro_proof_2}, \emph{the $\Omega_i$-integrals on the right-hand side can be absorbed into the left-hand side} (once $\lambda$ is sufficiently large).
From the above, we conclude that
\[
\lambda \int_{ \Omega_i } w_\lambda (f) \sum_{ \ms{V} \in \{ \ms{W}^\star, \ms{W}^1, \ms{W}^2 \} } \rho^{2p} | \ms{V} |^2 \, d g + \lambda \int_{ \Omega_i } w_\lambda (f) \sum_{ \ms{U} \in \{ \gv - \check{\gv}, \ms{Q}, \Lv - \check{\Lv}, \ms{B} \} } | \ms{U} |^2 \, d g \lesssim \int_{ \Omega_e } w_\lambda (f) \, ( \dots ) \, d g \text{.}
\]
Finally, $w_\lambda (f)$ in the above can be removed in the standard fashion by noting that $w_\lambda (f) \leq w_\lambda ( f_i )$ on $\Omega_e$ and $w_\lambda (f) \geq w_\lambda ( f_i )$ on $\Omega_i$.
The result \eqref{eq.intro_proof_goal} now follows by letting $\lambda \rightarrow \infty$.

\subsection{Comparison with Similar Results}

It is instructive to compare the proof of Theorem \ref{theo:adscft} with those of some related results in the existing literature.

\subsubsection{Biquard's Riemannian analogue}

We recall \cite{Biquard}, which considered asymptotically hyperbolic Einstein manifolds $( M, g )$.
These (Riemannian) manifolds also have a conformal boundary $( \partial M, \gm )$, as well as a Fefferman-Graham expansion from $\partial M$.
In this setting, \cite{Biquard} proved that \emph{the coefficients $( \gb{0}, \gb{n} )$ in the expansion uniquely determine the metric $g$ on $M$}---the analogue of Theorem \ref{theo:adscft}.

The main difference between Theorem \ref{theo:adscft} and \cite{Biquard} is that the key equations in the latter are elliptic.
Recall that all hypersurfaces are pseudoconvex in elliptic settings, hence the major difficulties of zero-pseudoconvexity and of constructing pseudoconvex hypersurfaces are entirely avoided.

Moreover, \cite{Biquard} can avoid working with the curvature directly, instead deriving a second order elliptic equation for the analogue of the second fundamental form $\Lv$, which sees arbitrary second derivatives of the metric $\gv$ on the right hand side.\footnote{This equation can be derived from analogues of \eqref{basiceq} and the Bianchi equation for $\partial_\rho \wv^2_{ab}$.}
Since the equation is elliptic, the Carleman estimate for $\Lv$ allows for controlling \emph{two} derivatives of $\Lv$ in terms of two derivatives of $\gv$.
Also, as the (commuted) transport equation for $\gv$ estimates two derivatives of $\gv$ in terms of two derivatives of $\Lv$, the Carleman estimates already close at the level of the second fundamental form.

For our setting, the analogous equation for $\Lv$ would be \emph{hyperbolic}, and the Carleman estimates cannot be closed in the same way, since the hyperbolic version loses a derivative compared with the elliptic case.
Consequently, we must also introduce the curvature as an unknown, which greatly complicates both our system and the ensuing analysis.

\subsubsection{The Ionescu-Klainerman symmetry extension} \label{sec.intro_ik}

Next, we look at \cite{ionescusk}, which proved a symmetry extension result similar to Theorem \ref{theo:killing}, but through finite hypersurface in a vacuum spacetime $( M, g )$.
In particular, \cite{ionescusk} showed that a Killing vector field $K$ on a domain $U \subset M$ can be extended through a point $p \in \partial U$, provided $\partial U$ is pseudoconvex near $p$.

The proof of this result begins by extending $K$ along a geodesic vector field $X$ through $\partial U$ using the Jacobi equation.
One key step in showing that this extended $K$ remains Killing is the derivation of a wave-transport system, on which a unique continuation result is applied:
\begin{align}
\label{eq.intro_ik_system} \nabla_X \mathbf{B} &= c \cdot \mathbf{B} + c \cdot \mathbf{P} \text{,} \\
\notag \nabla_X \mathbf{P} &= c \cdot \mathbf{A} + c \cdot \mathbf{B} + c \cdot \mathbf{P} \text{,} \\
\notag \Box \mathbf{A} &= c \cdot \nabla \mathbf{P} + c \nabla \mathbf{B} + c \cdot \mathbf{P} + c \cdot \mathbf{B} + c \cdot {\bf A} \cdot {\bf A} \text{.}
\end{align}
Here, ``$c \, \cdot$" denotes various contractions with tensorial coefficients, which we avoid specifying here.
The unknowns $\mathbf{B}$, $\mathbf{P}$, $\mathbf{A}$ are spacetime tensor fields, roughly described as follows:
\begin{itemize}
\item $\mathbf{B}$ consists of $\mi{L}_K g$ plus a specially chosen antisymmetric renormalization term $\omega$, while $\mathbf{P}$ consists of certain careful combinations of $\nabla \mi{L}_K g$ and $\nabla \omega$:
\[
\mathbf{B}_{\alpha \beta} := \frac{1}{2} ( \mi{L}_K g_{\alpha \beta} + \omega_{\alpha \beta} ) \text{,} \qquad \mathbf{P}_{\alpha \beta \mu} := \nabla_\alpha \pi_{\beta \mu} - \nabla_\beta \pi_{\alpha \mu} - \nabla_\mu \omega_{\alpha \beta} \text{.}
\]

\item $\mathbf{A}$ is a ``modified Lie derivative" of the Weyl curvature $W$:
\[
\mathbf{A}_{ \alpha \beta \mu \nu } = \mi{L}_K W_{ \alpha \beta \mu \nu } - ( \mathbf{B} \odot W )_{ \alpha \beta \mu \nu } \text{.}
\]
\end{itemize}
That $K$ is Killing follows from showing, via unique continuation, that $\mathbf{B}$, $\mathbf{P}$, $\mathbf{A}$ all vanish.

There are two connections we can make between the system \eqref{eq.intro_ik_system} and our results.
The first is that an analogous system can be applied to give a direct proof of Theorem \ref{theo:killing}, without appealing to Theorem \ref{theo:adscft}.
Setting $X := \rho \partial_\rho$ and $K$ to be extension of the Killing field from the conformal boundary, we obtain a system of the same form \eqref{eq.intro_ik_system}.\footnote{In aAdS settings, $\Box$ is replaced by $\Box + 2 n$ due to the cosmological constant.}
However, we would also need to apply vertical decompositions to $\mathbf{B}$, $\mathbf{P}$, $\mathbf{A}$, since different components have different asymptotic behaviors at the conformal boundary.
Nonetheless, this decomposed system has the same qualities as \eqref{eq.intro_proof_transport}--\eqref{eq.intro_proof_wave}, and we can similarly apply our Carleman estimates to this.

The second connection is that we can in fact draw a direct parallel between \eqref{eq.intro_ik_system} and our wave-transport system \eqref{eq.intro_proof_transport}--\eqref{eq.intro_proof_wave}.
One can construct a rough ``dictionary" between the unknowns of \cite{ionescusk} and our system by \emph{replacing each $\mi{L}_K$ applied to a quantity by the corresponding difference of that quantity for two metrics}.
More specifically, we identify the following:
\begin{itemize}
\item $\mi{L}_K g$ in \cite{ionescusk} corresponds to $\gv - \check{\gv}$ in our paper.

\item The renormalized term $\omega$ in \cite{ionescusk} corresponds to our renormalization $\ms{Q}$.

\item The components of $\mathbf{P}$ roughly map to both $\Lv - \check{\Lv}$ and $\ms{B}$ in our paper.

\item The modified Lie derivative $\mathbf{A}$ corresponds to our renormalized curvature differences $\ms{W}^\star$, $\ms{W}^1$, $\ms{W}^2$.
Moreover, $\mathbf{B} \otimes W$ connect directly to the renormalization terms in \eqref{eq.intro_proof_wdiff}.
\end{itemize}

\begin{remark}
As a result of the above, the preceding discussions of our system \eqref{eq.intro_proof_transport}--\eqref{eq.intro_proof_wave} also help to explain the various renormalizations used in \eqref{eq.intro_ik_system} and \cite{ionescusk}.
\end{remark}


\subsection{Further Questions}

Finally, we conclude the introduction by discussing some further directions of investigation that are related to or raised by Theorem \ref{theo:adscft}.

\subsubsection{The Case $n = 2$.}

Recall Theorem \ref{theo:adscft}---and Theorems \ref{theo:killing} and \ref{theo:symmetries} by extension---all assume that the dimension $n$ of the conformal boundary is strictly greater than $2$.
This raises the question of whether analogues of Theorems \ref{theo:adscft}, \ref{theo:killing}, and \ref{theo:symmetries} hold in the case $n = 2$.

In fact, the problem simplifies considerably when $n = 2$ due to the rigidity of low-dimensional settings.
In particular, since the Weyl curvature vanishes identically in $3$ dimensions, it follows already that any vacuum aAdS spacetime when $n = 2$ must be locally isometric to the ($3$-dimensional) AdS metric.
Furthermore, as all curvature terms disappear from the system \eqref{basiceq}, one can prove unique continuation using only transport equations (and avoiding wave equations).
This yields analogues of all our main theorems for $n=2$, but from \emph{any} domain $\mi{D}$---\emph{without requiring the GNCC}.\footnote{This is consistent with the fact that the Einstein-vacuum equations lose their hyperbolicity in ($2+1$)-dimensions.}

\subsubsection{Optimal Boundary Conditions}

An often studied setting in the physics literature is the case when the boundary region $\mi{D} \subset \mi{I}$ in Theorem \ref{theo:adscft} is a causal diamond,\footnote{$\mc{I}^+$ and $\mc{I}^-$ denotes the causal future and past, respectively, in $( \mi{I}, \gm )$.}
\begin{equation}
\label{eq.intro_causal_diamond} \mi{D} := \mc{I}^+ (p) \cap \mc{I}^- (q) \text{,} \qquad p, q \in \mi{I} \text{.}
\end{equation}
Unfortunately, one expects that causal diamonds \eqref{eq.intro_causal_diamond}, regardless of how large they are, should generically fail to satisfy the GNCC when $n > 2$; see the argument in \cite[Section 3.3]{ChatzikaleasShao}.
As a result, Theorem \ref{theo:adscft} fails to apply when $\mi{D}$ is as in \eqref{eq.intro_causal_diamond}---in other words, we cannot establish that vacuum aAdS spacetimes are uniquely determined by their boundary data on a causal diamond.

This leads to the question of whether the GNCC can be further refined, so that Theorem \ref{theo:adscft} can be somehow extended to apply to $\mi{D}$ as in \eqref{eq.intro_causal_diamond}.
One observation here is that the failure of the GNCC is due only to the presence of corners in $\partial \mi{D}$ where the boundaries of $\mc{I}^+ (p)$ and $\mc{I}^- (q)$ intersect.
Near these corners, one can find near-boundary null (spacetime) geodesics ``flying over" but avoiding $\mi{D}$.
This leads to the following question:\ \emph{Could boundary data $( \gb{0}, \gb{n} )$ on $\mi{D}$ uniquely determine the vacuum aAdS spacetime near some proper subset $\mi{D}' \subseteq \mi{D}$}, in particular when $\mi{D}$ is sufficiently large, and when $\mi{D}'$ is sufficiently far from any corners in $\partial \mc{I}^+ (p) \cap \partial \mc{I}^- (q)$?

At the same time, one may ask whether this refined GNCC can also be formulated for more general domains $\mi{D} \subset \mi{I}$.
More specifically, one can formulate the following:

\begin{problem} \label{prb.intro_egncc}
Consider the setting of Theorem \ref{theo:adscft}.
Show that if $\mi{D}' \subseteq \mi{D}$ satisfies some (yet to be formulated) ``refined GNCC" relative to $\mi{D}$, then $( \mi{M}, g )$ is uniquely determined near $\mi{D}'$ by the boundary data $( \gb{0}, \gb{n} )$ on $\mi{D}$, again up to gauge equivalence.
\end{problem}

Keeping with the above intuitions, the optimal formulation of such a ``refined GNCC" would be one that directly characterizes null geodesic trajectories near the conformal boundary.
Such a criterion would confirm the belief that unique continuation holds if and only if one cannot construct geometric optics counterexamples near the conformal boundary to unique continuation for waves, as in \cite{alin_baou:non_unique}.
However, a proof of such a statement may require incorporating, in a novel manner, ideas from microlocal analysis and propagation of singularities.

\subsubsection{Global Correspondences}

Our main result, Theorem \ref{theo:adscft}, is ``local" in nature, in the sense that the vacuum spacetime is only uniquely determined near the conformal boundary.
This is due to our rather general setup, which does not provide any information on the global spacetime geometry.
However, this leaves open the question of whether a more global unique continuation result can be established if more additional assumptions are imposed.

For example, one can consider aAdS spacetimes $( \mi{M}, g )$ that are global perturbations, in some sense, of a Kerr-AdS spacetime.
One can then ask whether the boundary data $( \gb{0}, \gb{n} )$ determines the spacetime in the full domain of outer communications, or if additional conditions are needed to rule out bifurcating counterexamples.
In the positive scenario, another physical question of interest is whether one can construct a one-to-one correspondence between conformal boundary data $( \gb{0}, \gb{n} )$ and some (appropriately conceived) data on the black hole horizon.

In \cite[Section 6]{hol_shao:uc_ads}, the authors applied the Carleman estimates of that paper to show that the linearized Einstein-vacuum equations on AdS spacetime (formulated as Bianchi equations for spin-$2$ fields) is globally characterized by its boundary data on a sufficiently long time interval.
Upcoming work by McGill and the second author will extend this to the nonlinear setting---roughly, under additional global assumptions, AdS spacetime is globally uniquely determined, as a solution to \eqref{evec}, by its holographic boundary data.
An interesting next step would be to explore whether these analyses can be extended to black hole aAdS spacetimes.

\subsection{Organization of the Paper}

In Section \ref{sec.aads}, we provide a detailed development of vacuum FG-aAdS segments and our vertical tensor formalism.
Section \ref{sec.sys} is dedicated to the wave-transport system that is at the heart of the proof of Theorem \ref{theo:adscft}, while Section \ref{sec.carleman} presents the key Carleman estimates for both wave and transport equations.
Finally, Section \ref{sec.proof} proves the key unique continuation result for vacuum FG-aAdS segments, and Section \ref{sec.app} proves our main results:\ Theorems \ref{theo:adscft}, \ref{theo:killing}, and \ref{theo:symmetries}.
Finally, various proofs and derivations are presented separately in Appendix \ref{sec.extra}.

\subsection*{Acknowledgments}

A.S.~acknowledges support by EPSRC grant EP/R011982/1 for a portion of this project.
G.H.~acknowledges support by the Alexander von Humboldt Foundation in the framework of the Alexander von Humboldt Professorship endowed by the Federal Ministry of Education and Research as well as ERC Consolidator Grant 772249 and funding through Germany’s Excellence Strategy EXC 2044 390685587, Mathematics M\"unster: Dynamics–Geometry–Structure.

%% file: aads.tex
This section is devoted to developing the background material that will be used throughout this article.
First, we give a precise description of the aAdS setting that we will study, and we state the assumptions we will impose on our spacetimes and their conformal boundaries.
We then turn our attention to Einstein-vacuum spacetimes, and we recall the Fefferman--Graham partial expansions derived in \cite{shao:aads_fg}.
In the remaining parts, we recall the mixed tensor fields introduced in \cite{mcgill_shao:psc_aads}.\footnote{Similar notions were also used in \cite{hol_shao:uc_ads, hol_shao:uc_ads_ns}.}
These are used to make sense of the wave operator $\bar{\Box}$ applied to vertical tensor fields.
Finally, we derive various identities connecting vertical and spacetime geometric quantities.

\subsection{Asymptotically AdS Spacetimes} \label{sec.aads_fg}

The first objective is to give precise descriptions of the aAdS spacetimes that we will study.
We begin with the background manifold itself:\footnote{Most of the material in this subsection can also be found in \cite[Sections 2.1--2.3]{mcgill_shao:psc_aads}.
We give an abridged discussion here for the purpose of keeping the present article self-contained.}

\begin{definition} \label{def.aads_manifold}
An \emph{aAdS region} is a manifold with boundary of the form 
\begin{equation}
\label{eq.aads_manifold} \mi{M} := ( 0, \rho_0 ] \times \mi{I} \text{,} \qquad \rho_0 > 0 \text{,}
\end{equation}
where $\mi{I}$ is a smooth $n$-dimensional manifold, and where $n \in \N$.\footnote{While we refer to $\mi{M}$ as the aAdS region, this also implicitly includes the associated quantities $n$, $\mi{I}$, $\rho_0$.}
\end{definition}

\begin{definition} \label{def.aads_vertical}
Let $\mi{M}$ be an aAdS region.
Then:
\begin{itemize}
\item We let $\rho \in C^\infty ( \mi{M} )$ denote the projection onto its $( 0, \rho_0 ]$-component.

\item We let $\partial_\rho$ denote the lift to $\mi{M}$ of the canonical vector field $d_\rho$ on $( 0, \rho_0 ]$.

\item The \emph{vertical bundle} $\ms{V}^k_l \mi{M}$ of rank $( k, l )$ over $\mi{M}$ is defined to be the manifold consisting of all tensors of rank $( k, l )$ on each level set of $\rho$ in $\mi{M}$:
\begin{equation}
\label{eq.aads_vertical} \ms{V}^k_l \mi{M} = \bigcup_{ \sigma \in ( 0, \rho_0 ] } T^k_l \{ \rho = \sigma \} \text{.}
\end{equation}

\item A (smooth) section of $\ms{V}^k_l \mi{M}$ is called a \emph{vertical tensor field} of rank $( k, l )$.
\end{itemize}
\end{definition}

\begin{remark} \label{rmk.aads_tensor}
We adopt the following conventions and identifications on an aAdS region $\mi{M}$:
\begin{itemize}
\item We use italicized font, serif font, and Fraktur font (for instance, $g$, $\gv$, and $\gm$) to denote tensor fields on $\mi{M}$, vertical tensor fields, and tensor fields on $\mi{I}$, respectively.

\item Given $\sigma \in ( 0, \rho_0 ]$, we let $\ms{A} |_\sigma$ be the tensor field on $\mi{I}$ obtained from restricting $\ms{A}$ to $\{ \rho = \sigma \}$.

\item A vertical tensor field $\ms{A}$ of rank $( k, l )$ can be equivalently viewed as a one-parameter family, $\{ \ms{A} |_\sigma \mid \sigma \in ( 0, \rho_0 ] \}$, of rank $( k, l )$ tensor fields on $\mi{I}$.

\item Given a tensor field $\mf{A}$ on $\mi{I}$, we will also use $\mf{A}$ to denote the vertical tensor field on $\mi{M}$ obtained by extending $\mf{A}$ as a $\rho$-independent family of tensor fields on $\mi{I}$.\footnote{In particular, a scalar function on $\mi{I}$ also defines a $\rho$-independent function on $\mi{M}$.}

\item Any vertical tensor field $\ms{A}$ can be uniquely identified with a tensor field on $\mi{M}$ (of the same rank) via the following rule:~the contraction of any component of $\ms{A}$ with $\partial_\rho$ or $d \rho$ (whichever is appropriate) is defined to vanish identically.
\end{itemize}
Finally, unless otherwise specfied, we always implicitly assume any given tensor field is smooth.
\end{remark}

\begin{definition} \label{def.aads_lie}
Let $\mi{M}$ be an aAdS region.
\begin{itemize}
\item We use the symbol $\mi{L}$ to denote Lie derivatives of tensor fields, on both $\mi{M}$ and $\mi{I}$.

\item We can also make sense of Lie derivatives of any vertical tensor field $\ms{A}$ by treating it as a spacetime tensor field, as described in Remark \ref{rmk.aads_tensor}.

\item For convenience, we will often abbreviate $\mi{L}_{ \partial_\rho }$ as $\mi{L}_\rho$.
\end{itemize}
\end{definition}

\begin{remark}
Observe that, in the context of Definition \ref{def.aads_lie}, for any vertical tensor field $\ms{A}$, its Lie derivative $\mi{L}_\rho \ms{A}$ can also be characterized as the unique vertical tensor field satisfying\footnote{\eqref{eq.aads_vertical_lie} was the definition of $\mi{L}_\rho$ used in \cite{mcgill_shao:psc_aads, shao:aads_fg}, and can be shown to be consistent with Definition \ref{def.aads_lie}.}
\begin{equation}
\label{eq.aads_vertical_lie} \mi{L}_\rho \ms{A} |_\sigma = \lim_{ \sigma' \rightarrow \sigma } ( \sigma' - \sigma )^{-1} ( \ms{A} |_{ \sigma' } - \ms{A} |_\sigma ) \text{,} \qquad \sigma \in ( 0, \rho_0 ] \text{.}
\end{equation}
\end{remark}

In the following definitions, we establish conventions for coordinate systems and limits:

\begin{definition} \label{def.aads_index}
Let $\mi{M}$ be an aAdS region, and let $( U, \varphi )$ be a coordinate system on $\mi{I}$:
\begin{itemize}
\item Let $\varphi_\rho := ( \rho, \varphi )$ denote the corresponding lifted coordinates on $( 0, \rho_0 ] \times U$.

\item We use lower-case Latin indices $a, b, c, \dots$ to denote $\varphi$-coordinate components, and we use lower-case Greek indices $\alpha, \beta, \mu, \nu, \dots$ to denote $\varphi_\rho$-coordinate components.
As usual, repeated indices indicate summations over the appropriate components.

\item $( U, \varphi )$ is called \emph{compact} iff $\bar{U}$ is a compact subset of $\mi{I}$ and $\varphi$ extends smoothly to $\bar{U}$.

\item Given a vertical tensor field $\ms{A}$ of rank $( k, l )$, we define (with respect to $\varphi$-coordinates)
\begin{equation}
\label{eq.aads_norm} | \ms{A} |_{ M, \varphi } := \sum_{ m = 0 }^M \sum_{ \substack{ a_1, \dots, a_m \\ b_1, \dots, b_k \\ c_1, \dots, c_l } } | \partial^m_{ a_1 \dots a_m } \ms{A}^{ b_1 \dots b_k }_{ c_1 \dots c_l } | \text{.}
\end{equation}
\end{itemize}
\end{definition}

\begin{definition} \label{def.aads_limit}
Let $\mi{M}$ be an aAdS region, let $M \geq 0$, and let $\ms{A}$ and $\mf{A}$ be a vertical tensor field and a tensor field on $\mi{I}$, respectively, both of the same rank $(k, l)$.
\begin{itemize}
\item $\ms{A}$ is \emph{locally bounded in $C^M$} iff for any compact coordinates $( U, \varphi )$ on $\mi{I}$,
\begin{equation}
\label{eq.aads_bounded} \sup_{ ( 0, \rho_0 ] \times U } | \ms{A} |_{ M, \varphi } < \infty \text{.}
\end{equation}

\item $\ms{A}$ \emph{converges to $\mf{A}$ in $C^M$}, denoted $\ms{A} \rightarrow^M \mf{A}$, iff for any compact coordinates $( U, \varphi )$ on $\mi{I}$,
\begin{equation}
\label{eq.aads_limit} \lim_{ \sigma \searrow 0 } \sup_{ \{ \sigma \} \times U } | \ms{A} - \mf{A} |_{ M, \varphi } = 0 \text{.}
\end{equation}
\end{itemize}
\end{definition}

We now describe the metrics that we will consider on our aAdS segments.
This is summarized through the notion of ``FG-aAdS segments" from \cite{mcgill_shao:psc_aads, shao:aads_fg}.

\begin{definition} \label{def.aads_metric}
$( \mi{M}, g )$ is called an \emph{FG-aAdS segment} iff the following hold:\footnote{Though we refer to $( \mi{M}, g )$ as the FG-aAdS segment, this also implicitly includes the quantities $\gv$ and $\gm$ below.}
\begin{itemize}
\item $\mi{M}$ is an aAdS region, and $g$ is a Lorentzian metric on $\mi{M}$.

\item There exist a vertical tensor field $\gv$ of rank $( 0, 2 )$ and a Lorentzian metric $\gm$ on $\mi{I}$ with
\begin{equation}
\label{eq.aads_metric} g := \rho^{-2} ( d \rho^2 + \gv ) \text{,} \qquad \gv \rightarrow^0 \gm \text{.}
\end{equation}
\end{itemize}
\end{definition}

\begin{remark}
Given an FG-aAdS segment $( \mi{M}, g )$:
\begin{itemize}
\item We refer to the form \eqref{eq.aads_metric} of $g$ as the \emph{Fefferman--Graham} (or \emph{FG}) \emph{gauge condition}.

\item We refer to $( \mi{I}, \gm )$ as the \emph{conformal boundary} for $( \mi{M}, g, \rho )$.
\end{itemize}
\end{remark}

The following definitions describe the basic geometric objects in our setting:

\begin{definition} \label{def.aads_covar}
Given an FG-aAdS segment $( \mi{M}, g )$:
\begin{itemize}
\item Let $g^{-1}$, $\nabla$, $\nabla^\sharp$, and $R$ denote the metric dual, the Levi-Civita connection, the gradient, and the Riemann curvature (respectively) associated with the spacetime metric $g$.

\item Let $\gm^{-1}$, $\Dm$, $\Dm^\sharp$, and $\Rm$ denote the metric dual, the Levi-Civita connection, the gradient, and the Riemann curvature (respectively) associated with the boundary metric $\gm$.

\item Let $\gv^{-1}$, $\Dv$, $\Dv^\sharp$, and $\Rv$ denote the metric dual, the Levi-Civita connection, the gradient, and the Riemann curvature (respectively) associated with the vertical metric $\gv$.\footnote{More specifically, $\gv^{-1} |_\sigma$ and $\Rv |_\sigma$ are the metric dual and Riemann curvature of $\gv |_\sigma$ for any $\sigma \in ( 0, \rho_0 ]$.
Moreover, $\Dv$ and $\Dv^\sharp$ act like the Levi-Civita connection and the gradient for $\gv |_\sigma$ on each $\{ \rho = \sigma \}$, for all $\sigma \in ( 0, \rho_0 ]$.}
\end{itemize}
As is standard, we omit the superscript ``${}^{-1}$" when describing metric duals in index notion.
\end{definition}

\begin{definition} \label{def.aads_curvature}
Furthermore, given an FG-aAdS segment $( \mi{M}, g )$:
\begin{itemize}
\item Let $W$, $Rc$, and $Rs$ denote the Weyl, Ricci, and scalar curvatures (respectively) for $g$.

\item Let $\Rcm$ and $\Rsm$ denote the Ricci and scalar curvatures (respectively) for $\gm$.
\end{itemize}
\end{definition}

\subsection{Vacuum Spacetimes}

The final assumption we will pose is that our spacetime satisfies the Einstein-vacuum equations (with normalized negative cosmological constant).

\begin{definition} \label{def.aads_vacuum}
An FG-aAdS segment $( \mi{M}, g )$ is called \emph{vacuum} iff the following holds:
\begin{equation}
\label{eq.aads_vacuum} Rc - \frac{1}{2} Rs \cdot g + \Lambda \cdot g = 0 \text{,} \qquad \Lambda := - \frac{ n (n - 1) }{ 2 } \text{.}
\end{equation}
\end{definition}

\begin{proposition} \label{thm.aads_einstein}
Suppose $( \mi{M}, g )$ is a vacuum FG-aAdS segment.
Then,
\begin{equation}
\label{eq.aads_einstein} Rc = - n \cdot g \text{,} \qquad Rs = - n (n + 1) \text{.}
\end{equation}
Furthermore, the following holds with respect to any coordinates on $\mi{M}$:
\begin{equation}
\label{eq.aads_einstein_weyl} W_{ \alpha \beta \gamma \delta } = R_{ \alpha \beta \gamma \delta } + g_{ \alpha \gamma } g_{ \beta \delta } - g_{ \alpha \delta } g_{ \beta \gamma } \text{.}
\end{equation}
\end{proposition}

\begin{proof}
These are direct computations; see \cite[Proposition 2.24]{shao:aads_fg}.
\end{proof}

The following results, which give partial Fefferman--Graham expansions for Einstein-vacuum metrics from the conformal boundary---are a portion of the main results of \cite{shao:aads_fg}:

\begin{definition} \label{def.aads_regular}
Fix an integer $M \geq 0$.
An FG-aAdS segment $( \mi{M}, g )$ is \emph{regular to order $M$} iff:
\begin{itemize}
\item $\gv$ is locally bounded in $C^{ M + 2 }$.

\item The following holds for any compact coordinates $( U, \varphi )$ on $\mi{I}$:
\begin{equation}
\label{eq.aads_regular} \sup_U \int_0^{ \rho_0 } | \mi{L}_\rho \gv |_{ 0, \varphi } |_\sigma d \sigma < \infty \text{.}
\end{equation}
\end{itemize}
\end{definition}

\begin{definition} \label{def.aads_depend}
Let $( \mi{M}, g )$ be a FG-aAdS segment, and let $k \geq 2$.
We say that a tensor field $\mf{A}$ on $\mi{I}$ \emph{depends only on $\gm$ to order $k$} iff $\mf{A}$ can be expressed as contractions and tensor products of zero or more instances of each of the following: $\gm, \gm^{-1}, \Rm, \dots, \Dm^{ k - 2 } \Rm$.
\end{definition}

\begin{theorem}{\cite[Theorem 3.3]{shao:aads_fg}} \label{thm.aads_fg}
Let $( \mi{M}, g )$ be a vacuum FG-aAdS segment, and assume $n > 1$.
Moreover, suppose $( \mi{M}, g )$ is regular to some order $M \geq n + 2$.
Then:
\begin{itemize}
\item $\gv$ and $\gv^{-1}$ satisfy
\begin{equation}
\label{eq.aads_fg_low} \gv \rightarrow^M \gm \text{,} \qquad \gv^{-1} \rightarrow^M \gm^{-1} \text{.}
\end{equation}

\item There exist tensor fields $\gb{k}$, $0 \leq k < n$, on $\mi{I}$ such that
\begin{equation}
\label{eq.aads_fg_high} \mi{L}_\rho^k \gv \rightarrow^{ M - k } k! \, \gb{k} \text{,} \qquad \rho \mi{L}_\rho^{ k + 1 } \gv \rightarrow^{ M - k } 0 \text{.}
\end{equation}
Furthermore, $\gb{0} = \gm$, and the following properties hold:
\begin{itemize}
\item If $1 \leq k < n$ is odd, then $\gb{k} = 0$.

\item If $2 \leq k < n$, then $\gb{k}$ depends only on $\gm$ to order $k$.
In particular,
\begin{equation}
\label{eq.aads_fg_schouten} \gb{2} = - \frac{ 1 }{ n - 2 } \left[ \Rcm - \frac{ 1 }{ 2 ( n - 1 ) } \Rsm \cdot \gm \right] \text{,} \qquad n > 2 \text{.}
\end{equation}
\end{itemize}

\item There exists a tensor field $\gb{\star}$ on $\mi{I}$ such that
\begin{equation}
\label{eq.aads_fg_top} \rho \mi{L}_\rho^{ n + 1 } \gv \rightarrow^{ M - n } n! \, \gb{\star} \text{.}
\end{equation}
In addition, $\gb{\star}$ satisfies the following:
\begin{itemize}
\item Both the $\gm$-trace and the $\gm$-divergence of $\gb{\star}$ vanish on $\mi{I}$.

\item If $n$ is odd or if $n = 2$, then $\gb{\star} = 0$.

\item If $n$ is even, then $\gb{\star}$ depends only on $\gm$ to order $n$.
\end{itemize}

\item There exist a $C^{ M - n }$ tensor field $\gb{\dagger}$ on $\mi{I}$ such that
\begin{equation}
\label{eq.fg_main_free} \mi{L}_\rho^n \gv - n! \, ( \log \rho ) \gb{\star} \rightarrow^{ M - n } n! \, \gb{\dagger} \text{.}
\end{equation}
\end{itemize}
\end{theorem}

\begin{remark} \label{rmk.aads_strong}
In particular, when $n > 2$, Theorem \ref{thm.aads_fg} implies that any vacuum FG-aAdS segment $( \mi{M}, g )$ is also a \emph{strongly FG-aAdS segment}, in the sense of \cite[Definition 2.13]{mcgill_shao:psc_aads}---that is,
\begin{equation}
\label{eq.aads_strong} \gv \rightarrow^3 \gm \text{,} \qquad \mi{L}_\rho \gv \rightarrow^2 0 \text{,} \qquad \mi{L}_\rho^2 \gv \rightarrow^1 \gs \text{,}
\end{equation}
for some rank $( 0, 2 )$ tensor field $\gs$ on $\mi{I}$, and $\mi{L}_\rho^3 \gv$ is locally bounded in $C^0$.\footnote{Notice also that when $n = 2$, the conclusions of Theorem \ref{thm.aads_fg} still imply the limits \eqref{eq.aads_strong}.}
These are the main regularity and asymptotic assumptions required for the Carleman estimates of \cite{mcgill_shao:psc_aads} to hold.
\end{remark}

\begin{corollary}{\cite[Theorem 3.6]{shao:aads_fg}} \label{thm.aads_fg_exp}
Assume the hypotheses of Theorem \ref{thm.aads_fg}, and let the quantities $\gb{0}, \dots, \gb{n-1}, \gb{\star}$ be as in the conclusions of Theorem \ref{thm.aads_fg}.
Then, there exists a $C^{ M - n }$ tensor field $\gb{n}$ on $\mi{I}$ and a vertical tensor field $\ms{r}_{ \gv }$ such that the following partial expansion holds for $\gv$,
\begin{equation}
\label{eq.aads_fg_exp} \gv = \begin{cases} \sum_{ k = 0 }^{ \frac{ n - 1 }{2} } \gb{ 2k } \rho^{ 2 k } + \gb{n} \rho^n + \ms{r}_{ \gv } \rho^n & \text{$n$ odd,} \\ \sum_{ k = 0 }^{ \frac{ n - 2 }{2} } \gb{ 2k } \rho^{ 2 k } + \gb{\star} \rho^n \log \rho + \gb{n} \rho^n + \ms{r}_{ \gv } \rho^n & \text{$n$ even,} \end{cases}
\end{equation}
where the remainder $\ms{r}_{ \gv }$ satisfies
\begin{equation}
\label{eq.aads_fg_error} \ms{r}_{ \gv } \rightarrow^{ M - n } 0 \text{.}
\end{equation}
Furthermore, $\gb{n}$ satisfies the following:
\begin{itemize}
\item If $n$ is odd, then the $\gm$-trace and $\gm$-divergence of $\gb{n}$ vanish on $\mi{I}$.

\item On the other hand, if $n$ is even, then the $\gm$-trace of $\gb{n}$ depends only on $\gm$ to order $n$, and the $\gm$-divergence of $\gb{n}$ depend only on $\gm$ to order $n+1$.
\end{itemize}
\end{corollary}

\begin{remark}
The conclusions of Theorem \ref{thm.aads_fg} and Corollary \ref{thm.aads_fg_exp} imply that the coefficients $\gb{1}, \dots, \gb{n-1}, \gb{\star}$---as well as the $\gm$-trace and the $\gm$-divergence of $\gb{n}$---are determined by the boundary metric $\gm = \gb{0}$.
As a result, we can view $\gb{0}$ and the $\gm$-trace-free, $\gm$-divergence-free part of $\gb{n}$ as the ``free" data for the Einstein-vacuum equations at the conformal boundary.
\end{remark}

\begin{definition} \label{def.aads_boundary_data}
Let $( \mi{M}, g )$ be a vacuum FG-aAdS segment, and let $\gb{0}, \dots, \gb{n-1}, \gb{\star}, \gb{n}$ be defined as in the statements of Theorem \ref{thm.aads_fg} and Corollary \ref{thm.aads_fg_exp} above.
We then refer to the triple $( \mi{I}, \gb{0}, \gb{n} )$ as the \emph{holographic data} associated to or induced by $( \mi{M}, g )$.
\end{definition}

\subsection{The Mixed Covariant Formalism} \label{sec.aads_mixed}

In this subsection, we recall the notion of mixed tensor fields from \cite{mcgill_shao:psc_aads}.
In order to better handle some of the more complicated tensorial expressions in this section, we will make use of the following notational conventions for multi-indices:

\begin{definition} \label{def.aads_indices}
In general, we will use symbols containing an overhead bar to denote multi-indices.
Moreover, given a multi-index $\ix{A} := A_1 \dots A_l$ (with spacetime or vertical components):
\begin{itemize}
\item For any $1 \leq i \leq l$, we write \smash{$\ixd{A}{i}$} to denote $\ix{A}$, but with $A_i$ removed.
Moreover, given another index $B$, we write \smash{$\ixr{A}{i}{B}$} to denote $\ix{A}$, but with $A_i$ replaced by $B$.

\item Similarly, given any $1 \leq i, j \leq l$, with $i \neq j$, we write \smash{$\ixd{A}{ i, j }$} to denote $\ix{A}$ except with $A_i$ and $A_j$ removed.
Furthermore, given any indices $B$ and $C$, we write \smash{$\ixr{A}{ i, j }{ B, C }$} to denote $\ix{A}$, but with $A_i$ and $A_j$ replaced by $B$ and $C$, respectively.
\end{itemize}
\end{definition}

The first step in this process is to construct connections on the vertical bundles.
The Levi-Civita connections $\Dv$ already define covariant derivatives of vertical tensor fields in the vertical directions.
We now extend these connections to also act in the $\rho$-direction.

\begin{proposition} \label{thm.aads_vertical_connection}
Let $( \mi{M}, g )$ be an FG-aAdS segment.
There exists a (unique) family of connections $\Dvm$ on the vertical bundles $\ms{V}^k_l \mi{M}$, for all ranks $( k, l )$, such that given any vertical tensor field $\ms{A}$ of rank $( k, l )$, the following formula holds, with respect to any coordinates $( U, \varphi )$ on $\mi{I}$,
\begin{align}
\label{eq.aads_vertical_connection} \bar{\Dv}_c \ms{A}^{ \ix{a} }{}_{ \bar{b} } &= \Dv_c \ms{A}^{ \bar{a} }{}_{ \bar{b} } \text{,} \\
\notag \bar{\Dv}_\rho \ms{A}^{ \ix{a} }{}_{ \ix{b} } &= \mi{L}_\rho \ms{A}^{ \ix{a} }{}_{ \ix{b} } + \frac{1}{2} \sum_{ i = 1 }^k \gv^{ a_i c } \mi{L}_\rho \gv_{ c d } \, \ms{A}^{ \ixr{a}{i}{d} }{}_{ \ix{b} } - \frac{1}{2} \sum_{ j = 1 }^l \gv^{ c d } \mi{L}_\rho \gv_{ b_j c } \, \ms{A}^{ \ix{a} }{}_{ \ixr{b}{j}{d} } \text{,}
\end{align}
where $\ix{a} := a_1 \dots a_k$ and $\ix{b} := b_1 \dots b_l$ are multi-indices.

Furthermore, given any vector field on $X$ on $\mi{M}$, the operator $\Dvm_X$ satisfies the following:
\begin{itemize}
\item For any vertical tensor fields $\ms{A}$ and $\ms{B}$,
\begin{equation}
\label{eq.aads_vertical_leibniz} \Dvm_X ( \ms{A} \otimes \ms{B} ) = \Dvm_X \ms{A} \otimes \ms{B} + \ms{A} \otimes \Dvm_X \ms{B} \text{.}
\end{equation}

\item For any vertical tensor field $\ms{A}$ and any tensor contraction operation $\mc{C}$,
\begin{equation}
\label{eq.aads_vertical_contract} \Dvm_X ( \mc{C} \ms{A} ) = \mc{C} ( \Dvm_X \ms{A} ) \text{.}
\end{equation}

\item $\Dvm_X$ annihilates the vertical metric:
\begin{equation}
\label{eq.aads_vertical_compat} \Dvm_X \gv = 0 \text{,} \qquad \Dvm_X \gv^{-1} = 0 \text{.}
\end{equation}
\end{itemize}
\end{proposition}

\begin{proof}
See \cite[Definition 2.22, Proposition 2.23]{mcgill_shao:psc_aads}.
\end{proof}

In summary, Proposition \ref{thm.aads_vertical_connection} states that $\Dvm$ extends the vertical Levi-Civita connections $\Dv$ to all directions along $\mi{M}$, satisfy the same algebraic properties as the usual Levi-Civita derivatives (such as $\nabla$ and $\Dv$), and are compatible with the vertical metric $\gv$.

\begin{remark}
If we identify vertical tensor fields with spacetime tensor fields via Remark \ref{rmk.aads_tensor}, then $\Dvm$ can alternately be defined as the Levi-Civita connection associated with $\rho^2 g$.
\end{remark}

Next, we construct mixed tensor bundles and their associated connections.

\begin{definition} \label{def.aads_mixed}
Let $( \mi{M}, g )$ be an FG-aAdS segment.
We then define the \emph{mixed bundle} of rank $( \kappa, \lambda; k, l )$ over $\mi{M}$ to be the tensor product bundle given by
\begin{equation}
\label{eq.aads_mixed} T^\kappa_\lambda \ms{V}^k_l \mi{M} := T^\kappa_\lambda \mi{M} \otimes \ms{V}^k_l \mi{M} \text{.}
\end{equation}
We refer to sections of $T^\kappa_\lambda \ms{V}^k_l \mi{M}$ as \emph{mixed tensor fields} of rank $( \kappa, \lambda; k, l )$.

Moreover, we define the connection $\bar{\nabla}$ on $T^\kappa_\lambda \ms{V}^k_l \mi{M}$ to be the tensor product connection of the spacetime connection $\nabla$ on $T^\kappa_\lambda \mi{M}$ and the vertical connection $\Dvm$ on $\ms{V}^k_l \mi{M}$.
More specifically, given any vector field $X$ on $\mi{M}$, tensor field $G$ on $\mi{M}$, and vertical tensor field $\ms{B}$, we have
\[
\nablam_X ( G \otimes \ms{B} ) := \nabla_X G \otimes \ms{B} + G \otimes \bar{\Dv}_X \ms{B} \text{.}
\]
\end{definition}

Less formally, mixed tensor fields are those with some components designated as ``spacetime" and other designated as ``vertical".
The mixed connections are then defined on mixed tensor fields by acting like $\nabla$ on the spacetime components and like $\Dvm$ on the vertical components.

\begin{remark}
Any tensor field of rank $( \kappa, \lambda )$ on $\mi{M}$ can be viewed as a mixed tensor field, with rank $( \kappa, \lambda; 0, 0 )$.
Similarly, any vertical tensor field is a mixed tensor field.
\end{remark}

\begin{proposition} \label{thm.aads_mixed_connection}
Let $( \mi{M}, g )$ be an FG-aAdS segment.
Then, given any vector field $X$ on $\mi{M}$:
\begin{itemize}
\item The following holds for any mixed tensor fields $\mathbf{A}$ and $\mathbf{B}$:\footnote{As usual, $\mathbf{A} \otimes \mathbf{B}$ is defined componentwise by multiplying the components of $\mathbf{A}$ and $\mathbf{B}$.}
\begin{equation}
\label{eq.aads_mixed_connection_leibniz} \nablam_X ( \mathbf{A} \otimes \mathbf{B} ) = \nablam_X \mathbf{A} \otimes \mathbf{B} + \mathbf{A} \otimes \nablam_X \mathbf{B} \text{.}
\end{equation}

\item The operator $\nablam_X$ annihilates both the spacetime and the vertical metrics:
\begin{equation}
\label{eq.aads_mixed_connection_compat} \nablam_X g = 0 \text{,} \qquad \nablam_X g^{-1} = 0 \text{,} \qquad \nablam_X \gv = 0 \text{,} \qquad \nablam_X \gv^{-1} = 0 \text{.}
\end{equation}
\end{itemize}
\end{proposition}

\begin{proof}
See \cite[Proposition 2.28]{mcgill_shao:psc_aads}.
\end{proof}

In summary, the mixed connections $\nablam$ naturally extend $\nabla$ and $\Dvm$ to mixed fields, have the same algebraic properties as the usual Levi-Civita derivatives, and are compatible with both $g$ and $\gv$.

The main reason for expanding our scope from vertical to mixed tensor fields is that we can now make sense of higher covariant derivatives of mixed tensor fields:

\begin{definition} \label{def.aads_mixed_wave}
Given an FG-aAdS segment $( \mi{M}, g )$ and a mixed tensor field $\mathbf{A}$ of rank $( \kappa, \lambda; k, l )$:
\begin{itemize}
\item The \emph{mixed covariant differential} of $\mathbf{A}$ is the mixed tensor field $\bar{\nabla} \mathbf{A}$, of rank $( \kappa, \lambda + 1; k, l )$, that maps each vector field $X$ on $\mi{M}$ (in the extra covariant slot) to $\bar{\nabla}_X \mathbf{A}$.

\item The \emph{mixed Hessian} $\nablam^2 \mathbf{A}$ is then defined to be the mixed covariant differential of $\bar{\nabla} \mathbf{A}$.

\item In particular, we now define $\Boxm \mathbf{A}$---the \emph{wave operator} applied to $\mathbf{A}$---to be the $g$-trace of $\bar{\nabla}^2 \mathbf{A}$, where the trace is applied to the two derivative components.
\end{itemize}
\end{definition}

\begin{remark}
In this article, we will only consider $\Boxm$ applied to vertical tensor fields.
The main novelty, and subtlety, in this case is that the outer derivative acts as a spacetime derivative $\nabla$ on the inner derivative slot and as a vertical derivative $\bar{\Dv}$ on the vertical tensor field itself.
\end{remark}

Finally, we list the following identities, which will be useful in upcoming computations:

\begin{proposition} \label{thm.aads_Gamma}
Let $( \mi{M}, g )$ be an FG-aAdS segment.
In addition, let $( U, \varphi )$ denote coordinates on $\mi{I}$, and let $\Gamma$ and $\Gammav$ denote Christoffel symbols in $\varphi_\rho$-coordinates for $\nabla$ and $\Dvm$, respectively:
\begin{equation}
\label{eq.aads_Gamma_def} \nabla_\alpha \partial_\beta := \Gamma^\gamma_{ \alpha \beta } \partial_\gamma \text{,} \qquad \Dvm_\alpha \partial_b := \Gammav^c_{ \alpha b } \partial_c \text{.}
\end{equation}
Then, the following relations hold:
\begin{align}
\label{eq.aads_Gamma} \Gamma_{ \rho \rho }^\alpha = - \rho^{-1} \delta^\alpha{}_\rho \text{,} &\qquad \Gamma_{ a \rho }^\rho = 0 \text{,} \\
\notag \Gamma_{ a \rho }^c = - \rho^{-1} \delta^c{}_a + \frac{1}{2} \gv^{ c d } \mi{L}_\rho \gv_{ a d } \text{,} &\qquad \Gamma_{ \rho a }^c - \Gammav_{ \rho a }^c = - \rho^{-1} \delta^c{}_a \text{,} \\
\notag \Gamma_{ a b }^\rho = \rho^{-1} \gv_{ a b } - \frac{1}{2} \mi{L}_\rho \gv_{ a b } \text{,} &\qquad \Gamma_{ a b }^c - \Gammav_{ a b }^c = 0 \text{.}
\end{align}
Furthermore, for any mixed tensor field $\mathbf{A}$ of rank $( \kappa, \lambda; k, l )$, we have, in $\varphi$- and $\varphi_\rho$-coordinates,
\begin{align}
\label{eq.aads_deriv} \nablam_\gamma \mathbf{A}^{ \ix{\alpha} }{}_{ \ix{\beta} }{}^{ \ix{a} }{}_{ \ix{b} } &= \partial_\gamma ( \mathbf{A}^{ \ix{\alpha} }{}_{ \ix{\beta} }{}^{ \ix{a} }{}_{ \ix{b} } ) + \sum_{ i = 1 }^\kappa \Gamma^{ \alpha_i }_{ \gamma \delta } \, \mathbf{A}^{ \ixr{\alpha}{i}{\delta} }{}_{ \ix{\beta} }{}^{ \ix{a} }{}_{ \ix{b} } - \sum_{ j = 1 }^\lambda \Gamma^\delta_{ \gamma \beta_j } \, \mathbf{A}^{ \ix{\alpha} }{}_{ \ixr{\beta}{j}{\delta} }{}^{ \ix{a} }{}_{ \ix{b} } \\
\notag &\qquad + \sum_{ i = 1 }^k \Gammav^{ a_i }_{ \gamma d } \, \mathbf{A}^{ \ix{\alpha} }{}_{ \ix{\beta} }{}^{ \ixr{a}{i}{d} }{}_{ \ix{b} } - \sum_{ j = 1 }^l \Gammav^d_{ \gamma b_j } \, \mathbf{A}^{ \ix{\alpha} }{}_{ \ix{\beta} }{}^{ \ix{a} }{}_{ \ixr{b}{j}{d} } \text{,}
\end{align}
where $\ix{\alpha} := \alpha_1 \dots \alpha_\kappa$, $\ix{\beta} := \beta_1 \dots \beta_\lambda$, $\ix{a} := a_1 \dots a_k$, and $\ix{b} := b_1 \dots b_l$.
\end{proposition}

\begin{proof}
\eqref{eq.aads_Gamma} follows from \eqref{eq.aads_metric} and \eqref{eq.aads_vertical_connection}, while \eqref{eq.aads_deriv} follows from Definition \ref{def.aads_mixed}.
\end{proof}

\subsection{Some General Formulas}

Next, we provide some general formulas for vertical tensor fields, as well as relations between spacetime and vertical tensor fields.
These will be used for deriving many of the equations we will need for proving Theorem \ref{theo:adscft}.
Moreover, we give a general development here, as these will be of independent interest beyond the present article.

\begin{remark}
Similar formulas were proved in \cite{mcgill:loc_ads}, however here we need more details for the ``error" terms.
Thus, we provide full derivations of these properties in Appendix \ref{sec.extra}.
\end{remark}

First, we devise some schematic notations, originally from \cite{shao:aads_fg}, for describing error terms:

\begin{definition} \label{def.aads_schematic}
Let $( \mi{M}, g )$ be an FG-aAdS segment.
Given any $N \geq 1$ and vertical tensor fields $\ms{A}_1, \dots, \ms{A}_N$ on $\mi{M}$, we write $\sch{ \ms{A}_1, \dots, \ms{A}_N }$ to represent any vertical tensor field of the form
\begin{equation}
\label{eq.aads_schematic} \sum_{ j = 1 }^J \mc{Q}_j ( \ms{A}_1 \otimes \dots \otimes \ms{A}_N ) \text{,} \qquad J \geq 0 \text{,}
\end{equation}
where each $\mc{Q}_j$, $1 \leq j \leq J$, is a composition of zero or more of the following operations:
\begin{itemize}
\item Component permutations.

\item (Non-metric) contractions.

\item Multiplications by a scalar constant.
\end{itemize}
\end{definition}

\begin{definition} \label{def.aads_schematic_g}
Let $( \mi{M}, g )$ be an FG-aAdS segment.
\begin{itemize}
\item For any $N \geq 1$, we define the shorthands
\begin{equation}
\label{eq.aads_schematic_g} \gv^N := \bigotimes_{ i = 1 }^N \gv \text{,} \qquad \gv^{-N} := \bigotimes_{ i = 1 }^N \gv^{-1} \text{.}
\end{equation}

\item For brevity, we also use the shorthand $\Lv$ to denote the $\rho$-derivative of $\gv$:
\begin{equation}
\label{eq.aads_sff} \Lv := \mi{L}_\rho \gv \text{.}
\end{equation}
\end{itemize}
\end{definition}

Next, we establish some general identities for vertical tensor fields:

\begin{proposition} \label{thm.aads_comm}
Let $( \mi{M}, g )$ be an FG-aAdS segment.
Then:
\begin{itemize}
\item The following commutation identities hold for any vertical tensor field $\ms{A}$:\footnote{To clarify, if $\ms{A}$ has rank $( k, l )$, then $\Dvm_\rho ( \Dv \ms{A} )$ denotes $\Dvm_\rho$ acting on the rank $( k, l + 1 )$ field $\Dv \ms{A}$, while $\Dv ( \Dvm_\rho \ms{A} )$ denotes the covariant differential of the rank $( k, l )$ field $\Dvm_\rho \ms{A}$.
A similar point holds for $\mi{L}_\rho ( \Dv \ms{A} )$ and $\Dv ( \mi{L}_\rho \ms{A} )$.}
\begin{align}
\label{eq.aads_comm} \mi{L}_\rho ( \Dv \ms{A} ) &= \Dv ( \mi{L}_\rho \ms{A} ) + \mi{S} ( \gv^{-1}, \Dv \Lv, \ms{A} ) \text{,} \\
\notag \Dvm_\rho ( \Dv \ms{A} ) &= \Dv ( \Dvm_\rho \ms{A} ) + \mi{S} ( \gv^{-1}, \Lv, \Dv \ms{A} ) + \mi{S} ( \gv^{-1}, \Dv \Lv, \ms{A} ) \text{.}
\end{align}

\item The following identity holds for vertical tensor field $\ms{A}$ and $p \in \R$,
\begin{equation}
\label{eq.aads_comm_rho} \Boxm ( \rho^p \ms{A} ) = \rho^p \Boxm \ms{A} + 2 p \rho^{ p + 1 } \Dvm_\rho \ms{A} - p ( n - p ) \rho^p \ms{A} + p \rho^2 \, \mi{S} ( \gv^{-1}, \Lv, \rho^p \ms{A} ) \text{.}
\end{equation}

\item Furthermore, the following hold for any vertical tensor field $\ms{A}$:
\begin{align}
\label{eq.aads_boxm} \Boxm \ms{A} &= \rho^2 \mi{L}_\rho^2 \ms{A} - ( n - 1 ) \rho \mi{L}_\rho \ms{A} + \rho^2 \gv^{ab} \Dv_{ab} \ms{A} + \rho^2 \, \sch{ \gv^{-1}, \Lv, \mi{L}_\rho \ms{A} } \\
\notag &\qquad + \rho \, \sch{ \gv^{-1}, \Lv, \ms{A} } + \rho^2 \, \sch{ \gv^{-1}, \mi{L}_\rho \Lv, \ms{A} } + \rho^2 \, \sch{ \gv^{-2}, \Lv, \Lv, \ms{A} } \text{.}
\end{align}
\end{itemize}
\end{proposition}

\begin{proof}
See Appendix \ref{sec.aads_comm}.
\end{proof}

Lastly, the key differential equations behind the proof of Theorem \ref{theo:adscft} are given in terms of spacetime tensor fields.
On the other hand, our main quantities of analysis are vector tensor fields, for which one can make sense of limits at the conformal boundary.
Thus, we will need to convert our equations for spacetime quantities into corresponding equations for vertical quantities.

\begin{proposition} \label{thm.aads_decomp}
Let $( \mi{M}, g )$ be an FG-aAdS segment.
Let $F$ be a rank $( 0, l_1 + l_2 )$ tensor field on $\mi{M}$, and let $\ms{f}$ be the rank $( 0, l_2 )$ vertical tensor field defined, in any coordinates $( U, \varphi )$ on $\mi{I}$, by
\begin{equation}
\label{eq.aads_decomp_phi} \ms{f}_{ \ix{a} } := F_{ \ix{\rho} \ix{a} } \text{,}
\end{equation}
where the multi-index $\ix{\rho} := \rho \dots \rho$ represents $l_1$ copies of $\rho$, while $\ix{a} := a_1 \dots a_{ l_2 }$.

Then, the following identities hold with respect any coordinates $( U, \varphi )$ on $\mi{I}$,
\begin{align}
\label{eq.aads_decomp} \nabla_\rho F_{ \ix{\rho} \ix{a} } &= \rho^{ - l_1 - l_2 } \Dvm_\rho ( \rho^{ l_1 + l_2 } \ms{f} )_{ \ix{a} } \text{,} \\
\notag \nabla_c F_{ \ix{\rho} \ix{a} } &= \Dvm_c \ms{f}_{ \ix{a} } + \rho^{-1} \sum_{ i = 1 }^{ l_1 } ( \ms{f}^\rho_i )_{ c \ix{a} } - \rho^{-1} \sum_{ j = 1 }^{ l_2 } \gv_{ c a_j } \, ( \ms{f}^v_j )_{ \ixd{a}{j} } \\
\notag &\qquad + \sum_{ i = 1 }^{ l_1 } \mi{S} ( \gv^{-1}, \Lv, \ms{f}^\rho_i )_{ c \ix{a} } + \sum_{ j = 1 }^{ l_2 } \mi{S} ( \Lv; \ms{f}^v_j )_{ c \ix{a} } \text{,} \\
\notag \Box F_{ \ix{\rho} \ix{a} } &= \rho^{ - l_1 - l_2 } \Boxm ( \rho^{ l_1 + l_2 } \ms{f} )_{ \ix{a} } + 2 \rho \sum_{ i = 1 }^{ l_1 } \gv^{ b c } \, \Dvm_b ( \ms{f}^\rho_i )_{ c \ix{a} } - 2 \rho \sum_{ j = 1 }^{ l_2 } \Dvm_{ a_j } ( \ms{f}^v_j )_{ \ixd{a}{j} } - ( n l_1 + l_2 ) \, \ms{f}_{ \ix{a} } \\
\notag &\qquad - 2 \sum_{ i = 1 }^{ l_1 } \sum_{ j = 1 }^{ l_2 } ( \ms{f}^{ \rho, v }_{ i, j } )_{ a_j \ixd{a}{j} } + 2 \sum_{ 1 \leq i < j \leq l_1 } \gv^{ b c } \, ( \ms{f}^{ \rho, \rho }_{ i, j } )_{ b c \ix{a} } + 2 \sum_{ 1 \leq i < j \leq l_2 } \gv_{ a_i a_j } \, ( \ms{f}^{ v, v }_{ i, j } )_{ \ixd{a}{i,j} } \\
\notag &\qquad + \rho^2 \sum_{ i = 1 }^{ l_1 } \mi{S} ( \gv^{-2}, \Lv, \Dvm \ms{f}^\rho_i )_{ \ix{a} } + \rho^2 \sum_{ j = 1 }^{ l_2 } \mi{S} ( \gv^{-1}, \Lv, \Dvm \ms{f}^v_j )_{ \ix{a} } \\
\notag &\qquad + \rho^2 \sum_{ i = 1 }^{ l_1 } \mi{S} ( \gv^{-2}, \Dv \Lv, \ms{f}^\rho_i )_{ \ix{a} } + \rho^2 \sum_{ j = 1 }^{ l_2 } \mi{S} ( \gv^{-1}, \Dv \Lv, \ms{f}^v_j )_{ \ix{a} } \\
\notag &\qquad + \rho \, \mi{S} ( \gv^{-1}, \Lv, \ms{f} )_{ \ix{a} } + \rho^2 \, \mi{S} ( \gv^{-2}, \Lv, \Lv, \ms{f} )_{ \ix{a} } \\
\notag &\qquad + \rho \sum_{ 1 \leq i < j \leq l_1 } \mi{S} ( \gv^{-2}, \Lv, \ms{f}^{ \rho, \rho }_{ i, j } )_{ \ix{a} } + \rho^2 \sum_{ 1 \leq i < j \leq l_1 } \mi{S} ( \gv^{-3}, \Lv, \Lv, \ms{f}^{ \rho, \rho }_{ i, j } )_{ \ix{a} } \\
\notag &\qquad + \rho \sum_{ 1 \leq i < j \leq l_2 } \mi{S} ( \gv, \gv^{-1}, \Lv, \ms{f}^{ v, v }_{ i, j } )_{ \ix{a} } + \rho^2 \sum_{ 1 \leq i < j \leq l_2 } \mi{S} ( \gv^{-1}, \Lv, \Lv, \ms{f}^{ v, v }_{ i, j } )_{ \ix{a} } \\
\notag &\qquad + \rho \sum_{ i = 1 }^{ l_1 } \sum_{ j = 1 }^{ l_2 } \mi{S} ( \gv^{-1}, \Lv, \ms{f}^{ \rho, v }_{ i, j } )_{ \ix{a} } + \rho \sum_{ i = 1 }^{ l_1 } \sum_{ j = 1 }^{ l_2 } \mi{S} ( \gv, \gv^{-2}, \Lv, \ms{f}^{ \rho, v }_{ i, j } )_{ \ix{a} } \\
\notag &\qquad + \rho^2 \sum_{ i = 1 }^{ l_1 } \sum_{ j = 1 }^{ l_2 } \mi{S} ( \gv^{-2}, \Lv, \Lv, \ms{f}^{ \rho, v }_{ i, j } )_{ \ix{a} } \text{.}
\end{align}
where $\ix{\rho}$ and $\ix{a}$ are as defined above, and where:
\begin{itemize}
\item For any $1 \leq i \leq l_1$, the rank $( 0, l_2 + 1 )$ vertical tensor field $\ms{f}^\rho_i$ is given by
\begin{equation}
\label{eq.aads_decomp_phi_rv} ( \ms{f}^\rho_i )_{ b \ix{a} } := F_{ \ixr{\rho}{i}{b} \ix{a} } \text{.}
\end{equation}

\item For any $1 \leq j \leq l_2$, the rank $( 0, l_2 - 1 )$ vertical tensor field $\ms{f}^v_j$ is given by
\begin{equation}
\label{eq.aads_decomp_phi_vr} ( \ms{f}^v_j )_{ \ixd{a}{j} } := F_{ \ix{\rho} \ixr{a}{j}{\rho} } \text{.}
\end{equation}

\item For any $1 \leq i, j \leq l_1$ with $i \neq j$, the rank $( 0, l_2 + 2 )$ vertical field $\ms{f}^{ \rho, \rho }_{ i, j }$ is given by
\begin{equation}
\label{eq.aads_decomp_phi_rrvv} ( \ms{f}^{ \rho, \rho }_{ i, j } )_{ c b \ix{a} } := F_{ \ixr{\rho}{i,j}{c,b} \ix{a} } \text{.}
\end{equation}

\item For any $1 \leq i, j \leq l_2$ with $i \neq j$, the rank $( 0, l_2 - 2 )$ vertical field $\ms{f}^{ v, v }_{ i, j }$ is given by
\begin{equation}
\label{eq.aads_decomp_phi_vvrr} ( \ms{f}^{ v, v }_{ i, j } )_{ \ixd{a}{i,j} } := F_{ \ix{\rho} \ixr{a}{i,j}{\rho,\rho} } \text{.}
\end{equation}

\item For any $1 \leq i \leq l_1$ and $1 \leq j \leq l_2$, the rank $( 0, l_2 )$ vertical field $\ms{f}^{ \rho, v }_{ i, j }$ is given by
\begin{equation}
\label{eq.aads_decomp_phi_rvvr} ( \ms{f}^{ \rho, v }_{ i, j } )_{ b \ixd{a}{j} } := F_{ \ixr{\rho}{i}{b} \ixr{a}{j}{\rho} } \text{.}
\end{equation}
\end{itemize}
\end{proposition}

\begin{proof}
See Appendix \ref{sec.aads_decomp}.
\end{proof}

%% file: system.tex
In this section, we establish various geometric identities relating the metrics and curvatures of vacuum FG-aAdS segments.
We then apply these identities in order to derive a system of wave and transport equations that are satisfied by the \emph{difference} of two vacuum FG-aAdS geometries.
This wave-transport system will be central to the proof of our main results.

\subsection{The Structure Equations}

We now consider several identities connecting different geometric quantities in vacuum aAdS spacetimes.
We begin by defining vertical tensor fields that capture the nontrivial components of the spacetime Weyl curvature:

\begin{definition} \label{def.aads_weyl}
Let $( \mi{M}, g )$ be a vacuum FG-aAdS segment.
We then define vertical tensor fields $\wv^0$, $\wv^1$, and $\wv^2$---of ranks $( 0, 4 )$, $( 0, 3 )$, and $( 0, 2 )$, respectively---by the formulas
\begin{equation}
\label{eq.aads_weyl} \wv^0_{ a b c d } := \rho^2 \, W_{ a b c d } \text{,} \qquad \wv^1_{ a b c } := \rho^2 \, W_{ \rho a b c } \text{,} \qquad \wv^2_{ a b } := \rho^2 \, W_{ \rho a \rho b } \text{.}
\end{equation}
In addition, when $n > 2$, we let $\wv^\star$ denote the rank $( 0, 4 )$ vertical tensor field defined as
\begin{equation}
\label{eq.aads_wstar} \wv^\star_{ a b c d } := \wv^0_{ a b c d } - \frac{1}{ n - 2 } ( \gv_{ a d } \wv^2_{ b c } + \gv_{ b c } \wv^2_{ a d } - \gv_{ a c } \wv^2_{ b d } - \gv_{ b d } \wv^2_{ a c } ) \text{.}
\end{equation}
In both \eqref{eq.aads_weyl} and \eqref{eq.aads_wstar}, the indices are respect to arbitrary coordinates $( U, \varphi )$ (and $\varphi_\rho$) on $\mi{I}$.
\end{definition}

\begin{remark}
The reason for the renormalization \eqref{eq.aads_wstar} is that $\wv^\star$ satisfies a tensorial wave equation (see \eqref{eq.aads_wave}) that can be treated with our Carleman estimates, whereas $\wv^0$ does not.
\end{remark}

The following three identities, derived in \cite{shao:aads_fg}, relate the spacetime Weyl curvature (expressed in terms of \eqref{eq.aads_weyl}) to the vertical metric and its derivatives.

\begin{proposition} \label{thm.aads_connection}
Let $( \mi{M}, g )$ be a vacuum FG-aAdS segment.
Then, the following identities hold with respect to an arbitrary coordinate system $( U, \varphi )$ on $\mi{I}$:
\begin{align}
\label{eq.aads_connection} \wv^1_{ c a b } &= \frac{1}{2} \Dv_b \Lv_{ a c } - \frac{1}{2} \Dv_a \Lv_{ b c } \text{,} \\
\notag \wv^0_{ a b c d } &= \Rv_{ a b c d } + \frac{1}{4} ( \Lv_{ a d } \Lv_{ b c } - \Lv_{ b d } \Lv_{ a c } ) + \frac{1}{2} \rho^{-1} ( \gv_{ b d } \Lv_{ a c } + \gv_{ a c } \Lv_{ b d } - \gv_{ a d } \Lv_{ b c } - \gv_{ b c } \Lv_{ a d } ) \text{,} \\
\notag \wv^2_{ a b } &= - \frac{1}{2} \mi{L}_\rho \Lv_{ a b } + \frac{1}{2} \rho^{-1} \Lv_{ a b } + \frac{1}{4} \gv^{ c d } \Lv_{ a d } \Lv_{ b c } \text{.}
\end{align}
\end{proposition}

\begin{proof}
See \cite[Proposition 2.25]{shao:aads_fg} and \eqref{eq.aads_sff}.
\end{proof}

Next, we derive identities satisfied by the spacetime Weyl curvature itself.
First, we recall the more familiar formulas in terms of spacetime tensor fields:

\begin{proposition} \label{thm.aads_weyl_pre}
Let $( \mi{M}, g )$ be a vacuum FG-aAdS segment.
Then, the following identities hold for the spacetime Weyl curvature $W$, with respect to any coordinates on $\mi{M}$:
\begin{align}
\label{eq.aads_weyl_pre} g^{ \mu \nu } \nabla_\mu W_{ \nu \beta \gamma \delta } &= 0 \text{,} \\
\notag ( \Box + 2 n ) W_{ \alpha \beta \gamma \delta } &= g^{ \lambda \kappa } g^{ \mu \nu } ( 2 W_{ \lambda \alpha \mu \delta } W_{ \kappa \beta \nu \gamma } - 2 W_{ \lambda \alpha \mu \gamma } W_{ \kappa \beta \nu \delta } - W_{ \lambda \mu \gamma \delta } W_{ \alpha \beta \kappa \nu } ) \text{.}
\end{align}
\end{proposition}

\begin{proof}
See Appendix \ref{sec.aads_weyl_pre}.
\end{proof}

The next step is to reformulate Proposition \ref{thm.aads_weyl_pre} in terms of the corresponding vertical quantities $\ms{w}^0$, $\ms{w}^1$, $\ms{w}^2$.
We begin with the \emph{vertical Bianchi identities} for $\ms{w}^0$, $\ms{w}^1$, and $\ms{w}^2$:

\begin{proposition} \label{thm.aads_bianchi}
Let $( \mi{M}, g )$ be a vacuum FG-aAdS segment.
Then, the following vertical Bianchi identities hold with respect to any coordinates $( U, \varphi )$ on $\mi{I}$:
\begin{align}
\label{eq.aads_bianchi} \Dvm_\rho \wv^2_{ a b } &= \gv^{ c d } \Dv_c \wv^1_{ b a d } + ( n - 2 ) \rho^{-1} \wv^2_{ a b } + \sch{ \gv^{-2}, \Lv, \wv^0 }_{ a b } + \sch{ \gv^{-1}, \Lv, \wv^2 }_{ a b } \text{,} \\
\notag \Dvm_\rho \wv^1_{ a b c } &= - \gv^{ d e } \Dv_d \wv^0_{ e a b c } + ( n - 2 ) \rho^{-1} \wv^1_{ a b c } + \sch{ \gv^{-1}, \Lv, \wv^1 }_{ a b c } \text{,} \\
\notag \Dvm_\rho \wv^1_{ a b c } &= \Dv_b \wv^2_{ a c } - \Dv_c \wv^2_{ a b } + \rho^{-1} \wv^1_{ a b c } + \sch{ \gv^{-1}, \Lv, \wv^1 }_{ a b c } \text{,} \\
\notag \Dvm_\rho \wv^0_{ a b c d } &= \Dv_a \wv^1_{ b c d } - \Dv_b \wv^1_{ a c d } + \rho^{-1} \gv_{ a d } \wv^2_{ b c } + \rho^{-1} \gv_{ b c } \wv^2_{ a d } - \rho^{-1} \gv_{ a c } \wv^2_{ b d } - \rho^{-1} \gv_{ b d } \wv^2_{ a c } \\
\notag &\qquad + \sch{ \gv^{-1}, \Lv, \wv^0 }_{ a b c d } + \sch{ \Lv, \wv^2 }_{ a b c d } \text{.}
\end{align}
\end{proposition}

\begin{proof}
See Appendix \ref{sec.aads_bianchi}.
\end{proof}

In the following, we derive the wave equations satisfied by $\wv^\star$, $\wv^1$, and $\wv^2$:

\begin{proposition} \label{thm.aads_wave}
Let $( \mi{M}, g )$ be a vacuum FG-aAdS segment, and let $n > 2$.
Then,\footnote{The assumption $n > 2$ is needed only to make sense of $\wv^\star$; the first two parts of \eqref{eq.aads_wave} also hold when $n = 2$.}
\begin{align} 
\label{eq.aads_wave} \Boxm \wv^2 + 2 ( n - 2 ) \wv^2 &= \rho^2 \, \sch{ \gv^{-2}, \Lv, \Dv \wv^1 } + \rho^2 \, \sch{ \gv^{-2}, \Dv \Lv, \wv^1 } + \rho \, \mi{S} ( \gv^{-1}, \Lv, \wv^2 ) \\
\notag &\qquad + \rho^2 \, \sch{ \gv^{-2}, \Lv, \Lv, \wv^2 } + \rho \, \sch{ \gv^{-2}, \Lv, \wv^0 } \\
\notag &\qquad + \rho^2 \, \sch{ \gv^{-3}, \Lv, \Lv, \wv^0 } + \rho \, \sch{ \gv, \gv^{-2}, \Lv, \wv^2 } \\
\notag &\qquad + \rho^2 \, \sch{ \gv^{-1}, \wv^2, \wv^2 } + \rho^2 \, \sch{ \gv^{-2}, \wv^1, \wv^1 } + \rho^2 \, \sch{ \gv^{-2}, \wv^0, \wv^2 } \text{,} \\
\notag \Boxm \wv^1 + ( n - 1 ) \wv^1 &= \rho^2 \, \sch{ \gv^{-2}, \Lv, \Dv \wv^0 } + \rho^2 \, \sch{ \gv^{-1}, \Lv, \Dv \wv^2 } + \rho^2 \, \sch{ \gv^{-2}, \Dv \Lv, \wv^0 } \\
\notag &\qquad + \rho^2 \, \sch{ \gv^{-1}, \Dv \Lv, \wv^2 } + \rho \, \sch{ \gv^{-1}, \Lv, \wv^1 } + \rho^2 \, \sch{ \gv^{-2}, \Lv, \Lv, \wv^1 } \\
\notag &\qquad + \rho \, \sch{ \gv, \gv^{-2}, \Lv, \wv^1 } + \rho^2 \, \sch{ \gv^{-1}, \wv^1, \wv^2 } + \rho^2 \, \sch{ \gv^{-2}, \wv^0, \wv^1 } \text{,} \\
\notag \Boxm \wv^\star &= \rho^2 \, \sch{ \gv^{-1}, \Lv, \Dv \wv^1 } + \rho^2 \, \sch{ \gv, \gv^{-2}, \Lv, \Dv \wv^1 } + \rho^2 \, \sch{ \gv^{-1}, \Dv \Lv, \wv^1 } \\
\notag &\qquad + \rho^2 \, \sch{ \gv, \gv^{-2}, \Dv \Lv, \wv^1 } + \rho \, \sch{ \gv^{-1}, \Lv, \wv^0 } + \rho \, \sch{ \gv, \gv^{-1}, \Lv, \wv^2 } \\
\notag &\qquad + \rho^2 \, \sch{ \gv, \gv^{-2}, \Lv, \Lv, \wv^2 } + \rho \, \sch{ \gv, \gv^{-2}, \Lv, \wv^0 } \\
\notag &\qquad + \rho^2 \, \sch{ \gv^{-2}, \Lv, \Lv, \wv^0 } + \rho^2 \, \sch{ \gv, \gv^{-3}, \Lv, \Lv, \wv^0 } \\
\notag &\qquad + \rho \, \sch{ \gv^2, \gv^{-2}, \Lv, \wv^2 } + \rho^2 \, \sch{ \gv^{-1}, \Lv, \Lv, \wv^2 } + \rho^2 \, \sch{ \wv^2, \wv^2 } \\
\notag &\qquad + \rho^2 \, \sch{ \gv, \gv^{-1}, \wv^2, \wv^2 } + \rho^2 \, \sch{ \gv^{-1}, \wv^1, \wv^1 } + \rho^2 \, \sch{ \gv^{-2}, \wv^0, \wv^0 } \\
\notag &\qquad + \rho^2 \, \sch{ \gv, \gv^{-2}, \wv^1, \wv^1 } + \rho^2 \, \sch{ \gv, \gv^{-2}, \wv^0, \wv^2 } \text{.}
\end{align}
\end{proposition}

\begin{proof}
See Appendix \ref{sec.aads_wave}.
\end{proof}

Finally, we list some asymptotic bounds for various geometric quantities.
To make these easier to state, we construct the following notations, which were also used in \cite{mcgill:loc_ads}.

\begin{definition} \label{def.aads_O}
Let $( \mi{M}, g )$ be an FG-aAdS segment, fix an integer $M \geq 0$, and let $h \in C^\infty ( \mi{M} )$.
\begin{itemize}
\item We use the notation $\oo{M}{ h }$ to refer to any vertical tensor field $\ms{a}$ satisfying the following bound for any compact coordinate system $( U, \varphi )$ on $\mi{I}$:
\begin{equation}
\label{eq.aads_O} | \ms{a} |_{ M, \varphi } \lesssim_\varphi h \text{.}
\end{equation}

\item Given a vertical tensor field $\ms{B}$, we use the notation $\oo{M}{ h; \ms{B} }$ to refer to any vertical tensor field of the form $\sch{ \ms{e}, \ms{B} }$, where $\ms{e}$ is a vertical tensor field satisfying $\ms{e} = \oo{M}{ h }$.
\end{itemize}
\end{definition}

\begin{proposition} \label{thm.aads_O_regular}
Suppose $( \mi{M}, g )$ is a vacuum FG-aAdS segment, and assume $( \mi{M}, g )$ is regular to order $M \geq n + 2$.\footnote{See Definition \ref{def.aads_regular}.}
Then, the following properties hold for $\gv$ and $\Lv$:
\begin{equation}
\label{eq.aads_O_g} \gv = \oo{M}{1} \text{,} \qquad \gv^{-1} = \oo{M}{1} \text{,} \qquad \Lv = \oo{ M - 1 }{1} = \oo{ M - 2 }{ \rho } \text{,} \qquad \mi{L}_\rho \Lv = \oo{ M - 2 }{1} \text{.}
\end{equation}
Furthermore, we have the following properties for $\wv^0$, $\wv^1$, and $\wv^2$:
\begin{align}
\label{eq.aads_O_w} \wv^0 = \oo{ M - 2 }{1} \text{,} &\qquad \mi{L}_\rho \wv^0 = \oo{ M - 3 }{1} \text{,} \\
\notag \wv^1 = \oo{ M - 2 }{1} = \oo{ M - 3 }{ \rho } \text{,} &\qquad \mi{L}_\rho \wv^1 = \oo{ M - 3 }{1} \text{,} \\
\notag \wv^2 = \oo{ M - 2 }{1} = \oo{ M - 3 }{ \rho } \text{,} &\qquad \mi{L}_\rho \wv^2 = \oo{ M - 3 }{1} \text{,} \\
\notag \wv^\star = \oo{ M - 2 }{1} \text{,} &\qquad \mi{L}_\rho \wv^\star = \oo{ M - 3 }{1} \text{.}
\end{align}
\end{proposition}

\begin{proof}
See Appendix \ref{sec.aads_O_regular}.
\end{proof}

\subsection{Difference Relations} \label{sec.sys_diff}

We now consider two aAdS geometries on a common manifold, and we derive equations relating quantities representing the \emph{difference} between two geometries.
More specifically, we consider in this subsection \emph{two} vacuum aAdS metrics on a common manifold---that is, we consider two vacuum FG-aAdS segments $( \mi{M}, g )$ and $( \mi{M}, \check{g} )$, with
\begin{equation}
\label{metricsinsameform} \mi{M} := ( 0, \rho_0 ] \times \mi{I} \text{,} \qquad g := \rho^{-2} ( d \rho^2 + \gv ) \text{,} \qquad \check{g} := \rho^{-2} ( d \rho^2 + \check{\gv} ) \text{.}
\end{equation}
Note that $g$ and $\check{g}$ live on a common manifold $\mi{M}$, and they share a common ``radial" variable $\rho$ that is used for the Fefferman-Graham expansion with respect to both $g$ and $\check{g}$.

To simplify matters, we will adopt the conventions for describing two geometries:

\begin{definition} \label{def.sys_notation}
Let $( \mi{M}, g )$ and $( \mi{M}, \check{g} )$ denote two vacuum FG-aAdS segments, as in \eqref{metricsinsameform}.
\begin{itemize}
\item We use the notations described in Section \ref{sec.aads} for both $( \mi{M}, g )$ and $( \mi{M}, \check{g} )$.

\item In particular, all objects associated with $\check{g}$ will use the usual notation, but with a ``check" added above the symbol (e.g.~$\gv$ and $\check{\gv}$ for the associated vertical metrics).
\end{itemize}
\end{definition}

\begin{remark}
Since $( \mi{M}, g )$, $( \mi{M}, \check{g} )$ in Definition \ref{def.sys_notation} have the same vertical bundles, the calculus for vertical tensors developed in Section \ref{sec.aads} also applies to differences of corresponding geometric quantities, e.g.\ $\gv - \check{\gv}$ and $\Lv - \check{\Lv}$.
In particular, all the equations in Propositions \ref{thm.aads_connection}, \ref{thm.aads_bianchi}, \ref{thm.aads_wave} hold for both the $g$- and $\check{g}$-geometries, hence we can consider differences of all these identities. 
\end{remark}

We now derive a closed system of wave and transport equations for the difference between two aAdS geometries.
Later, we will show in Section \ref{sec.proof} that the proof of Theorem \ref{theo:adscft} reduces precisely to unique continuation results, and hence Carleman estimates, for this coupled system.

To obtain this system, it will be convenient to have some additional auxiliary quantities:

\begin{definition} \label{def.sys_auxiliary}
Let $( \mi{M}, g )$ and $( \mi{M}, \check{g} )$ denote two vacuum FG-aAdS segments.
\begin{itemize}
\item We define the rank $( 0, 2 )$ vertical tensor field $\ms{Q}$ to be the solution of the transport equation
\begin{equation}
\label{Qdef} \mi{L}_\rho \ms{Q}_{ca} = \frac{1}{2} \gv^{de} [ \Lv_{ad} ( \gv - \check{\gv} + \ms{Q} )_{ce} - \Lv_{cd} ( \gv - \check{\gv} + \ms{Q} )_{ae} ] \text{,} \qquad \ms{Q} \rightarrow^0 0 \text{.}
\end{equation}

\item We define the rank $( 0, 3 )$ vertical tensor field $\ms{B}$ by
\begin{equation}
\label{Bdef} \ms{B}_{cab} := \Dv_c ( \gv - \check{\gv} )_{ab} - \Dv_a ( \gv - \check{\gv} )_{cb} - \Dv_b \ms{Q}_{ca}
\end{equation}

\item We define the vertical tensor fields $\ms{W}^2, \ms{W}^1, \ms{W}^\star$---of ranks $( 0, 2 )$, $( 0, 3 )$, $( 0, 4 )$, respectively---as follows:\ given a multi-index $\bar{a}$ of appropriate length, we set
\begin{align}
\label{Wdef} \ms{W}^2_{ \bar{a} } := ( \wv^2 - \check{\wv}^2 )_{ \bar{a} } + \gv^{bc} \sum_{ j = 1 }^2 \check{\wv}^2_{ \bar{a}_j [b] } ( \gv - \check{\gv} + \ms{Q} )_{ a_j c } \text{,} &\qquad \bar{a} := ( a_1 a_2 ) \text{,} \\
\notag \ms{W}^1_{ \bar{a} } := ( \wv^1 - \check{\wv}^1 )_{ \bar{a} } + \gv^{bc} \sum_{ j = 1 }^3 \check{\wv}^1_{ \bar{a}_j [b] } ( \gv - \check{\gv} + \ms{Q} )_{ a_j c } \text{,} &\qquad \bar{a} := ( a_1 a_2 a_3 ) \text{,} \\
\notag \ms{W}^\star_{ \bar{a} } := ( \wv^\star - \check{\wv}^\star )_{ \bar{a} } + \gv^{bc} \sum_{ j = 1 }^4 \check{\wv}^\star_{ \bar{a}_j [b] } ( \gv - \check{\gv} + \ms{Q} )_{ a_j c } \text{,} &\qquad \bar{a} := ( a_1 a_2 a_3 a_4 ) \text{.}
\end{align}
\end{itemize}
In the above, all indices are with respect to any arbitrary coordinate system $( U, \varphi )$ on $\mi{I}$.
\end{definition}

\begin{proposition} \label{thm.sys_O_QB}
Let $( \mi{M}, g )$ and $( \mi{M}, \check{g} )$ denote two vacuum FG-aAdS segments.
Then:
\begin{itemize}
\item The following asymptotic relations hold (with $\ms{Q}$, $\ms{B}$ be as in Definition \ref{def.sys_auxiliary}):
\begin{equation}
\label{eq.sys_O_QB} \gv^{-1} - \check{\gv}^{-1} = \oo{ M }{ 1; \gv - \check{\gv} } \text{,} \qquad \ms{Q} = \oo{ M - 1 }{ \rho } = \oo{ M - 2 }{ \rho^2 } \text{,} \qquad \ms{B} = \oo{ M - 2 }{1} \text{.}
\end{equation}

\item The following hold with respect to any coordinate system $( U, \varphi )$ on $\mi{I}$, where the Christoffel symbols $\Gamma$, $\check{\Gamma}$, $\Gammav$, $\check{\Gammav}$ are defined as in Proposition \ref{thm.aads_Gamma}:
\begin{align}
\label{eq.sys_diff_Gamma} ( \Gamma - \check{\Gamma} )^\alpha_{ \rho \rho } = ( \Gamma - \check{\Gamma} )^\rho_{ \alpha \rho } &= 0 \text{,} \\
\notag ( \Gamma - \check{\Gamma} )^\rho_{ a b } &= \rho^{-1} ( \gv - \check{\gv} )_{ab} + \frac{1}{2} ( \Lv - \check{\Lv} )_{ab} \text{,} \\
\notag ( \Gamma - \check{\Gamma} )^a_{ \rho b } = ( \Gammav - \check{\Gammav} )^a_{ \rho b } &= \oo{ M - 2 }{ \rho; \gv - \check{\gv} }^a{}_b + \oo{M}{ 1; \Lv - \check{\Lv} }^a{}_b \text{,} \\
\notag ( \Gamma - \check{\Gamma} )^c_{ a b } = ( \Gammav - \check{\Gammav} )^c_{ a b } &= \frac{1}{2} \check{\gv}^{cd} [ \Dv_a ( \gv - \check{\gv} )_{db} + \Dv_b ( \gv - \check{\gv} )_{da} - \Dv_d ( \gv - \check{\gv} )_{ab} ] \\
\notag &= \oo{M}{ 1; \Dv ( \gv - \check{\gv} ) }^c_{ab} \text{.}
\end{align}
\end{itemize}
\end{proposition}

\begin{proof}
See Appendix \ref{sec.sys_O_QB}.
\end{proof}

The following two propositions contain our main wave-transport system.
In particular, the key step in the proof of our main result is a unique continuation result on this system.

\begin{proposition} \label{thm.sys_transport}
Let $( \mi{M}, g )$ and $( \mi{M}, \check{g} )$ denote two vacuum FG-aAdS segments.
Then,
\begin{align} 
\label{eq.sys_transport} \mi{L}_\rho ( \gv - \check{\gv} ) &= \Lv - \check{\Lv} \text{,} \\
\notag \mi{L}_\rho \ms{Q} &= \oo{M-2}{ \rho; \gv - \check{\gv} } + \oo{M-2}{ \rho; \ms{Q} } \text{,} \\
\notag \mi{L}_\rho [ \rho^{-1} ( \Lv - \check{\Lv} ) ] &= - 2 \rho^{-1} \ms{W}^2 + \oo{M-3}{ 1; \gv - \check{\gv} } + \oo{M-3}{ 1; \ms{Q} } + \oo{M-2}{ 1; \Lv - \check{\Lv} } \text{,} \\
\notag \mi{L}_\rho \ms{B} &= \sch{ \ms{W}^1 } + \oo{M-3}{ \rho; \gv - \check{\gv} } + \oo{M-3}{ \rho; \ms{Q} } + \oo{ M - 1 }{ 1; \Lv - \check{\Lv} } + \oo{ M - 2 }{ \rho; \ms{B} } \text{.}
\end{align}
In addition, the following derivative transport equations hold:
\begin{align}
\label{eq.sys_transport_deriv} \mi{L}_\rho \Dv ( \gv - \check{\gv} ) &= \Dv ( \Lv - \check{\Lv} ) + \oo{ M - 3 }{ \rho; \gv - \check{\gv} } \text{,} \\
\notag \mi{L}_\rho \Dv \ms{Q} &= \oo{M-2}{ \rho; \Dv ( \gv - \check{\gv} ) } + \oo{M-2}{ \rho; \Dv \ms{Q} } + \oo{M-3}{ \rho; \gv - \check{\gv} } + \oo{M-3}{ \rho; \ms{Q} } \text{,} \\
\notag \mi{L}_\rho \Dv [ \rho^{-1} ( \Lv - \check{\Lv} ) ] &= - 2 \rho^{-1} \Dv \ms{W}^2 + \oo{M-3}{ 1; \Dv ( \gv - \check{\gv} ) } + \oo{M-3}{ 1; \Dv \ms{Q} } + \oo{M-2}{ 1; \Dv ( \Lv - \check{\Lv} ) } \\
\notag &\qquad + \oo{M-4}{ 1; \gv - \check{\gv} } + \oo{M-4}{ 1; \ms{Q} } + \oo{M-3}{ 1; \Lv - \check{\Lv} } \text{,} \\
\notag \mi{L}_\rho \Dv \ms{B} &= \sch{ \Dv \ms{W}^1 } + \oo{M-3}{ \rho; \Dv ( \gv - \check{\gv} ) } + \oo{M-3}{ \rho; \Dv \ms{Q} } + \oo{ M - 1 }{ 1; \Dv ( \Lv - \check{\Lv} ) } \\
\notag &\qquad + \oo{ M - 2 }{ \rho; \Dv \ms{B} } + \oo{M-4}{ \rho; \gv - \check{\gv} } + \oo{M-4}{ \rho; \ms{Q} } \\
\notag &\qquad + \oo{ M - 2 }{ 1; \Lv - \check{\Lv} } + \oo{ M - 3 }{ \rho; \ms{B} } \text{.}
\end{align}
\end{proposition}

\begin{proof}
See Appendix \ref{sec.sys_transport}.
\end{proof}

\begin{proposition} \label{thm.sys_wave}
Let $( \mi{M}, g )$ and $( \mi{M}, \check{g} )$ denote two vacuum FG-aAdS segments.
Then,
\begin{align}
\label{eq.sys_wave} \Boxm \ms{W}^2 + 2(n-2) \ms{W}^2 &= \ms{F}^2 \text{,} \\
\notag \Boxm \ms{W}^1 + (n-1) \ms{W}^1 &= \ms{F}^1\text{,} \\
\notag \Boxm \ms{W}^\star &= \ms{F}^\star \text{,}
\end{align}
where each $\ms{F} \in \{ \ms{F}^2, \ms{F}^1, \ms{F}^\star \}$ is schematically of the form
\begin{align}
\label{Fdef} \ms{F} &= \oo{ M - 4 }{ \rho^2; \gv - \check{\gv} }_{ \bar{a} } + \oo{ M - 3 }{ \rho^2; \ms{Q} } + \oo{ M - 3 }{ \rho; \Lv - \check{\Lv} }_{ \bar{a} } \\
\notag &\qquad + \oo{M-3}{ \rho^2; \Dv ( \gv - \check{\gv} ) }_{ \bar{a} } + \oo{M-3}{ \rho^2; \Dv \ms{Q} }_{ \bar{a} } + \oo{ M - 2 }{ \rho^2; \Dv ( \Lv - \check{\Lv} ) }_{ \bar{a} } \\
\notag &\qquad + \oo{M-2}{ \rho^2; \Dv \ms{B} }_{ \bar{a} } + \sum_{ \ms{V} \in \{ \ms{W}^\star, \ms{W}^1, \ms{W}^2 \} } [ \oo{ M - 3 }{ \rho^2; \ms{V} } + \oo{ M - 2 }{ \rho^3; \Dv \ms{V} } ]_{ \bar{a} } \text{.}
\end{align}
\end{proposition}

\begin{proof}
See Appendix \ref{sec.sys_wave}.
\end{proof}

Finally, we collect here convenient forms for the differences of the Bianchi equations \eqref{eq.aads_bianchi}.
These will be needed in another step in the proof of our main result---showing that the quantities relating to the differences of two geometries vanish to arbitrarily high order.

\begin{proposition} \label{thm.sys_vanish}
Let $( \mi{M}, g )$ and $( \mi{M}, \check{g} )$ denote two vacuum FG-aAdS segments.
Then,
\begin{align}
\label{eq.sys_vanish} \mi{L}_\rho ( \gv - \check{\gv} ) &= \Lv - \check{\Lv} \text{,} \\
\notag \mi{L}_\rho [ \rho^{-1} ( \Lv - \check{\Lv} ) ] &= - 2 \rho^{-1} ( \wv^2 - \check{\wv}^2 ) + \mi{O}_{ M - 2 } ( \rho; \gv - \check{\gv} ) + \oo{ M - 2 }{ 1; \Lv - \check{\Lv} } \text{,} \\
\notag \mi{L}_\rho [ \rho^{2-n} ( \wv^2 - \check{\wv}^2 ) ] &= \oo{M}{ \rho^{2-n}; \Dv ( \wv^1 - \check{\wv}^1 ) } + \mi{O}_{ M - 4 } ( \rho^{3-n}; \gv - \check{\gv} ) \\
\notag &\qquad + \oo{ M - 2 }{ \rho^{2-n}; \Lv - \check{\Lv} } + \oo{ M - 3 }{ \rho^{3-n}; \Dv ( \gv - \check{\gv} ) } \\
\notag &\qquad + \oo{ M - 2 }{ \rho^{3-n}; \wv^0 - \check{\wv}^0 } + \oo{ M - 2 }{ \rho^{3-n}; \wv^2 - \check{\wv}^2 } \text{,} \\ 
\notag \mi{L}_\rho [ \rho^{-1} ( \wv^1 - \check{\wv}^1 ) ] &= \rho^{-1} \, \sch{ \Dv ( \wv^2 - \check{\wv}^2 ) } + \oo{ M - 3 }{ 1; \gv - \check{\gv} } \\
\notag &\qquad + \oo{ M - 2 }{ \rho^{-1}; \Lv - \check{\Lv} } + \oo{ M - 3 }{ 1; \Dv ( \gv - \check{\gv} ) } \\
\notag &\qquad + \oo{ M - 1 }{ 1; \wv^1 - \check{\wv}^1 } \text{,} \\
\notag \mi{L}_\rho ( \wv^0 - \check{\wv}^0 ) &= \sch{ \Dv ( \wv^1 - \check{\wv}^1 ) } + \oo{ M - 2 }{ \rho^{-1}; \wv^2 - \check{\wv}^2 } \\
\notag &\qquad + \oo{ M - 3 }{ 1; \gv - \check{\gv} } + \oo{ M - 2 }{ 1; \Lv - \check{\Lv} } \\
\notag &\qquad + \oo{ M - 3 }{ \rho; \Dv ( \gv - \check{\gv} ) } + \oo{ M - 2 }{ \rho; \wv^0 - \check{\wv}^0 } \text{.} 
\end{align}
\end{proposition}

\begin{proof}
See Appendix \ref{sec.sys_vanish}.
\end{proof}

\begin{remark}
The system \eqref{eq.sys_vanish} is sufficient to obtain higher-order vanishing of the differences of geometric quantities at the boundary (see Proposition \ref{thm.proof_vanish}).
However, we will need to work with the larger system \eqref{eq.sys_transport}--\eqref{eq.sys_wave}---containing renormalized quantities $\ms{Q}$, $\ms{B}$, $\ms{W}^\star$, $\ms{W}^1$, $\ms{W}^2$---in order to close the Carleman estimates, within which we cannot afford a loss in derivatives.
\end{remark}

%% file: carleman.tex
In this section, we state the two Carleman estimates for vertical tensor fields that constitute the main analytic ingredients for the proof of our main results.

\subsection{The Wave Carleman Estimate}

The first Carleman estimate we discuss is that for wave equations---namely, the main results obtained in \cite{ChatzikaleasShao, hol_shao:uc_ads, hol_shao:uc_ads_ns, mcgill_shao:psc_aads}.
We begin by discussing the best-known conditions needed on the conformal boundary for such an estimate to hold.

\begin{definition} \label{def.carleman_h}
Let $( \mi{M}, g )$ be an FG-aAdS segment, and let $\mf{h}$ be a Riemannian metric on $\mi{I}$.
\begin{itemize}
\item We can also view $\mf{h}$ as a $\rho$-independent vertical Riemannian metric.\footnote{See Remark \ref{rmk.aads_tensor}.}

\item For a vertical tensor field $\ms{A}$, we write $| \ms{A} |_{ \mf{h} }$ to denote its pointwise $\mf{h}$-norm.
In other words, if $\ms{A}$ has rank $( k, l )$, then with respect to any coordinate system $( U, \varphi )$ on $\mi{I}$, we have
\begin{equation}
\label{eq.carleman_norm} | \ms{A} |_{ \mf{h} }^2 = \mf{h}_{ a_1 c_1 } \dots \mf{h}_{ a_k c_k } \mf{h}^{ b_1 d_1 } \dots \mf{h}^{ b_l d_l } \ms{A}^{ a_1 \dots a_k }_{ b_1 \dots b_l } \ms{A}^{ c_1 \dots c_k }_{ d_1 \dots d_l } \text{.}
\end{equation}
\end{itemize}
\end{definition}

\begin{remark}
The metric $\mf{h}$ is only used as a coordinate-independent way to measure the sizes of vertical tensor fields.
Our main results will not depend on a particular choice of $\mf{h}$.
\end{remark}

\begin{definition}[Definition 3.1 of \cite{ChatzikaleasShao}] \label{def.carleman_gncc}
Let $( \mi{M}, g )$ be a vacuum FG-aAdS segment, let $\mf{h}$ be a Riemannian metric on $\mi{I}$, and let $\mi{D} \subset \mi{I}$ be open with compact closure.
We say $( \mi{D}, \mf{g} )$ satisfies the \emph{generalized null convexity criterion} (or \emph{GNCC}) iff there exist $\eta \in C^4 ( \bar{\mi{D}} )$ and $c > 0$ satisfying
\begin{align} \label{eq.carleman_gncc}
\begin{cases}
\left( \mf{D}^2 \eta + \frac{1}{n-2} \, \eta \, \Rcm \right) ( \mf{Z}, \mf{Z} ) > c \eta \, \mf{h} ( \mf{Z}, \mf{Z} ) &\text{in } \mi{D} \text{,} \\
\eta > 0 &\text{in } \mi{D} \text{,} \\
\eta = 0 &\text{on } \partial \mi{D} \text{,}
\end{cases} 
\end{align}
for all vectors fields $\mf{Z}$ on $\mi{D}$ satisfying $\mf{g} ( \mf{Z}, \mf{Z} ) = 0$.
\end{definition}

\begin{remark}
In Definition \ref{def.carleman_gncc}, we specialized to vacuum FG-aAdS segments.
However, the GNCC can be directly extended, as in \cite{ChatzikaleasShao}, to strongly FG-aAdS segments that are not necessarily vacuum.
In the more general setting, the Ricci curvature $\Rcm$ in \eqref{eq.carleman_gncc} is replaced by $-\gb{2}$.\footnote{Note that for vacuum FG-aAdS segments, \eqref{eq.aads_fg_schouten} implies $-\gb{2} ( \mf{Z}, \mf{Z} ) = \frac{1}{n-2} \, \Rcm ( \mf{Z}, \mf{Z} )$ for any $\mf{g}$-null vector field $\mf{Z}$.}
\end{remark}

Next, we recall some quantities that will be essential to our Carleman estimates:

\begin{definition} \label{def.carleman_f}
Assume the setting of Definition \ref{def.carleman_gncc}---in particular, let $( \mi{D}, \gm )$ satisfy the GNCC, with $\eta$ satisfying \eqref{eq.carleman_gncc}.
Within this setting, we define the following additional quantities:
\begin{itemize}
\item Let $f_\eta := f \in C^4 ( ( 0, \rho_0 ] \times \mi{D} )$ denote the function
\begin{equation}
\label{eq.carleman_f} f := \rho \eta^{-1} \text{.}
\end{equation}

\item In addition, given $f_\star > 0$, we define the domain
\begin{equation}
\label{eq.carleman_omega} \Omega_{ f_\star } := \{ z \in \mi{M} \mid f (z) < f_\star \} \text{.}
\end{equation}
\end{itemize}
\end{definition}

We can now state the precise form of the Carleman estimate for wave equations from \cite{ChatzikaleasShao}.
Here, to slightly simplify the presentation, we express this in a less general form than in \cite{ChatzikaleasShao}.

\begin{theorem}[Theorem 5.11 of \cite{ChatzikaleasShao}] \label{thm.carleman}
Let $( \mi{M}, g )$ be a vacuum FG-aAdS segment.
In addition:
\begin{itemize}
\item Let $\mf{h}$ be a Riemannian metric on $\mi{I}$, and $\mi{D} \subset \mi{I}$ be open with compact closure.

\item Assume $( \mi{D}, \gm )$ satisfies the GNCC, with $\eta \in C^4 ( \bar{\mi{D}} )$ as in \eqref{eq.carleman_gncc} and $\mf{h}$ as above.

\item Fix integers $k, l \geq 0$ and a constant $\sigma \in \R$.
\end{itemize}
Then, there exist $\mc{C}_0 \geq 0$ and $\mc{C}, \mc{C}_b > 0$ (depending on $\ms{g}$, $\mf{h}$, $\mi{D}$, $k$, $l$) such that for any $\kappa \in \mathbb{R}$ with 
\begin{equation}
\label{eq.carleman_ass_kappa} 2 \kappa \geq n - 1 + \mc{C}_0 \text{,} \qquad \kappa^2 - (n-2) \kappa + \sigma - (n-1) - \mc{C}_0 \geq 0 \text{,}
\end{equation}
and for any constants $f_\star, \lambda, p > 0$ with 
\begin{equation}
\label{eq.carleman_ass_params} 0 < f_\star \ll_{ \ms{g}, \mf{h}, \mi{D}, k, l } 1 \text{,} \qquad \lambda \gg_{ \ms{g}, \mf{h}, \mi{D}, k, l } |\kappa| + |\sigma| \text{,} \qquad 0 < 2 p < 1 \text{,}
\end{equation}
the following Carleman estimate holds for any vertical tensor field $\ms{\Phi}$ on $\mi{M}$ of rank $( k, l )$ such that both $\ms{\Phi}$ and $\smash{\nablam} \ms{\Phi}$ vanish identically on $\{ f = f_\star \}$:\footnote{For notational convenience, we replaced the parameter $\lambda$ in \cite{ChatzikaleasShao} by $\lambda/2$ here.}
\begin{align}
\label{eq.carleman} &\int_{ \Omega_{ f_\star } } e^{ -\lambda p^{-1} f^{p} } f^{ n - 2 - p - 2 \kappa } | ( \Boxm + \sigma ) \ms{\Phi} |_{ \mf{h} }^2 \, d \mu_g \\
\notag &\quad \qquad + \mc{C}_b \lambda^3 \limsup_{ \rho_\star \searrow 0 } \int_{ \Omega_{ f_\star } \cap \{ \rho = \rho_\star \} } [ | \Dvm_{ \partial_\rho } ( \rho^{-\kappa} \ms{\Phi} ) |_{ \mf{h} }^2 + | \ms{D} ( \rho^{-\kappa} \ms{\Phi} ) |_{ \mf{h} }^2 + | \rho^{-\kappa-1} \ms{\Phi} |_{ \mf{h} }^2 ] \, d \mu_{ \ms{g} } \\
\notag &\quad \geq \mc{C} \lambda \int_{ \Omega_{ f_\star } } e^{-\lambda p^{-1} f^{p}} f^{ n - 2 - 2 \kappa } ( \rho^4 | \Dvm_{ \partial_\rho } \ms{\Phi} |_{ \mf{h} }^2 + \rho^4 | \ms{D} \ms{\Phi} |_{ \mf{h} }^2 + f^{2p} | \ms{\Phi} |_{ \mf{h} }^2 ) \, d \mu_g \text{.}
\end{align}
Here, $d \mu_g$ denotes the volume form on $\mi{M}$ induced by the spacetime metric $g$, while $d \mu_{ \gv }$ denotes the volume forms on the level sets of $\rho$ induced by the vertical metric $\gv$.
\end{theorem}

\begin{remark}
We note that Theorem \ref{thm.carleman} only considers vacuum FG-aAdS segments, whereas the more general \cite[Theorem 5.11]{ChatzikaleasShao} also allows for some non-vacuum FG-aAdS segments (under a more general GNCC).
Moreover, \cite[Theorem 5.11]{ChatzikaleasShao} allows for an additional first-order term in the wave equation that is vanishing at a slower ``critical" rate toward the conformal boundary.\footnote{See the quantity $X$ in \cite[Theorem 5.11]{ChatzikaleasShao}; however, we will not need this extra generality here.}
\end{remark}

\begin{remark}
If $k = l = 0$ (i.e.\ $\ms{\Phi}$ is scalar), then Theorem \ref{thm.carleman} holds with $\mc{C}_0 = 0$; see \cite{ChatzikaleasShao}.
\end{remark}

\subsection{The Transport Carleman Estimate}

Next, we prove a simple Carleman estimate for transport equations, in the same setting and with the same weights as in Theorem \ref{thm.carleman}.
 
\begin{proposition} \label{thm.carleman_transport}
Let $( \mi{M}, g )$ be a vacuum FG-aAdS segment.
In addition:
\begin{itemize}
\item Let $\mf{h}$ be a Riemannian metric on $\mi{I}$, and $\mi{D} \subset \mi{I}$ be open with compact closure.

\item Assume $( \mi{D}, \gm )$ satisfies the GNCC, with $\eta \in C^4 ( \bar{\mi{D}} )$ as in \eqref{eq.carleman_gncc} and $\mf{h}$ as above.
\end{itemize}
Then, for any $s \geq 0$, $\kappa \in \R$, and $\lambda, f_\star, p > 0$ satisfying
\begin{equation}
\label{eq.carleman_transport_ass} 2 \kappa \geq \max ( n - 2, s - 3 ) \text{,} \qquad 0 < 2 p < 1 \text{,} \qquad 0 < f_\ast \ll_{ \gv, \mi{D} } 1 \text{,}
\end{equation}
there exist $\mc{C}', \mc{C}_b' > 0$ (depending on $\gv$, $\mf{h}$, $\mi{D}$) such that for every vertical tensor field $\Psi$ on $\mi{M}$,
\begin{align}
\label{eq.carleman_transport} &\int_{ \Omega_{ f_\star } } e^{ -\lambda p^{-1} f^p } f^{ n - 2 - p - 2 \kappa } \rho^{s+2} | \mi{L}_\rho \ms{\Psi} |_{ \mf{h} }^2 \, d \mu_g + \mc{C}_b' \lambda \limsup_{ \rho_\star \searrow 0 } \int_{ \Omega_{ f_\star } \cap \{ \rho = \rho_\ast \} } \rho^s | \rho^{ - \kappa - 1 } \ms{\Psi} |_{ \mf{h} }^2 \, d \mu_{ \gv } \\
\notag &\quad \geq \mc{C} \lambda \int_{ \Omega_{ f_\star } } e^{ -\lambda p^{-1} f^p } f^{ n - 2 - 2 \kappa } \rho^s | \ms{\Psi} |_{ \mf{h} }^2 \, d \mu_g \text{.}
\end{align}
\end{proposition}

\begin{proof}
Using that both $\mf{h}$ and $\eta$ are $\rho$-independent, and recalling \eqref{eq.carleman_f}, we obtain
\begin{align}
\label{eql.carleman_transport_a} &\mi{L}_\rho \left( e^{ - \frac{ \lambda f^p }{p} } f^{ n - 2 - 2 \kappa } \rho^{ s - n } | \ms{\Psi} |_{ \mf{h} }^2 \right) + ( 2 \kappa + 2 - s + \lambda f^p ) e^{ - \frac{ \lambda f^p }{p} } f^{ n - 2 - 2 \kappa } \rho^{ s - n - 1 } | \ms{\Psi} |_{ \mf{h} }^2 \\
\notag &\quad \leq e^{ - \frac{ \lambda f^p }{p} } f^{ n - 2 - 2 \kappa } \rho^{ s - n } | \ms{\Psi} |_{ \mf{h} } | \mi{L}_\rho \ms{\Psi} |_{ \mf{h} } \text{.}
\end{align}
Applying the Cauchy-Schwarz inequality to the right-hand side of the above and rearranging yields
\begin{align}
\label{eql.carleman_transport_b} ( 2 \kappa + 2 - s ) e^{ - \frac{ \lambda f^p }{p} } f^{ n - 2 - 2 \kappa } \rho^{ s - n - 1 } | \ms{\Psi} |_{ \mf{h} }^2 &\leq - \mi{L}_\rho \left( e^{ - \frac{ \lambda f^p }{p} } f^{ n - 2 - p - 2 \kappa } \rho^{ s - n } | \ms{\Psi} |_{ \mf{h} }^2 \right) \\
\notag &\qquad + \frac{1}{4} \lambda^{-1} e^{ - \frac{ \lambda f^p }{p} } f^{ n - 2 - p - 2 \kappa } \rho^{ s - n + 1 } | \mi{L}_\rho \ms{\Psi} |_{ \mf{h} }^2 \text{.}
\end{align}

We now integrate the above inequality over the region $\Omega_{ f_\star, \rho_\star } := \Omega_{ f_\star } \cap \{ \rho > \rho_\star \}$ for an arbitrary $0 < \rho_\star \ll f_\star$, first over level sets of $\rho$ with respect to the $\rho$-independent volume forms $d \mu_{ \mf{h} }$ induced by $\mf{h}$, and then over $\rho$.
This yields the estimate
\begin{align*}
&( 2 \kappa + 2 - s ) \lambda \int_{ \Omega_{ f_\star, \rho_\star } } e^{ - \frac{ \lambda f^p }{p} } f^{ n - 2 - 2 \kappa } \rho^{ s - n - 1 } | \ms{\Psi} |_{ \mf{h} }^2 \ d \mu_{ \mf{h} } d \rho \\
&\quad \leq \int_{ \Omega_{ f_\star, \rho_\star } } \left[ \frac{1}{4} e^{ - \frac{ \lambda f^p }{p} } f^{ n - 2 - p - 2 \kappa } \rho^{ s - n + 1 } | \mi{L}_\rho \ms{\Psi} |_{ \mf{h} }^2 - \lambda \mi{L}_\rho \left( e^{ - \frac{ \lambda f^p }{p} } f^{ n - 2 - 2 \kappa } \rho^{ s - n } | \ms{\Psi} |_{ \mf{h} }^2 \right) \right] d \mu_{ \mf{h} } d \rho \\
&\quad \leq \frac{1}{4} \int_{ \Omega_{ f_\star, \rho_\star } } e^{ - \frac{ \lambda f^p }{p} } f^{ n - 2 - p - 2 \kappa } \rho^{ s - n + 1 } | \mi{L}_\rho \ms{\Psi} |_{ \mf{h} }^2 \, d \mu_{ \mf{h} } d \rho + \lambda \int_{ \Omega_{ f_\star } \cap \{ \rho = \rho_\star \} } e^{ - \frac{ \lambda f^p }{p} } f^{ n - 2 - 2 \kappa } \rho^{ s - n } | \ms{\Psi} |_{ \mf{h} }^2 d \mu_{ \mf{h} } \text{,}
\end{align*}
where in the last step, we applied the fundamental theorem of calculus; note in particular that the ensuing boundary term on $\{ f = f_\star \}$ is negative and can hence be neglected.
The above now yields, for some constants $\mc{C}', \mc{C}_{ b, 0 }' > 0$ (depending on $\gv$, $\mf{h}$, $\mi{D}$),
\begin{align}
\label{eql.carleman_transport_0} \mc{C}' \lambda \int_{ \Omega_{ f_\star, \rho_\star } } e^{ - \frac{ \lambda f^p }{p} } f^{ n - 2 - 2 \kappa } \rho^s | \ms{\Psi} |_{ \mf{h} }^2 \ d \mu_g &\leq \frac{1}{4} \int_{ \Omega_{ f_\star, \rho_\star } } e^{ - \frac{ \lambda f^p }{p} } f^{ n - 2 - p - 2 \kappa } \rho^{ s + 2 } | \mi{L}_\rho \ms{\Psi} |_{ \mf{h} }^2 \, d \mu_g \\
\notag &\qquad + \mc{C}_{ b, 0 }' \lambda \int_{ \Omega_{ f_\star } \cap \{ \rho = \rho_\star \} } e^{ - \frac{ \lambda f^p }{p} } f^{ n - 2 - 2 \kappa } \rho^{ s - n } | \ms{\Psi} |_{ \mf{h} }^2 d \mu_{ \gv } \text{,}
\end{align}
since $d \mu_{ \mf{h} }$ and $d \mu_{ \gv }$ are comparable due to the compactness of $\bar{\mi{D}}$, and since \eqref{eq.aads_metric} implies
\[
d \mu_g = \rho^{ -n - 1 } \, d \mu_{ \gv } d \rho \text{.}
\]

Finally, letting $\rho_\star \searrow 0$ in \eqref{eql.carleman_transport_0} yields the desired \eqref{eq.carleman_transport}, since
\[
e^{ - \lambda p^{-1} f^p } \leq 1 \text{,} \qquad f^{ n - 2 - 2 \kappa } \rho^{ s - n } | \ms{\Psi} |_{ \mf{h} }^2 = \eta^{ 2 \kappa - ( n - 2 ) } \, \rho^s | \rho^{ - \kappa - 1 } \ms{\Psi} |_{ \mf{h} }^2 \text{,}
\]
and since the exponent $2 \kappa - ( n - 2 ) \geq 0$ by \eqref{eq.carleman_transport_ass}.
\end{proof}

\begin{remark}
Note that the GNCC is not required for Proposition \ref{thm.carleman_transport} and its proof.
However, it is convenient to include this in the statement Proposition \ref{thm.carleman_transport}, as the transport Carleman estimate will be applied concurrently with and in the same setting as the wave Carleman estimate.
\end{remark}

%% file: proof.tex
In this section, we turn toward establishing unique continuation for the Einstein-vacuum equations.
In particular, we prove the following theorem, which can be seen as a precise statement of most of Theorem \ref{theo:adscft}---the special case in which the pair of holographic data are equal.\footnote{The remainder of Theorem \ref{theo:adscft}, which deals with gauge covariance issues, is treated in Section \ref{sec.app_gauge} below.}

\begin{theorem} \label{thm.adscft}
Let $n > 2$, and let $( \mi{M}, g )$, $( \mi{M}, \check{g} )$ be vacuum FG-aAdS segments (on a common aAdS region $\mi{M} := ( 0, \rho_0 ] \times \mi{I}$),\footnote{This is the same setup as in Section \ref{sec.sys_diff}; the reader is referred there for further notational details.} with associated holographic data $( \mi{I}, \gb{0}, \gb{n} )$, $( \mi{I}, \gbc{0}, \gbc{n} )$ (respectively).
In addition, let $\mi{D} \subset \mi{I}$ be open with compact closure, and assume:
\begin{itemize}
\item $( \mi{M}, g )$ and $( \mi{M}, \check{g} )$ are regular to some large enough order $M_0$ (depending on $\gv$, $\check{\gv}$, $\mi{D}$).

\item The holographic data coincide on $\mi{D}$:
\begin{equation}
\label{eq.adscft_data} ( \gb{0}, \gb{n} ) |_{ \mi{D} } = ( \gbc{0}, \gbc{n} ) |_{ \mi{D} } \text{.}
\end{equation}

\item $( \mi{D}, \gb{0} )$ (or equivalently, $( \mi{D}, \gbc{0} )$) satisfies the GNCC.
\end{itemize}
Then, $g = \check{g}$ on a neighbourhood of $\{ 0 \} \times \mi{D}$ (viewed as part of the conformal boundary)---more specifically, there exists some sufficiently small $f_\star > 0$ such that $g = \check{g}$ on the region $\Omega_{ f_\star }$.\footnote{See Definition \ref{def.carleman_f} for the definition of $\Omega_{ f_\star }$.}
\end{theorem}

\begin{remark}
Note that Theorem \ref{thm.adscft} offers a more precise conclusion than Theorem \ref{theo:adscft}---in the special case \eqref{eq.adscft_data}, the isometry in Theorem \ref{theo:adscft} is simply the identity map.
\end{remark}

The proof of Theorem \ref{thm.adscft} is given in the remainder of this section.
Throughout, we assume the hypotheses of Theorem \ref{thm.adscft} hold.
Furthermore, we adopt the notational conventions of Section \ref{sec.sys_diff} regarding quantities defined with respect to $g$ and $\check{g}$.

\subsection{Deducing Higher-Order Vanishing} \label{sec:addvanish}

For convenience, we can assume, without loss of generality, that $M_0 - n$ is an even natural number.
The first step is to derive a sufficiently high order of vanishing for $\gv - \check{\gv}$, $\Lv - \check{\Lv}$, $\wv^2 - \check{\wv}^2$, $\wv^1 - \check{\wv}^1$, and $\wv^0 - \check{\wv}^0$.
This will ensure that we can apply the Carleman estimates to (variants of) these quantities as needed later in the proof.

The key is to use the equations from Proposition \ref{thm.sys_vanish} satisfied by these quantities to exchange regularity in the vertical directions for higher orders of vanishing.
More specifically, the following proposition is a quantitative statement of the fact that if the our spacetime is regular to high enough order, then we can achieve a corresponding order of vanishing for the above quantities as $\rho \searrow 0$:

\begin{proposition} \label{thm.proof_vanish}
Consider the quantities $\ms{g} - \check{\ms{g}}$, $\Lv - \check{\Lv}$, $\wv^2 - \check{\wv}^2$, $\wv^1 - \check{\wv}^1$, $\wv^0 - \check{\wv}^0$.
Then:
\begin{itemize}
\item There exist vertical tensor fields $\ms{r}_{ \gv }$, $\ms{r}_{ \Lv }$ such that
\begin{align}
\label{eq.proof_vanish_BP} \gv - \check{\gv} = \rho^{ M_0 - 2 } \, \ms{r}_{ \gv } \text{,} &\qquad \ms{r}_{ \gv } \rightarrow^2 0 \text{,} \\
\notag \Lv - \check{\Lv} = \rho^{ M_0 - 3 } \, \ms{r}_{ \Lv } \text{,} &\qquad \ms{r}_{ \Lv } \rightarrow^2 0 \text{.}
\end{align}

\item There exist vertical tensor fields $\ms{r}_{ \wv^2 }$, $\ms{r}_{ \wv^1 }$, $\ms{r}_{ \wv^0 }$ such that
\begin{align}
\label{eq.proof_vanish_A} \wv^2 - \check{\wv}^2 = \rho^{ M_0 - 4 } \, \ms{r}_{ \wv^2 } \text{,} &\qquad \ms{r}_{ \wv^2 } \rightarrow^2 0 \text{,} \\
\notag \wv^1 - \check{\wv}^1 = \rho^{ M_0 - 3 } \, \ms{r}_{ \wv^1 } \text{,} &\qquad \ms{r}_{ \wv^1 } \rightarrow^1 0 \text{,} \\
\notag \wv^0 - \check{\wv}^0 = \rho^{ M_0 - 4 } \, \ms{r}_{ \wv^0 } \text{,} &\qquad \ms{r}_{ \wv^0 } \rightarrow^2 0 \text{.}
\end{align}
\end{itemize}
\end{proposition}

\begin{proof}
First, note that since $\gb{0} = \gbc{0}$ and $\gb{n} = \gbc{n}$, then Theorem \ref{thm.aads_fg} implies
\begin{equation}
\label{eql.proof_vanish_fg} \gb{k} = \gbc{k} \text{,} \qquad \gb{\star} = \gbc{\star} \text{,} \qquad 0 \leq k \leq n \text{,}
\end{equation}
where the Fefferman-Graham coefficients $\gb{2}, \dots, \gb{\star}, \gb{n}$ and $\gbc{2}, \dots, \gbc{\star}, \gbc{n}$ are defined as in Theorem \ref{thm.aads_fg} and Corollary \ref{thm.aads_fg_exp}---with respect to $g$ and $\check{g}$, respectively.
The above and Corollary \ref{thm.aads_fg_exp} then imply that there are vertical tensor fields $\ms{r}_{ \gv }$, $\ms{r}_{ \Lv }$ satisfying
\begin{align}
\label{eql.proof_vanish_a} \gv - \check{\gv} = \rho^n \, \ms{r}_{ \gv } \text{,} &\qquad \ms{r}_{ \gv } \rightarrow^{ M_0 - n } 0 \text{,} \\
\notag \Lv - \check{\Lv} = \rho^{ n-1 } \, \ms{r}_{ \Lv } \text{,} &\qquad \ms{r}_{ \Lv } \rightarrow^{ M_0-n } 0 \text{.} 
\end{align}
Moreover, the equations \eqref{eq.aads_connection}, along with \eqref{eql.proof_vanish_a}, yield vertical tensor fields $\ms{r}_{ \wv^2 }$, $\ms{r}_{ \wv^1 }$, $\ms{r}_{ \wv^0 }$ with 
\begin{align}
\label{eql.proof_vanish_b} \wv^2 - \check{\wv}^2 = \rho^{n-2} \, \ms{r}_{ \wv^2 } \text{,} &\qquad \ms{r}_{ \wv^2 } \rightarrow^{ M_0-n } 0 \text{,} \\
\notag \wv^1 - \check{\wv}^1 = \rho^{ n-1 } \, \ms{r}_{ \wv^1 } \text{,} &\qquad \ms{r}_{ \wv^1 } \rightarrow^{ M_0-n-1 } 0 \text{,} \\
\notag \wv^0 - \check{\wv}^0 = \rho^{ n-2 } \, \ms{r}_{ \wv^0 } \text{,} &\qquad \ms{r}_{ \wv^0 } \rightarrow^{ M_0-n } 0 \text{.}
\end{align}

The idea now is to use the Bianchi and transport system \eqref{eq.sys_vanish} to inductively improve the initial orders of vanishing from \eqref{eql.proof_vanish_a} and \eqref{eql.proof_vanish_b}. 
Observe that by \eqref{eql.proof_vanish_a} and \eqref{eql.proof_vanish_b}, the quantities
\[
\rho^{-n+2} ( \wv^2 - \check{\wv}^2 ) \text{,} \qquad \rho^{-1} ( \wv^1 - \check{\wv}^1 ) \text{,} \qquad \wv^0 - \check{\wv}^0 \text{,} \qquad \rho^{-1} ( \Lv - \check{\Lv} ) \text{,} \qquad \gv - \check{\gv}
\]
all vanish in the boundary limit $\rho \searrow 0$.
Thus, we can now integrate all the equations in \eqref{eq.sys_vanish} successively from $\rho = 0$, without obtaining any boundary terms as $\rho \searrow 0$.

Integrating the third equation in \eqref{eq.sys_vanish} and using \eqref{eql.proof_vanish_a}--\eqref{eql.proof_vanish_b}, we deduce the existence of a (new) vertical tensor field $\ms{r}_{ \wv^2 }$ (for brevity, we use the same notation as before) such that
\begin{equation}
\label{eql.proof_vanish_1} \wv^2 - \check{\wv}^2 = \rho^n \, \ms{r}_{ \wv^2 } \text{,} \qquad \ms{r}_{ \wv^2 } \rightarrow^{ M_0 - n - 2 } 0 \text{.}
\end{equation}
Integrating next the fifth part of \eqref{eq.sys_vanish} and using \eqref{eql.proof_vanish_a}, \eqref{eql.proof_vanish_b}, and the improved asymptotics for $\wv^2 - \check{\wv}^2$ in \eqref{eql.proof_vanish_1} yields a vertical tensor field $\ms{r}_{ \wv^0 }$ such that
\begin{equation}
\label{eql.proof_vanish_2} \wv^0 - \check{\wv}^0 = \rho^n \, \ms{r}_{ \wv^0 } \text{,} \qquad \ms{r}_{ \wv^0 } \rightarrow^{ M_0 - n - 2 } 0 \text{.}
\end{equation}
The next quantity in the hierarchy is $\Lv - \check{\Lv}$; integrating the second identity in \eqref{eq.sys_vanish} and using \eqref{eql.proof_vanish_a} and \eqref{eql.proof_vanish_1}, we deduce the existence of a vertical tensor field $\ms{r}_{ \Lv }$ such that
\begin{equation}
\label{eql.proof_vanish_3} \Lv - \check{\Lv} = \rho^{n+1} \ms{r}_{ \Lv } \text{,} \qquad \ms{r}_{ \Lv } \rightarrow^{ M_0 - n - 2 } 0 \text{.}
\end{equation}
Then, integrating the first part of \eqref{eq.sys_vanish} and applying \eqref{eql.proof_vanish_3} yields an $\ms{r}_{ \gv }$ with
\begin{equation}
\label{eql.proof_vanish_4} \gv - \check{\gv} = \rho^{n+2} \, \ms{r}_{ \gv } \text{,} \qquad \ms{r}_{ \gv } \rightarrow^{ M_0 - n - 2 } 0 \text{.}
\end{equation}
Integrating finally the fourth identity of \eqref{eq.sys_vanish} and recalling \eqref{eql.proof_vanish_b}--\eqref{eql.proof_vanish_4} results in an $\ms{r}_{ \wv^1 }$ with
\begin{equation}
\label{eql.proof_vanish_5} \wv^1 - \check{\wv}^1 = \rho^{ n + 1 } \, \ms{r}_{ \wv^1 } \text{,} \qquad \ms{r}_{ \wv^1 } \rightarrow^{ M_0 - n - 3 } 0 \text{.}
\end{equation}

At this point, the vanishing of all quantities has been improved by two powers of $\rho$ at the cost of two derivatives of regularity, compared to the initial asymptotics \eqref{eql.proof_vanish_a}--\eqref{eql.proof_vanish_b}.
By iterating this process (now inserting the improved asymptotics into the right-hand sides of \eqref{eq.sys_vanish}), we can repeatedly trade two derivatives of regularity for two powers of $\rho$.
Since $M_0 - n$ is even by assumption, the proof is completed after iterating this process $( M_0 - n - 2 ) / 2$ times.
\end{proof}

\begin{corollary} \label{cor:vanish}
Let the renormalized quantities $\ms{Q}$, $\ms{B}$, $\ms{W}^2$, $\ms{W}^1$, $\ms{W}^\star$ be as in Definition \ref{def.sys_auxiliary}.
Then, there exist vertical tensor fields $\ms{r}_{ \ms{Q} }$, $\ms{r}_{ \ms{B} }$, $\ms{r}_{ \ms{W}^2 }$, $\ms{r}_{ \ms{W}^1 }$, $\ms{r}_{ \ms{W}^\star }$ such that
\begin{align}
\label{eq.proof_vanish_renorm} \ms{Q} = \rho^{ M_0 } \, \ms{r}_{ \ms{Q} } \text{,} &\qquad \ms{r}_{ \ms{Q} } \rightarrow^2 0 \text{,} \\
\notag \ms{B} = \rho^{ M_0 - 2 } \, \ms{r}_{ \ms{B} } \text{,} &\qquad \ms{r}_{ \ms{B} } \rightarrow^1 0 \text{,} \\
\notag \ms{W}^2 = \rho^{ M_0 - 4 } \, \ms{r}_{ \ms{W}^2 } \text{,} &\qquad \ms{r}_{ \ms{W}^2 } \rightarrow^2 0 \text{,} \\
\notag \ms{W}^1 = \rho^{ M_0 - 3 } \, \ms{r}_{ \ms{W}^1 } \text{,} &\qquad \ms{r}_{ \ms{W}^1 } \rightarrow^1 0 \text{,} \\
\notag \ms{W}^\star = \rho^{ M_0 - 4 } \, \ms{r}_{ \ms{W}^\star } \text{,} &\qquad \ms{r}_{ \ms{W}^\star } \rightarrow^2 0 \text{.}
\end{align}
\end{corollary}

\begin{proof}
This is an immediate consequence of \eqref{eq.aads_wstar}, Definition \ref{def.sys_auxiliary}, and Proposition \ref{thm.proof_vanish}.
\end{proof}

\subsection{Applying the Carleman Estimates} \label{sec.proof_carleman}

Fix any Riemannian metric $\mf{h}$ on $\mi{I}$, and fix additional constants $p, f_\star > 0$ (whose precise values will be determined later) such that
\begin{equation}
\label{eq.proof_lambda} 0 < 2 p < 1 \text{,} \qquad f_\star \ll_{ \ms{g}, \mi{D} } 1 \text{.}
\end{equation}
Furthermore, as long as $M_0$ is sufficiently large (depending on $\ms{g}$, $\mi{D}$), we can find a sufficiently large $\kappa \in \R$ (again, the precise value will be determined later) satisfying
\begin{equation}
\label{eq.proof_kappa} 1 \ll_{ n, \ms{g}, \mi{D} } \kappa \leq M_0 - 5 \text{.}
\end{equation}

\begin{remark}
In particular, $p, f_\star, \kappa$ are chosen so that we can apply the Carleman estimates \eqref{eq.carleman} and \eqref{eq.carleman_transport} as needed.
Also, although all the above parameters depend on $\mf{h}$, we can instead view this as a dependence on $\mi{D}$, as any choice of a Riemannian metric $\mf{h}$ on $\bar{\mi{D}}$ suffices.
\end{remark}

In addition, let $f$ be as in \eqref{eq.carleman_f}, where $\eta$ is given from the GNCC assumption.
We also let
\begin{equation}
\label{eq.proof_chi} \chi := \bar{\chi} \circ f \text{,}
\end{equation}
where $\bar{\chi}: \R \rightarrow [0,1]$ is a smooth cut-off function satisfying
\begin{equation}
\label{eq.proof_cutoff} \bar{\chi} (s) = \begin{cases} 1 & s \leq f_i \text{,} \\ 0 & s \geq f_e \text{,} \end{cases} \qquad 0 < f_i < f_e < f_\star \text{,}
\end{equation}
and we define the following associated subregions of $\Omega_{ f_\star }$ (see \eqref{eq.carleman_omega}):
\begin{align}
\label{eq.proof_omega} \Omega_i := \{ f < f_i \} \text{,} \qquad \Omega_e := \{ f_i \leq f < f_e \} \text{.}
\end{align}
Lastly, for convenience, we define the following collections of quantities:
\begin{equation}
\label{eq.proof_sets} \Xi^\star := \{ \ms{W}^\star, \ms{W}^1, \ms{W}^2 \} \text{,} \qquad \Upsilon^\star := \{ \gv - \check{\gv}, \ms{Q}, \Lv - \check{\Lv}, \ms{B} \} \text{.}
\end{equation}

We now apply Theorem \ref{thm.carleman} to $\chi \ms{V}$, for each $\ms{V} \in \Xi^\star$, on the region $\Omega_{ f_\star }$ (since $\chi \ms{V}$ vanishes near $\{ f = f_\star \}$).
Restricting the right-hand side of \eqref{eq.carleman} to $\Omega_i$, on which $\chi \equiv 1$, we then obtain
\begin{align}
\label{eq.proof_wave_0} &\lambda \int_{ \Omega_i } e^{ -\frac{ \lambda f^p }{p} } f^{ n - 2 - 2 \kappa } \sum_{ \ms{V} \in \Xi^\star } ( \rho^{2p} | \ms{V} |_{ \mf{h} }^2 + \rho^4 | \ms{D} \ms{V} |_{ \mf{h} }^2 ) \, d \mu_g \\
\notag &\quad \lesssim \int_{ \Omega_i \cup \Omega_e } e^{ -\frac{ \lambda f^p }{p} } f^{ n - 2 - p - 2 \kappa } \sum_{ \ms{V} \in \Xi^\star } | ( \Boxm + \sigma_{ \ms{V} } ) ( \chi \ms{V} ) |_{ \mf{h} }^2 \, d \mu_g \\
\notag &\quad \lesssim \int_{ \Omega_i \cup \Omega_e } e^{ -\frac{ \lambda f^p }{p} } f^{ n - 2 - 2 \kappa } \rho^{-p} \sum_{ \ms{V} \in \Xi^\star } ( | \mc{F}_{ \ms{V} } |_{ \mf{h} }^2 + | \mc{G}_{ \ms{V} } |_{ \mf{h} }^2 ) \, d \mu_g
\end{align}
for any $\lambda \gg_{ \ms{g}, \mi{D} } 1$, where $( \sigma_{ \ms{W}^2 }, \sigma_{ \ms{W}^1 }, \sigma_{ \ms{W}^\star } ) = ( 2n-4, n-1, 0 )$, and where
\begin{equation}
\label{eq.proof_wave_FG} \mc{F}_{ \ms{V} } := \chi ( \Boxm + \sigma_{ \ms{V} } ) \ms{V} \text{,} \qquad \mc{G}_{ \ms{V} } := ( \Boxm + \sigma_{ \ms{V} } ) ( \chi \ms{V} ) - \mc{F}_{ \ms{V} } \text{,} \qquad \ms{V} \in \Xi^\star \text{.}
\end{equation}
Note also in the last step of \eqref{eq.proof_wave_0}, we used that $\rho \lesssim f$ by definition.

Observe the boundary limit as $\rho_\star \searrow 0$ in \eqref{eq.carleman} vanishes due to Corollary \ref{cor:vanish} and \eqref{eq.proof_kappa}.
Moreover, note that since every term in $\mc{G}_{ \ms{V} }$ involves at least one derivative of $\chi$, then it is supported in $\Omega_e$.
Expanding $\mc{F}_{ \ms{V} }$ on $\Omega_i$ using Proposition \ref{thm.sys_wave}, then \eqref{eq.proof_wave_0} becomes
\begin{align}
\label{eq.proof_wave_1} &\lambda \int_{ \Omega_i } e^{ -\frac{ \lambda f^p }{p} } f^{ n - 2 - 2 \kappa } \sum_{ \ms{V} \in \Xi^\star } ( \rho^{2p} | \ms{V} |_{ \mf{h} }^2 + \rho^4 | \ms{D} \ms{V} |_{ \mf{h} }^2 ) \, d \mu_g \\
\notag &\quad \lesssim \int_{ \Omega_i } e^{ -\frac{ \lambda f^p }{p} } f^{ n - 2 - 2 \kappa } \left[ \sum_{ \ms{V} \in \Xi^\star } ( \rho^{ 4 - p } | \ms{V} |_{ \mf{h} }^2 + \rho^{ 6 - p } | \Dv \ms{V} |_{ \mf{h} }^2 ) + \sum_{ \ms{U} \in \Upsilon^\star } ( \rho^{ 2 - p } | \ms{U} |_{ \mf{h} }^2 + \rho^{ 4 - p } | \Dv \ms{U} |_{ \mf{h} }^2 ) \right] d \mu_g \\
\notag &\quad \qquad + \int_{ \Omega_e } e^{ -\frac{ \lambda f^p }{p} } f^{ n - 2 - 2 \kappa } \rho^{-p} \sum_{ \ms{V} \in \Xi^\star } ( | \mc{F}_{ \ms{V} } |_{ \mf{h} }^2 + | \mc{G}_{ \ms{V} } |_{ \mf{h} }^2 ) \, d \mu_g \text{.}
\end{align}

Next, we can similarly apply the transport Carleman estimate of Proposition \ref{thm.carleman_transport} to each $\ms{U} \in \Upsilon^\ast$ on $\Omega_i$.
In particular, applying \eqref{eq.carleman_transport} with $s = 2$ (for $\ms{U} = \Lv - \check{\Lv}$) and $s = 0$ (otherwise) yields
\begin{align}
\label{eq.proof_transport_0} &\lambda \int_{ \Omega_i } e^{ -\frac{ \lambda f^p }{p} } f^{ n - 2 - 2 \kappa } \sum_{ \ms{U} \in \Upsilon^\star } | \ms{U} |_{ \mf{h} }^2 \, d \mu_g \\
\notag &\quad \lesssim \int_{ \Omega_i } e^{ -\frac{ \lambda f^p }{p} } f^{ n - 2 - p - 2 \kappa } \left[ \sum_{ \ms{U} \in \Upsilon^\star \setminus \{ \Lv - \check{\Lv} \} } \rho^2 | \mi{L}_\rho \ms{U} |_{ \mf{h} }^2 + \rho^4 | \mi{L}_\rho [ \rho^{-1} ( \Lv - \check{\Lv} ) ] |_{ \mf{h} }^2 \right] d \mu_g \text{,}
\end{align}
where $\lambda$ is as before.
(The boundary term as $\rho_\star \searrow 0$ in \eqref{eq.carleman_transport} vanishes by Proposition \ref{thm.proof_vanish}, Corollary \ref{cor:vanish}, and \eqref{eq.proof_kappa}.
Also, note the cutoff $\chi$ is not needed here, since Proposition \ref{thm.carleman_transport} does not require $\Psi$ to vanish near $\{ f = f_\ast \}$.)
Applying \eqref{eq.sys_transport} to the right-hand side of \eqref{eq.proof_transport_0}, we obtain that
\begin{align}
\label{eq.proof_transport_1} &\lambda \int_{ \Omega_i } e^{ -\frac{ \lambda f^p }{p} } f^{ n - 2 - 2 \kappa } \sum_{ \ms{U} \in \Upsilon^\star } | \ms{U} |_{ \mf{h} }^2 \, d \mu_g \\
\notag &\quad \lesssim \int_{ \Omega_i } e^{ -\frac{ \lambda f^p }{p} } f^{ n - 2 - 2 \kappa } \left[ \sum_{ \ms{V} \in \Xi^\star } \rho^{ 2 - p } | \ms{V} |_{ \mf{h} }^2 + \sum_{ \ms{U} \in \Upsilon^\star } \rho^{ 2 - p } | \ms{U} |_{ \mf{h} }^2 \right] d \mu_g \text{.}
\end{align}
We can also obtain analogous bounds for $\Dv \ms{U}$, for each $\ms{U} \in \Upsilon^\ast$.
In particular, we apply Proposition \ref{thm.carleman_transport} to the equations \eqref{eq.sys_transport_deriv} (with $s = 5$ for $\ms{U} = \Lv - \check{\Lv}$ and $s = 3$ otherwise), which yields
\begin{align}
\label{eq.proof_transport_2} &\lambda \int_{ \Omega_i } e^{ -\frac{ \lambda f^p }{p} } f^{ n - 2 - 2 \kappa } \sum_{ \ms{U} \in \Upsilon^\star } \rho^3 | \Dv \ms{U} |_{ \mf{h} }^2 \, d \mu_g \\
\notag &\quad \lesssim \int_{ \Omega_i } e^{ -\frac{ \lambda f^p }{p} } f^{ n - 2 - 2 \kappa } \left[ \sum_{ \ms{V} \in \Xi^\star } \rho^{ 5 - p } | \Dv \ms{V} |_{ \mf{h} }^2 + \sum_{ \ms{U} \in \Upsilon^\star } ( \rho^{ 5 - p } | \ms{U} |_{ \mf{h} }^2 + \rho^{ 5 - p } | \Dv \ms{U} |_{ \mf{h} }^2 ) \right] d \mu_g \text{.}
\end{align}

Summing the estimates \eqref{eq.proof_wave_1}, \eqref{eq.proof_transport_1}, and \eqref{eq.proof_transport_2} yields
\begin{align*}
&\lambda \int_{ \Omega_i } e^{ -\frac{ \lambda f^p }{p} } f^{ n - 2 - 2 \kappa } \left[ \sum_{ \ms{V} \in \Xi^\star } ( \rho^{2p} | \ms{V} |_{ \mf{h} }^2 + \rho^4 | \ms{D} \ms{V} |_{ \mf{h} }^2 ) + \sum_{ \ms{U} \in \Upsilon^\star } ( | \ms{U} |_{ \mf{h} }^2 + \rho^3 | \Dv \ms{U} |_{ \mf{h} }^2 ) \right] d \mu_g \\
&\quad \lesssim \int_{ \Omega_i } e^{ -\frac{ \lambda f^p }{p} } f^{ n - 2 - 2 \kappa } \left[ \sum_{ \ms{V} \in \Xi^\star } ( \rho^{ 2 - p } | \ms{V} |_{ \mf{h} }^2 + \rho^{ 5 - p } | \Dv \ms{V} |_{ \mf{h} }^2 ) + \sum_{ \ms{U} \in \Upsilon^\star } ( \rho^{ 2 - p } | \ms{U} |_{ \mf{h} }^2 + \rho^{ 4 - p } | \Dv \ms{U} |_{ \mf{h} }^2 ) \right] d \mu_g \\
&\quad \qquad + \int_{ \Omega_e } e^{ -\frac{ \lambda f^p }{p} } f^{ n - 2 - 2 \kappa } \rho^{-p} \sum_{ \ms{V} \in \Xi^\star } ( | \mc{F}_{ \ms{V} } |_{ \mf{h} }^2 + | \mc{G}_{ \ms{V} } |_{ \mf{h} }^2 ) \, d \mu_g \text{.}
\end{align*}
Taking $\lambda$ sufficiently large in the above, the integral over $\Omega_i$ in the right-hand side can be absorbed into the left-hand side, since the quantities in the right contain higher powers of $\rho$ than those in the left (and since $\rho$ is bounded on $\Omega_i$).
As a result, we have
\begin{align}
\label{eq.proof_carleman} &\lambda \int_{ \Omega_i } e^{ -\frac{ \lambda f^p }{p} } f^{ n - 2 - 2 \kappa } \left[ \sum_{ \ms{V} \in \Xi^\star } ( \rho^{2p} | \ms{V} |_{ \mf{h} }^2 + \rho^4 | \ms{D} \ms{V} |_{ \mf{h} }^2 ) + \sum_{ \ms{U} \in \Upsilon^\star } ( | \ms{U} |_{ \mf{h} }^2 + \rho^3 | \Dv \ms{U} |_{ \mf{h} }^2 ) \right] d \mu_g \\
\notag &\quad \lesssim \int_{ \Omega_e } e^{ -\frac{ \lambda f^p }{p} } f^{ n - 2 - 2 \kappa } \rho^{-p} \sum_{ \ms{V} \in \Xi^\star } ( | \mc{F}_{ \ms{V} } |_{ \mf{h} }^2 + | \mc{G}_{ \ms{V} } |_{ \mf{h} }^2 ) \, d \mu_g \text{.}
\end{align}

Finally, using that the Carleman weight
\[
w_\lambda (f) := e^{-2 \lambda p^{-1} f^p } f^{ n-2-2\kappa-p }
\]
satisfies $w_\lambda (f) \leq w_\lambda ( f_i )$ on $\Omega_e$ and $w_\lambda (f) \geq w_\lambda ( f_i )$ on $\Omega_i$, we can eliminate $w_\lambda (f)$ from \eqref{eq.proof_carleman}:
\begin{align*}
&\lambda \int_{ \Omega_i } \left[ \sum_{ \ms{V} \in \Xi^\star } ( \rho^{2p} | \ms{V} |_{ \mf{h} }^2 + \rho^4 | \ms{D} \ms{V} |_{ \mf{h} }^2 ) + \sum_{ \ms{U} \in \Upsilon^\star } ( | \ms{U} |_{ \mf{h} }^2 + \rho^3 | \Dv \ms{U} |_{ \mf{h} }^2 ) \right] d \mu_g \\
&\quad \lesssim \int_{ \Omega_e } \rho^{-p} \sum_{ \ms{V} \in \Xi^\star } ( | \mc{F}_{ \ms{V} } |_{ \mf{h} }^2 + | \mc{G}_{ \ms{V} } |_{ \mf{h} }^2 ) \, d \mu_g \text{.}
\end{align*}
Since the right-hand side of the above is finite by our vanishing assumptions, we can divide by $\lambda$ and take the limit $\lambda \rightarrow \infty$ to deduce in particular that $\gv - \check{\gv} \equiv 0$ in $\Omega_i$.

\begin{remark}
Note that by taking $f_i$ to be arbitrarily close to $f_\star$, the above argument yields that $\gv - \check{\gv}$ vanishes (along with all the other associated difference quantities) on $\Omega_{ f_\star }$.
\end{remark}

Since $g$ and $\check{g}$ are assumed to be in the FG gauge \eqref{eq.aads_metric}, and with the same radial coordinate $\rho$, the above immediately yields $g = \check{g}$ on $\Omega_i$.
This finishes the proof of Theorem \ref{thm.adscft}.

%% file: app.tex
In this section, we state and prove precise versions of the main results of this paper:
\begin{itemize}
\item Section \ref{sec.app_gauge} is dedicated to the precise statement and proof of Theorem \ref{theo:adscft}.

\item Section \ref{sec.app_symm} is dedicated to the precise statements and proofs of Theorems \ref{theo:killing} and \ref{theo:symmetries}.
\end{itemize}

\subsection{Gauge Covariance} \label{sec.app_gauge}

In Theorem \ref{thm.adscft}, we had assumed the two spacetimes in question had identical boundary data.
Here, we extend Theorem \ref{thm.adscft}---we establish the same conclusions, but we weaken the assumptions so that the spacetimes have \emph{gauge-equivalent} boundary data.
The key idea is to apply a specific change of coordinates (preserving the Fefferman-Graham gauge condition) on one spacetime, so that the boundary data for the two spacetimes become equal.

Before discussing the precise result, let us first give a more detailed description of gauge transformations in our current context of FG-aAdS segments.

\begin{definition} \label{def.app_gauge}
Let $( \mi{M}, g )$ be an FG-aAdS segment, and let $\mi{V} \subset \mi{M}$ be an open neighbourhood of the conformal boundary.\footnote{More precisely, for any $p \in \mi{I}$, there exists $\rho_p > 0$ such that $\mi{V}$ contains $( 0, \rho_p ) \times \{ p \}$.}
A function $\check{\rho} \in C^\infty ( \mi{V} )$ is called an \emph{FG radius} for $( \mi{V}, g )$ iff
\begin{equation}
\label{eq.app_gauge} \check{\rho} > 0 \text{,} \qquad \check{\rho}^{-2} g^{-1} ( d \check{\rho}, d \check{\rho} ) = 1 \text{.}
\end{equation}
\end{definition}

\begin{remark}
Note that $\rho$ itself is an FG radius for $( \mi{M}, g )$.
\end{remark}

Informally, we can view such an FG radius as a change of boundary defining function (from $\rho$ to $\check{\rho}$) that preserves the Fefferman-Graham gauge condition \eqref{eq.aads_metric}.
A more detailed justification, in terms of our language of FG-aAdS segments, arises from the following.

Consider a FG radius $\check{\rho}$ for $( \mi{V}, g )$ satisfying \eqref{eq.app_gauge}.
Given $\sigma > 0$ and $p \in \mi{I}$, we identify the pair $( \sigma, p )$ with $\gamma_p (\sigma) \in \mi{V}$, where $\gamma_p$ is the integral curve of the $\check{\rho}^2 g$-gradient of $\check{\rho}$ that satisfies\footnote{In other words, $\gamma_p$ emanates from the point $p$ on the conformal boundary.}
\[
\lim_{ \sigma \searrow 0 } \gamma_p ( \sigma ) = ( 0, p ) \text{.}
\]
(when all the above quantities exist).
Then, given any $\mi{J} \subset \mi{I}$ with $\bar{\mi{J}} \subseteq \mi{I}$, the above identifies $( 0, \check{\rho}_0 ] \times \mi{J}$, for some $\check{\rho}_0 > 0$, with an open submanifold $\smash{\check{\mi{M}}} \subseteq \mi{M}$.
Moreover, by the second part of \eqref{eq.app_gauge}, the projection onto the $( 0, \check{\rho}_0 ]$-component is simply $\check{\rho}$.
As this gradient is ($g$-)normal to the level sets of $\check{\rho}$, then $g$ is given---in terms of $\check{\rho}$ and frames transported along this gradient---by
\begin{equation}
\label{eq.app_fg_gauge} g = \check{\rho}^{-2} ( d \check{\rho}^2 + \ms{g}^\ast ) \text{,}
\end{equation}
for some ($\check{\rho}$-)vertical tensor field $\ms{g}^\ast$.
In other words, $\tilde{\rho}$ generates an FG-aAdS segment, characterized by \eqref{eq.app_fg_gauge}, that is isometric to (part of) the original FG-aAdS segment $( \mi{M}, g )$ defined from $\rho$.

It is well-known in the physics literature (see, e.g., \cite{deharo_sken_solod:holog_adscft, imbim_schwim_theis_yanki:diffeo_holog}) that these transformations preserving the FG gauge can be characterized in terms of corresponding transformations of the boundary data.
Below, we present this in a more rigorous form, adapted to the setting of this paper.

\begin{proposition} \label{thm.app_gauge_change}
Let $( \mi{M}, g )$ be a vacuum FG-aAdS segment, and fix $\mf{a} \in C^\infty ( \mi{I} )$.
Then, there is a neighbourhood $\mi{V} \subset \mi{M}$ of the conformal boundary and a unique FG radius $\check{\rho}$ on $( \mi{V}, g )$ such that
\begin{equation}
\label{eq.app_gauge_lim} \frac{ \check{\rho} }{ \rho } \rightarrow^0 e^{ \mf{a} } \text{.}
\end{equation}

Furthermore, let $( \mi{I}, \gb{0}, \gb{n} )$ and $( \mi{I}, \gbc{0}, \gbc{n} )$ denote the holographic data associated to the FG-aAdS segments constructed from $( \mi{M}, g )$ with respect to $\rho$ and $\check{\rho}$, respectively.
Then, there exists a universal algebraic function $\mi{F}$ (i.e., independent of $\mi{M}, g, \check{g}, \rho, \check{\rho}$) such that
\begin{equation}
\label{eq.app_gauge_g} \gbc{0} = e^{ 2 \mf{a} } \gb{0} \text{,} \qquad \gbc{n} = \mi{F} ( \gb{n}, \gb{0}, \Rm, \dots, \Dm^{n-2} \Rm, \mf{a}, \Dm \mf{a}, \dots, \Dm^n \mf{a} ) \text{,}
\end{equation}
where $\Dm$ and $\Rm$ denotes the Levi-Civita connection and Riemann curvature for $\gb{0}$, respectively.
\end{proposition}

\begin{proof}[Proof sketch of Proposition \ref{thm.app_gauge_change}.]
To obtain $\check{\rho}$, we adapt an argument inspired by the constructions from \cite{gra:vol_renorm}.
First, we consider the following ansatz for our desired FG radius:
\begin{equation}
\label{eql.adscft_gauge_0} \check{\rho} := e^{ \ms{a} } \rho \text{,} \qquad \ms{a} \in C^\infty ( \mi{M} ) \text{.}
\end{equation}
Then, the defining equation \eqref{eq.app_gauge} for $\check{\rho}$ expands as
\[
1 = \rho^{-2} g^{-1} ( d \rho, d \rho ) + 2 \rho \cdot \rho^{-2} g^{-1} ( d \rho, d \ms{a} ) + \rho^2 \cdot \rho^{-2} g^{-1} ( d \ms{a}, d \ms{a} ) \text{.}
\]
Since $g$ and $\rho$ satisfy \eqref{eq.aads_metric}, the above becomes
\begin{align}
\label{eql.adscft_gauge_pde} 0 &= 2 \rho^{-2} g^{-1} ( d \rho, d \ms{a} ) + \rho \cdot \rho^{-2} g^{-1} ( d \ms{a}, d \ms{a} ) \\
\notag &= ( 2 + \rho \mi{L}_\rho \ms{a} ) \mi{L}_\rho \ms{a} + \rho \, \gv^{-1} ( \Dv \ms{a}, \Dv \ms{a} ) \text{.}
\end{align}
Observe \eqref{eql.adscft_gauge_pde} yields a fully nonlinear equation for $\ms{a}$, which can be uniquely solved via the method of characteristics on a neighbourhood $\mi{V}$ of the conformal boundary given initial data
\[
\lim_{ \sigma \searrow 0 } \ms{a} |_\sigma = \mf{a} \text{.}
\]
Then, $\check{\rho} = e^{ \ms{a} } \rho$ yields the FG radius on $\mi{V}$ which satisfies \eqref{eq.app_gauge_lim}.

The relations between $( \gb{0}, \gb{n} )$ and $( \gbc{0}, \gbc{n} )$, as well as their derivations, are standard in physics literature---see, for instance, \cite{deharo_sken_solod:holog_adscft, imbim_schwim_theis_yanki:diffeo_holog}.
A detailed derivation of the first relation in \eqref{eq.app_gauge_g}, in the context of FG-aAdS segments, can be found in \cite[Proposition 3.4]{ChatzikaleasShao}.
\end{proof}

\begin{remark}
Using \eqref{eql.adscft_gauge_pde}, along with the techniques of \cite{shao:aads_fg}, one can derive boundary limits for $\rho$-derivatives of both $\ms{a}$ and $\check{\rho}$ in Proposition \ref{thm.app_gauge_change}.
From these limits, one then obtains partial series expansions for $\ms{a}$ and $\check{\rho}$ in terms of powers of $\rho$, similar to those in Corollary \ref{thm.aads_fg_exp}.
In particular, these expansions justify the ansatz for FG gauge transforms used in \cite[Equation (3.4)]{ChatzikaleasShao}.
\end{remark}

\begin{remark}
In addition to \eqref{eq.app_gauge_g}, there are similar transformation formulas for the other coefficients $\gb{2}, \gb{4}, \dots, \gb{\star}$ of the Fefferman-Graham expansion.
(For example, $-\gb{2}$ transforms like the Schouten tensor under conformal rescalings---see \cite[Proposition 3.4]{ChatzikaleasShao} for a detailed derivation.)
\end{remark}

In particular, Proposition \ref{thm.app_gauge_change} implies that two holographic data $( \mi{I}, \gb{0}, \gb{n} )$, $( \mi{I}, \gbc{0}, \gbc{n} )$ that are related via \eqref{eq.app_gauge_g} are associated to the \emph{same} aAdS spacetime, via two different FG gauges on this spacetime.
This motivates the following definition:

\begin{definition} \label{def.app_gauge_equiv}
Let $\mi{I}$ be an $n$-dimensional manifold, let $\gb{0}, \gbc{0}$ be two Lorentzian metrics on $\mi{I}$, and let $\gb{n}, \gbc{n}$ be two symmetric rank-$( 0, 2 )$ tensor fields on $\mi{I}$.
We say $( \gb{0}, \gb{n} )$ and $( \gbc{0}, \gbc{n} )$ are \emph{gauge-equivalent} on $\mi{D} \subset \mi{I}$ iff these quantities satisfy \eqref{eq.app_gauge_g} on $\mi{D}$ for some $\mf{a} \in C^\infty ( \mi{I} )$.
\end{definition}

We can now state the precise version of Theorem \ref{theo:adscft}, which combines the results of Theorem \ref{thm.adscft} while also taking into account the above gauge covariance:

\begin{theorem} \label{thm.adscft_gauge}
Let $n > 2$, and let $( \mi{M}, g )$, $( \mi{M}, \check{g} )$ be vacuum FG-aAdS segments (on a common aAdS region $\mi{M}$), with associated holographic data $( \mi{I}, \gb{0}, \gb{n} )$, $( \mi{I}, \gbc{0}, \gbc{n} )$ (respectively).
In addition, let $\mi{D} \subset \mi{I}$ be open with compact closure, and assume:
\begin{itemize}
\item $( \mi{M}, g )$ and $( \mi{M}, \check{g} )$ are regular to some large enough order $M_0$ (depending on $\gv$, $\check{\gv}$, $\mi{D}$).

\item $( \gb{0}, \gb{n} )$ and $( \gbc{0}, \gbc{n} )$ are gauge-equivalent on $\mi{D}$.

\item $( \mi{D}, \gb{0} )$ (or equivalently, $( \mi{D}, \gbc{0} )$) satisfies the GNCC.
\end{itemize}
Then, $g$ and $\check{g}$ are isometric near $\{ 0 \} \times \mi{D}$ (viewed as part of the conformal boundary).
To be more precise, there exists a sufficiently small $f_\star > 0$ and some $\Psi: \Omega_{ f_\star } \rightarrow \mi{M}$ such that\footnote{In other words, $\Psi$ is an isometry that fixes the conformal boundary.}
\begin{equation}
\label{eq.adscft_gauge_boundary} \Psi_\ast \check{g} = g \text{,} \qquad \lim_{ \sigma \searrow 0 } \Psi ( \sigma, p ) = ( 0, p ) \text{,} \quad p \in \mi{D} \text{.}
\end{equation}
\end{theorem}

\begin{proof}
By gauge equivalence, there exists $\mf{a} \in C^\infty ( \mi{I} )$ so that \eqref{eq.app_gauge_g} holds on $\mi{D}$.
Applying Proposition \ref{thm.app_gauge_change}, we obtain an FG radius $\check{\rho}$ such \eqref{eq.app_gauge_lim} holds.
Using $\check{\rho}$ (see the discussion below Definition \ref{def.app_gauge}), we can identify $( \mi{M}, g )$ with another FG-aAdS segment $( \mi{N}, g )$, for which the induced holographic data is $( \mi{I}, \gbc{0}, \gbc{n} )$.
Moreover, applying another isometry that identifies the FG gauges for $( \mi{N}, g )$ and $( \mi{M}, \check{g} )$,\footnote{Here, one maps the integral curves of the gradient of $\tilde{\rho}$ in $\mi{N}$ to the integral curves of the gradient of $\rho$ in $\mi{M}$.} we then arrive at the setting of Theorem \ref{thm.adscft}---two FG-aAdS segments on a common manifold $\mi{M}$, with identical data $( \gbc{0}, \gbc{n} )$ on $\mi{D}$.

The result now follows by applying Theorem \ref{thm.adscft}; the desired isometry $\Psi$ is then constructed by unwinding all the identifications made in the preceding discussion.
\end{proof}

\begin{remark}
Recall the GNCC is gauge-invariant, so that in the setting of Theorem \ref{thm.adscft_gauge}, we have that the GNCC holds for $( \mi{D}, \gb{0} )$ if and only if it holds for $( \mi{D}, \gbc{0} )$; see \cite[Proposition 3.6]{ChatzikaleasShao} for a proof of this statement.
Therefore, it suffices to assume the GNCC only for $\gb{0}$.
\end{remark}

\subsection{Extension of Symmetries} \label{sec.app_symm}

One application of Theorem \ref{thm.adscft} is that it immediately implies that holographic symmetries on the conformal boundary must be inherited in the bulk spacetime.
Here, we prove two versions of this, one for discrete and the other for continuous symmetries.

The first result is the precise analogue of Theorem \ref{theo:symmetries} from the introduction:

\begin{theorem} \label{thm.app_symm}
Let $n > 2$, and let $( \mi{M}, g )$ be a vacuum FG-aAdS segment, with associated holographic data $( \mi{I}, \gb{0}, \gb{n} )$.
Consider also a smooth, invertible function $\phi: \bar{\mi{D}} \rightarrow \mi{I}$,\footnote{More accurately, $\phi: \mi{J} \rightarrow \mi{I}$ for some open $\mi{J} \subseteq \mi{I}$ that contains $\bar{\mi{D}}$.} where $\mi{D} \subset \mi{I}$ is open with compact closure.
In addition, assume that the following hold:
\begin{itemize}
\item $( \mi{M}, g )$ is regular to some large enough order $M_0$ (depending on $\gv$, $\mi{D}$).

\item $( \gb{0}, \gb{n} )$ and $( \phi_\ast \gb{0}, \phi_\ast \gb{n} )$ are gauge-equivalent on $\mi{D}$.

\item $( \mi{D}, \gb{0} )$ satisfies the GNCC.
\end{itemize}
Then, $\phi$ extends to an isometry of $g$ near $\{ 0 \} \times \mi{D}$ (viewed as part of the conformal boundary).
To be more precise, there exists a sufficiently small $f_\star > 0$ and some $\Phi: \Omega_{ f_\star } \rightarrow \mi{M}$ such that\footnote{In other words, $\Phi$ is an isometry that asymptotes to $\phi$ at the conformal boundary.}
\begin{equation}
\label{eq.app_symm_boundary} \Phi_\ast g = g \text{,} \qquad \lim_{ \sigma \searrow 0 } \Phi ( \sigma, p ) = ( 0, \phi (p) ) \text{,} \quad p \in \mi{D} \text{.}
\end{equation}

Furthermore, if $\phi$ is a \emph{holographic isometry} on $\mi{D}$, i.e.,
\begin{equation}
\label{eq.app_symm_isom} ( \gb{0}, \gb{n} ) |_{ \mi{D} } = ( \phi_\ast \gb{0}, \phi_\ast \gb{n} ) |_{ \mi{D} } \text{,}
\end{equation}
then the bulk isometry $\Phi$ is given explicitly as
\begin{equation}
\label{eq.app_symm_extra} \Phi ( \sigma, p ) = ( \sigma, \phi (p) ) \text{,} \qquad ( \sigma, p ) \in \Omega_{ f_\ast } \text{.}
\end{equation}
\end{theorem}

\begin{proof}
First, applying a gauge transformation as in the proof of Theorem \ref{thm.adscft_gauge}, we can assume \eqref{eq.app_symm_isom} holds.
Thus, it suffices to show that the map $\Phi$ given by \eqref{eq.app_symm_extra} is an isometry.

Observe that \eqref{eq.app_symm_extra} implies the pullback $\Phi_\ast g$ satisfies
\[
\Phi_\ast g = \rho^{-2} ( d \rho^2 + \phi^\ast \gv ) \text{,}
\]
where $\phi_\ast \gv$ is defined to the the pullback through $\phi$ on each level set of $\rho$:
\[
( \phi_\ast \ms{g} ) |_\sigma := \phi_\ast ( \ms{g} |_\sigma ) \text{,} \qquad \sigma \in ( 0, \rho_0 ] \text{.}
\]
In particular, $\Phi_\ast g$ defines a vacuum FG-aAdS segment (with the same $\rho$ as before), whose associated boundary data is $( \phi^{-1} ( \mi{I} ), \phi_\ast \gb{0}, \phi_\ast \gb{n} )$.
From \eqref{eq.app_symm_isom}, we see that $g$ and $\Phi_\ast g$ have the same data on $\mi{D}$, hence Theorem \ref{thm.adscft} yields $g = \Phi_\ast g$ on some $\Omega_{ f_\star }$, for $f_\star > 0$ sufficiently small.
\end{proof}

\begin{remark}
Note that in general, $\phi |_{ \mi{D} }$ in Theorem \ref{thm.app_symm} is a conformal isometry of $\gm$.
Furthermore, the extra condition \eqref{eq.app_symm_isom} implies that $\phi |_{ \mi{D} }$ is a full isometry of $\gm$.
\end{remark}

Theorem \ref{thm.app_symm} implies the following extension result for Killing vector fields, which is, in addition, the precise analogue of Theorem \ref{theo:killing} from the introduction:

\begin{theorem} \label{thm.app_killing}
Let $n > 2$, and let $( \mi{M}, g )$ be a vacuum FG-aAdS segment, with holographic data $( \mi{I}, \gb{0}, \gb{n} )$.
Fix also a smooth vector field $\mf{K}$ on $\bar{\mi{D}}$,\footnote{More precisely, $\mf{K}$ is a vector field on some open $\mi{J} \subseteq \mi{I}$ that contains $\bar{\mi{D}}$.} with $\mi{D} \subset \mi{I}$ and $\bar{\mi{D}}$ compact, and assume:
\begin{itemize}
\item $( \mi{M}, g )$ is regular to some large enough order $M_0$ (depending on $\gv$, $\mi{D}$).

\item $( \gb{0}, \gb{n} )$ is gauge-equivalent on $\mi{D}$ to some $( \gbc{0}, \gbc{n} )$ satisfying
\begin{equation}
\label{eq.app_killing_ass} ( \mi{L}_{ \mf{K} } \gbc{0}, \mi{L}_{ \mf{K} } \gbc{n} ) |_{ \mi{D} } = 0 \text{.}
\end{equation}

\item $( \mi{D}, \gb{0} )$ satisfies the GNCC.
\end{itemize}
Then, $\mf{K}$ extends to a Killing vector field $K$ for $g$ near $\{ 0 \} \times \mi{D}$ (viewed as part of the conformal boundary).
More precisely, there exist a sufficiently small $f_\star > 0$ and vector field
\begin{equation}
\label{eq.app_killing_decomp} K := K^\rho \partial_\rho + \ms{K} \text{,}
\end{equation}
where $K^\rho \in C^\infty ( \Omega_{ f_\star } )$ and $\ms{K}$ is a vertical vector field on $\Omega_{ f_\star }$, such that\footnote{In other words, $K$ is a $g$-Killing vector field that asymptotes to $\mf{K}$ at the conformal boundary.}
\begin{equation}
\label{eq.app_killing} \mi{L}_K g = 0 \text{,} \qquad K^\rho \rightarrow^0 0 \text{,} \qquad \ms{K} \rightarrow^0 \mf{K} \text{.}
\end{equation}

Furthermore, if $\mf{K}$ is a \emph{holographic Killing field} on $\mi{D}$, i.e.,
\begin{equation}
\label{eq.app_killing_holo} ( \mi{L}_{ \mf{K} } \gb{0}, \mi{L}_{ \mf{K} } \gb{n} ) |_{ \mi{D} } = 0 \text{.}
\end{equation}
then $K$ is vertical ($K^\rho \equiv 0$) and can be explicitly described via the relation
\begin{equation}
\label{eq.app_killing_boundary} \mi{L}_\rho \ms{K} = 0 \text{,} \qquad \ms{K} |_\sigma \rightarrow^0 \mf{K} \text{.}
\end{equation}
\end{theorem}

\begin{proof}
Once again, by a gauge transformation, it suffices to consider the special case in which \eqref{eq.app_killing_holo} holds.
Let $\phi_s$, for $s \in \R$ small enough, denote transport along the integral curves of $\mf{K}$ by parameter $s$.
By definition, each $\phi_s$ is a holographic isometry on $\mi{D}$,
\[
( \gb{0}, \gb{n} ) |_{ \mi{D} } = ( \phi_{ s \ast } \gb{0}, \phi_{ s \ast } \gb{n} ) |_{ \mi{D} } \text{,}
\]
so by Theorem \ref{thm.app_symm}, it extends to a bulk isometry $\Phi_s$ on some $\Omega_{ f_\star }$, $f_\star > 0$, with
\begin{equation}
\label{eql.app_killing_0} \Phi_s ( \sigma, p ) := ( \sigma, \phi_s (p) ) \text{,}
\end{equation}
and with $\Phi_{ s \ast } g = g$ near $\{ 0 \} \times \mi{D}$.
Let $K$ be the generator of the family $\{ \Phi_s \}_{ s \in \R }$; note that by definition, $K$ satisfies the first part of \eqref{eq.app_killing}.
Finally, one directly deduces from \eqref{eql.app_killing_0} that $K$ must be vertical, and that $K := \ms{K}$ satisfies the transport relation \eqref{eq.app_killing_boundary}.
\end{proof}

%% file: details.tex
This appendix contains additional proofs and computational details for readers' convenience.
We begin with some preliminary formulas that will be useful in later sections.

\subsection{Proof of Proposition \ref{thm.aads_comm}} \label{sec.aads_comm}

Throughout, we assume indices are with respect to arbitrary coordinates $( U, \varphi )$ on $\mi{I}$, and we let $( k, l )$ be the rank of $\ms{A}$.
The first part of \eqref{eq.aads_comm} is a consequence of \eqref{eq.aads_sff} and the following commutation formula from \cite[Proposition 2.27]{shao:aads_fg}:
\begin{align*}
\mi{L}_\rho \Dv_c \ms{A}^{ \ix{a} }{}_{ \ix{b} } - \Dv_c ( \mi{L}_\rho \ms{A} )^{ \ix{a} }{}_{ \ix{b} } &= \frac{1}{2} \sum_{ i = 1 }^k \gv^{ a_i d } ( \Dv_c \mi{L}_\rho \gv_{ d e } + \Dv_e \mi{L}_\rho \gv_{ d c } - \Dv_d \mi{L}_\rho \gv_{ c e } ) \ms{A}^{ \ixr{a}{i}{e} }{}_{ \ix{b} } \\
\notag &\qquad - \frac{1}{2} \sum_{ j = 1 }^l \gv^{ d e } ( \Dv_c \mi{L}_\rho \gv_{ d b_j } + \Dv_{ b_j } \mi{L}_\rho \gv_{ d c } - \Dv_d \mi{L}_\rho \gv_{ c b_j } ) \ms{A}^{ \ix{a} }{}_{ \ixr{b}{j}{e} } \text{,}
\end{align*}

Next, recalling the second formula in \eqref{eq.aads_vertical_connection}, we obtain that
\begin{align*}
\bar{\Dv}_\rho \Dv_c \ms{A}^{ \ix{a} }{}_{ \ix{b} } &= \mi{L}_\rho \Dv_c \ms{A}^{ \ix{a} }{}_{ \ix{b} } - \frac{1}{2} \gv^{ d e } \mi{L}_\rho \gv_{ c d } \, \Dv_e \ms{A}^{ \ix{a} }{}_{ \ix{b} } + \frac{1}{2} \sum_{ i = 1 }^k \gv^{ a_i d } \mi{L}_\rho \gv_{ d e } \, \Dv_c \ms{A}^{ \ixr{a}{i}{e} }{}_{ \ix{b} } \\
&\qquad - \frac{1}{2} \sum_{ j = 1 }^l \gv^{ d e } \mi{L}_\rho \gv_{ b_j d } \, \Dv_c \ms{A}^{ \ix{a} }{}_{ \ixr{b}{j}{e} } \text{,} \\
\Dv_c ( \bar{\Dv}_\rho \ms{A} )^{ \ix{a} }{}_{ \ix{b} } &= \Dv_c ( \mi{L}_\rho \ms{A} )^{ \ix{a} }{}_{ \ix{b} } + \frac{1}{2} \sum_{ i = 1 }^k \Dv_c ( \gv^{ a_i d } \mi{L}_\rho \gv_{ d e } \, \ms{A}^{ \ixr{a}{i}{e} }{}_{ \ix{b} } ) - \frac{1}{2} \sum_{ j = 1 }^l \Dv_c ( \gv^{ d e } \mi{L}_\rho \gv_{ b_j d } \, \ms{A}^{ \ix{a} }{}_{ \ixr{b}{j}{e} } ) \text{.}
\end{align*}
Subtracting the above two equations and recalling the first part of \eqref{eq.aads_comm} yields
\begin{align*}
\bar{\Dv}_\rho \Dv_c \ms{A}^{ \ix{a} }{}_{ \ix{b} } &= \Dv_c ( \bar{\Dv}_\rho \ms{A} )^{ \ix{a} }{}_{ \ix{b} } + \mi{L}_\rho \Dv_c \ms{A}^{ \ix{a} }{}_{ \ix{b} } - \Dv_c ( \mi{L}_\rho \ms{A} )^{ \ix{a} }{}_{ \ix{b} } - \frac{1}{2} \gv^{ d e } \Lv_{ c d } \, \Dv_e \ms{A}^{ \ix{a} }{}_{ \ix{b} } \\
&\qquad - \frac{1}{2} \sum_{ i = 1 }^k \gv^{ a_i d } \Dv_c \Lv_{ d e } \, \ms{A}^{ \ixr{a}{i}{e} }{}_{ \ix{b} } + \frac{1}{2} \sum_{ j = 1 }^l \gv^{ d e } \Dv_c \Lv_{ b_j d } \, \ms{A}^{ \ix{a} }{}_{ \ixr{b}{j}{e} } \\
&= \Dv_c ( \bar{\Dv}_\rho \ms{A} )^{ \ix{a} }{}_{ \ix{b} } - \frac{1}{2} \gv^{ d e } \Lv_{ c d } \, \Dv_e \ms{A}^{ \ix{a} }{}_{ \ix{b} } + \frac{1}{2} \sum_{ i = 1 }^k \gv^{ a_i d } ( \Dv_e \Lv_{ c d } - \Dv_d \Lv_{ c e } ) \ms{A}^{ \ixr{a}{i}{e} }{}_{ \ix{b} } \\
&\qquad - \frac{1}{2} \sum_{ j = 1 }^l \gv^{ d e } ( \Dv_{ b_j } \Lv_{ c d } - \Dv_d \Lv_{ c b_j } ) \ms{A}^{ \ix{a} }{}_{ \ixr{b}{j}{e} } \text{,}
\end{align*}
from which the second identity in \eqref{eq.aads_comm} follows.

Next, recalling Definition \ref{def.aads_mixed_wave} for $\Boxm$, we expand (partially in $\varphi_\rho$-coordinates)
\begin{align}
\label{eql.aads_comm_rho_0} \Boxm ( \rho^p \ms{A} ) &= \gv^{ \alpha \beta } \nablam_\alpha ( \rho^p \nablam_\beta \ms{A} + p \rho^{ p - 1 } \nablam_\beta \rho \cdot \ms{A} ) \\
\notag &= \rho^p \Boxm \ms{A} + 2 p \rho^{ p - 1 } g^{ \alpha \beta } \nabla_\alpha \rho \, \Dvm_\beta \ms{A} + p ( p - 1 ) \rho^{ p - 2 } g^{ \alpha \beta } \nabla_\alpha \rho \nabla_\beta \rho \, \ms{A} + p \rho^{ p - 1 } \Box \rho \, \ms{A} \\
\notag &= \rho^p \Boxm \ms{A} + 2 p \rho^{ p + 1 } \, \Dvm_\rho \ms{A} + p ( p - 1 ) \rho^p \, \ms{A} + p \rho^{ p - 1 } \Box \rho \, \ms{A} \text{,}
\end{align}
where we also recalled \eqref{eq.aads_metric} in the last step.
By \eqref{eq.aads_metric}, \eqref{eq.aads_Gamma}, and \eqref{eq.aads_sff}, we have
\begin{align*}
\Box \rho &= - \rho^2 \Gamma^\rho_{ \rho \rho } - \rho^2 \gv^{ a b } \Gamma^\rho_{ a b } \\
&= - ( n - 1 ) \rho + \frac{1}{2} \rho^2 \gv^{ a b } \Lv_{ a b } \text{.}
\end{align*}
Thus, combining \eqref{eql.aads_comm_rho_0} with the above yields
\[
\Boxm ( \rho^p \ms{A} ) = \rho^p \Boxm \ms{A} + 2 p \rho^{ p + 1 } \, \Dvm_\rho \ms{A} - p ( n - p ) \rho^p \, \ms{A} + \frac{1}{2} p \rho^{ p + 1 } \gv^{ a b } \Lv_{ a b } \, \ms{A} \text{,}
\]
which immediately imples \eqref{eq.aads_comm_rho}.

Finally, for \eqref{eq.aads_boxm}, the definitions of $\nablam^2_{ \rho \rho }$ and $\Dvm_\rho^2$, along with Proposition \ref{thm.aads_Gamma}, yield:
\begin{align*}
\nablam_{ \rho \rho } \ms{A}^{ \ix{a} }{}_{ \ix{b} } &= \partial_\rho ( \Dvm_\rho \ms{A}^{ \ix{a} }{}_{ \ix{b} } ) - \Gamma^\rho_{ \rho \rho } \Dvm_\rho \ms{A}^{ \ix{a} }{}_{ \ix{b} } + \sum_{ i = 1 }^k \Gammav^{ a_i }_{ \rho d } \ms{A}^{ \ixr{a}{i}{d} }{}_{ \ix{b} } - \sum_{ j = 1 }^l \Gammav^d_{ \rho b_j } \ms{A}^{ \ix{a} }{}_{ \ixr{b}{j}{d} } \text{,} \\
\Dvm_\rho ( \Dvm_\rho \ms{A} )^{ \ix{a} }{}_{ \ix{b} } &= \partial_\rho ( \Dvm_\rho \ms{A}^{ \ix{a} }{}_{ \ix{b} } ) + \sum_{ i = 1 }^k \Gammav^{ a_i }_{ \rho d } \ms{A}^{ \ixr{a}{i}{d} }{}_{ \ix{b} } - \sum_{ j = 1 }^l \Gammav^d_{ \rho b_j } \ms{A}^{ \ix{a} }{}_{ \ixr{b}{j}{d} } \text{.}
\end{align*}
Subtracting the above two equations and recalling the first part of \eqref{eq.aads_Gamma} yields
\[
\nablam_{ \rho \rho } \ms{A} = \Dvm_\rho ( \Dvm_\rho \ms{A} ) + \rho^{-1} \Dvm_\rho \ms{A} \text{.}
\]
Furthermore, relating $\Dvm_\rho$ and $\mi{L}_\rho$ using \eqref{eq.aads_vertical_connection}, the above then becomes\footnote{Here, we also used the standard identity $\mi{L}_\rho \gv^{-1} = \sch{ \gv^{-2}, \Lv }$.}
\begin{align}
\label{eql.aads_boxm_0} \nablam_{ \rho \rho } \ms{A} &= \mi{L}_\rho ( \Dvm_\rho \ms{A} ) + \sch{ \gv^{-1}, \Lv, \Dvm_\rho \ms{A} } + \rho^{-1} \mi{L}_\rho \ms{A} + \rho^{-1} \sch{ \gv^{-1}, \Lv, \ms{A} } \\
\notag &= \mi{L}_\rho^2 \ms{A} + \rho^{-1} \mi{L}_\rho \ms{A} + \sch{ \gv^{-1}, \Lv, \mi{L}_\rho \ms{A} } + \rho^{-1} \sch{ \gv^{-1}, \Lv, \ms{A} } \\
\notag &\qquad + \sch{ \gv^{-1}, \mi{L}_\rho \Lv, \ms{A} } + \sch{ \gv^{-2}, \Lv, \Lv, \ms{A} } \text{.}
\end{align}

Next, a similar decomposition of mixed and vertical derivatives in vertical components yields
\begin{align*}
\nablam_{ab} \ms{A} &= \Dv_{ab} \ms{A} - \Gamma^\rho_{ab} \Dvm_\rho \ms{A} \\
&= \Dv_{ab} \ms{A} - \rho^{-1} \gv_{ab} \Dvm_\rho \ms{A} + \sch{ \Lv, \Dvm_\rho \ms{A} }_{ab} \\
&= \Dv_{ab} \ms{A} - \rho^{-1} \gv_{ab} \mi{L}_\rho \ms{A} + \rho^{-1} \gv_{ab} \, \sch{ \gv^{-1}, \Lv, \ms{A} } + \sch{ \Lv, \mi{L}_\rho \ms{A} }_{ab} + \sch{ \gv^{-1}, \Lv, \Lv, \ms{A} }_{ab} \text{.}
\end{align*}
Therefore, applying \eqref{eq.aads_metric}, Definition \ref{def.aads_mixed_wave}, \eqref{eql.aads_boxm_0}, and the above, we conclude that
\begin{align*}
\Boxm \ms{A} &= \rho^2 \nablam_{ \rho \rho } \ms{A} + \rho^2 \gv^{ab} \nablam_{ab} \ms{A} \\
&= \rho^2 \mi{L}_\rho^2 \ms{A} - ( n - 1 ) \rho \mi{L}_\rho \ms{A} + \rho^2 \gv^{ab} \Dv_{ab} \ms{A} + \rho^2 \, \sch{ \gv^{-1}, \Lv, \mi{L}_\rho \ms{A} } + \rho \, \sch{ \gv^{-1}, \Lv, \ms{A} } \\
\notag &\qquad + \rho^2 \, \sch{ \gv^{-1}, \mi{L}_\rho \Lv, \ms{A} } + \rho^2 \, \sch{ \gv^{-2}, \Lv, \Lv, \ms{A} } \text{,}
\end{align*}
which is precisely the last part of \eqref{eq.aads_boxm}.

\subsection{Proof of Proposition \ref{thm.aads_decomp}} \label{sec.aads_decomp}

We assume all indices are with respect to $\varphi$- and $\varphi_\rho$-coordinates, and we let $\Gamma$ and $\Gammav$ be the corresponding Christoffel symbols, as defined in Proposition \ref{thm.aads_Gamma}.
For future convenience, we also set $l := l_1 + l_2$, and we define (via indices) the vertical tensor fields
\begin{equation}
\label{eql.aads_decomp_0} \kv_{ a c } := \rho^{-1} \gv_{ a c } - \frac{1}{2} \Lv_{ a c } \text{,} \qquad \kv^b{}_a := \rho^{-1} \delta^b{}_a - \frac{1}{2} \gv^{ b c } \Lv_{ a c } \text{.}
\end{equation}

By the definitions of $\nabla$ and $\Dvm$ (see Proposition \ref{thm.aads_vertical_connection} and \eqref{eq.aads_deriv}), along with \eqref{eq.aads_Gamma}, we have
\[
\nabla_\rho F_{ \ix{\rho} \ix{a} } = \partial_\rho ( F_{ \ix{\rho} \ix{a} } ) - \sum_{ i = 1 }^{ l_1 } \Gamma^\rho_{ \rho \rho } \, F_{ \ix{\rho} \ix{a} } - \sum_{ j = 1 }^{ l_2 } \Gamma^b_{ \rho a_j } \, F_{ \ix{\rho} \ixr{a}{j}{b} } \text{,} \qquad \Dvm_\rho \ms{f}_{ \ix{a} } = \partial_\rho ( \ms{f}_{ \ix{a} } ) - \sum_{ j = 1 }^{ l_2 } \Gammav^b_{ \rho a_j } \ms{f}_{ \ixr{a}{j}{b} } \text{.}
\]
Subtracting the two equations above and applying \eqref{eq.aads_Gamma} and \eqref{eq.aads_decomp_phi} yields
\begin{align*}
\nabla_\rho F_{ \ix{\rho} \ix{a} } &= \Dvm_\rho \ms{f}_{ \ix{a} } - \sum_{ i = 1 }^{ l_1 } \Gamma^\rho_{ \rho \rho } \, \ms{f}_{ \ix{a} } - \sum_{ j = 1 }^{ l_2 } ( \Gamma^b_{ \rho a_j } - \Gammav^b_{ \rho a_j } ) \ms{f}_{ \ixr{a}{j}{b} } \\
&= \Dvm_\rho \ms{f}_{ \ix{a} } + l_1 \rho^{-1} \, \ms{f}_{ \ix{a} } + l_2 \rho^{-1} \, \ms{f}_{ \ix{a} } \\
&= \Dvm_\rho \ms{f}_{ \ix{a} } + ( l_1 + l_2 ) \rho^{-1} \, \ms{f}_{ \ix{a} } \text{,}
\end{align*}
from which the first equation in \eqref{eq.aads_decomp} follows.

Similarly, the definitions of $\nabla$ and $\Dvm$ imply
\[
\nabla_c F_{ \ix{\rho} \ix{a} } = \partial_c ( F_{ \ix{\rho} \ix{a} } ) - \sum_{ i = 1 }^{ l_1 } \Gamma_{ c \rho }^b \, F_{ \ixr{\rho}{i}{b} \ix{a} } - \sum_{ j = 1 }^{ l_2 } \Gamma_{ c a_j }^\beta \, F_{ \ix{\rho} \ixr{a}{j}{\beta} } \text{,} \qquad \Dvm_c \ms{f}_{ \ix{a} } = \partial_c ( \ms{f}_{ \ix{a} } ) - \sum_{ j = 1 }^{ l_2 } \Gammav_{ c a_j }^b \, \ms{f}_{ \ixr{a}{j}{b} } \text{.}
\]
Subtracting the above equations and recalling \eqref{eq.aads_Gamma}, \eqref{eq.aads_decomp_phi_rv}, and \eqref{eq.aads_decomp_phi_vr}, we obtain
\begin{align}
\label{eql.aads_decomp_00} \nabla_c F_{ \ix{\rho} \ix{a} } &= \Dvm_c \ms{f}_{ \ix{a} } - \sum_{ i = 1 }^{ l_1 } \Gamma_{ c \rho }^b \, F_{ \ixr{\rho}{i}{b} \ix{a} } - \sum_{ j = 1 }^{ l_2 } \Gamma_{ c a_j }^\rho \, F_{ \ix{\rho} \ixr{a}{j}{\rho} } \\
\notag &= \Dvm_c \ms{f}_{ \ix{a} } + \sum_{ i = 1 }^{ l_1 } \kv^b{}_c \, ( \ms{f}^\rho_i )_{ b \ix{a} } - \sum_{ j = 1 }^{ l_2 } \kv_{ c a_j } \, ( \ms{f}^v_j )_{ \ixd{a}{j} } \text{.}
\end{align}
Combining the above with \eqref{eql.aads_decomp_0} yields the second part of \eqref{eq.aads_decomp}.

For the last identity in \eqref{eq.aads_decomp}, we begin with $\rho$-derivatives.
First, by Proposition \ref{thm.aads_Gamma},
\begin{align}
\label{eql.aads_decomp_1} \nabla_{ \rho \rho } F_{ \ix{\rho} \ix{a} } &= \partial_\rho ( \nabla_\rho F_{ \ix{\rho} \ix{a} } ) - \Gamma^\rho_{ \rho \rho } \, \nabla_\rho F_{ \ix{\rho} \ix{a} } - \sum_{ i = 1 }^{ l_1 } \Gamma^\rho_{ \rho \rho } \, \nabla_\rho F_{ \ix{\rho} \ix{a} } - \sum_{ j = 1 }^{ l_2 } \Gamma^b_{ \rho a_j } \, \nabla_\rho F_{ \ix{\rho} \ixr{a}{j}{b} } \text{,} \\
\notag &= \partial_\rho ( \nabla_\rho F_{ \ix{\rho} \ix{a} } ) + ( l_1 + 1 ) \rho^{-1} \, \nabla_\rho F_{ \ix{\rho} \ix{a} } - \sum_{ j = 1 }^{ l_2 } \Gamma^b_{ \rho a_j } \, \nabla_\rho F_{ \ix{\rho} \ixr{a}{j}{b} } \text{,}
\end{align}
Similarly, for the corresponding mixed derivatives, we apply \eqref{eq.aads_deriv} and compute
\begin{align}
\label{eql.aads_decomp_2} \rho^{ -l } \nablam_{ \rho \rho } ( \rho^l \ms{f} )_{ \ix{a} } &= \rho^{ -l } \partial_\rho [ \Dvm_\rho ( \rho^l \ms{f} )_{ \ix{a} } ] - \Gamma^\rho_{ \rho \rho } \, \rho^{ -l } \Dvm_\rho  ( \rho^l \ms{f} )_{ \ix{a} } - \sum_{ j = 1 }^{ l_2 } \Gammav^b_{ \rho a_j } \, \rho^{ -l } \Dvm_\rho ( \rho^l \ms{f} )_{ \ixr{a}{j}{b} } \\
\notag &= \partial_\rho [ \rho^{ -l } \Dvm_\rho ( \rho^l \ms{f} )_{ \ix{a} } ] + ( l + 1 ) \rho^{-1} \, \rho^{ -l } \Dvm_\rho  ( \rho^l \ms{f} )_{ \ix{a} } - \sum_{ j = 1 }^{ l_2 } \Gammav^b_{ \rho a_j } \, \rho^{ -l } \Dvm_\rho ( \rho^l \ms{f} )_{ \ixr{a}{j}{b} } \text{,}
\end{align}
where we also recalled \eqref{eq.aads_Gamma} and the basic properties from Proposition \ref{thm.aads_mixed_connection}.
Subtracting \eqref{eql.aads_decomp_2} from \eqref{eql.aads_decomp_1}, while applying both \eqref{eq.aads_Gamma} and the first part of \eqref{eq.aads_decomp}, we obtain that
\begin{align}
\label{eql.aads_decomp_10} \nabla_{ \rho \rho } F_{ \ix{\rho} \ix{a} } &= \rho^{ -l } \nablam_{ \rho \rho } ( \rho^l \ms{f} )_{ \ix{a} } + ( l_1 - l ) \rho^{-1} \, \rho^{ -l } \Dvm_\rho ( \rho^l \ms{f} )_{ \ix{a} } - \sum_{ j = 1 }^{ l_2 } ( \Gamma^b_{ \rho a_j } - \Gammav^b_{ \rho a_j } ) \, \rho^{ -l } \Dvm_\rho ( \rho^l \ms{f} )_{ \ixr{a}{j}{b} } \\
\notag &= \rho^{ -l } \nablam_{ \rho \rho } ( \rho^l \ms{f} )_{ \ix{a} } + ( l_1 - l ) \rho^{-1} \, \rho^{ -l } \Dvm_\rho ( \rho^l \ms{f} )_{ \ix{a} } - l_2 \rho^{-1} \, \rho^{ -l } \Dvm_\rho ( \rho^l \ms{f} )_{ \ix{a} } \\
\notag &= \rho^{ -l } \nablam_{ \rho \rho } ( \rho^l \ms{f} )_{ \ix{a} } \text{.}
\end{align}

Next, applying again \eqref{eq.aads_deriv}, we compute
\begin{align*}
\nabla_{ b c } F_{ \ix{\rho} \ix{a} } &= \partial_b ( \nabla_c F_{ \ix{\rho} \ix{a} } ) - \Gamma^\alpha_{ b c } \, \nabla_\alpha F_{ \ix{\rho} \ix{a} } - \sum_{ i = 1 }^{ l_1 } \Gamma^d_{ b \rho } \, \nabla_c F_{ \ixr{\rho}{i}{d} \ix{a} } - \sum_{ j = 1 }^{ l_2 } \Gamma^\delta_{ b a_j } \, \nabla_c F_{ \ix{\rho} \ixr{a}{j}{\delta} } \text{,} \\
\rho^{-l} \nablam_{ b c } ( \rho^l \ms{f} )_{ \ix{a} } &= \partial_b ( \Dvm_c \ms{f}_{ \ix{a} } ) - \Gamma^\alpha_{ b c } \, \rho^{-l} \Dvm_\alpha ( \rho^l \ms{f} )_{ \ix{a} } - \sum_{ j = 1 }^{ l_2 } \Gammav^d_{ b a_j } \, \Dvm_c \ms{f}_{ \ixr{a}{j}{d} } \text{.}
\end{align*}
Subtracting the two equations and recalling \eqref{eq.aads_Gamma} then yields
\begin{align}
\label{eql.aads_decomp_11} \nabla_{ b c } F_{ \ix{\rho} \ix{a} } &= \rho^{-l} \nablam_{ b c } ( \rho^l \ms{f} )_{ \ix{a} } + \partial_b ( \nabla_c F_{ \ix{\rho} \ix{a} } - \Dvm_c \ms{f}_{ \ix{a} } ) - \Gamma^\alpha_{ b c } [ \nabla_\alpha F_{ \ix{\rho} \ix{a} } - \rho^{-l} \Dvm_\alpha ( \rho^l \ms{f} )_{ \ix{a} } ] \\
\notag &\qquad - \sum_{ i = 1 }^{ l_1 } \Gamma^d_{ b \rho } \, \nabla_c F_{ \ixr{\rho}{i}{d} \ix{a} } - \sum_{ j = 1 }^{ l_2 } \Gamma^\rho_{ b a_j } \, \nabla_c F_{ \ix{\rho} \ixr{a}{j}{\rho} } - \sum_{ j = 1 }^{ l_2 } \Gammav^d_{ b a_j } \, ( \nabla_c F_{ \ix{\rho} \ix{a} } - \Dvm_c \ms{f}_{ \ix{a} } ) \\
\notag &:= \rho^{-l} \nablam_{ b c } ( \rho^l \ms{f} )_{ \ix{a} } + I_1 + I_2 + I_3 + I_4 + I_5 \text{.}
\end{align}
To simplify the upcoming computations, we define, for all $1 \leq i \leq l_1$ and $1 \leq j \leq l_2$, the vertical tensor fields $\ms{z}$, $\ms{z}^\rho_i$, $\ms{z}^v_j$---of ranks $( 0, l_2 + 1 )$, $( 0, l_2 + 2 )$, $( 0, l_2 )$, respectively---via the index formulas
\begin{align}
\label{eql.aads_decomp_20} \ms{z}_{ c \ix{a} } &:= \nabla_c F_{ \ix{\rho} \ix{a} } - \Dvm_c \ms{f}_{ \ix{a} } \text{,} \\
\notag ( \ms{z}^\rho_i )_{ c b \ix{a} } &:= \nabla_c F_{ \ixr{\rho}{i}{b} \ix{a} } - \Dvm_c ( \ms{f}^\rho_i )_{ b \ix{a} } \text{,} \\
\notag ( \ms{z}^v_j )_{ c \ixd{a}{j} } &:= \nabla_c F_{ \ix{\rho} \ixr{a}{j}{\rho} } - \Dvm_c ( \ms{f}^v_j )_{ \ixd{a}{j} } \text{.}
\end{align}

Applying \eqref{eq.aads_Gamma}, the first part of \eqref{eq.aads_decomp}, and \eqref{eql.aads_decomp_20} to the term $I_2$ from \eqref{eql.aads_decomp_11}, we obtain
\[
I_2 = - \Gamma^d_{ b c } ( \nabla_d F_{ \ix{\rho} \ix{a} } - \Dvm_d \ms{f}_{ \ix{a} } ) \text{.}
\]
From \eqref{eql.aads_decomp_11}, the first part of \eqref{eql.aads_decomp_20}, and the above, we see that
\begin{equation}
\label{eql.aads_decomp_21} I_1 + I_2 + I_5 = \Dvm_b \ms{z}_{ c \ix{a} } \text{.}
\end{equation}
Similarly, for $I_3$ and $I_4$, we again apply \eqref{eq.aads_Gamma} and \eqref{eql.aads_decomp_20}:
\begin{align}
\label{eql.aads_decomp_22} I_3 &= \sum_{ i = 1 }^{ l_1 } \kv^d{}_b \, \Dvm_c ( \ms{f}^\rho_i )_{ d \ix{a} } + \sum_{ i = 1 }^{ l_1 } \kv^d{}_b \, ( \ms{z}^\rho_i )_{ c d \ix{a} } \text{,} \\
\notag I_4 &= - \sum_{ j = 1 }^{ l_2 } \kv_{ a_j b } \, \Dvm_c ( \ms{f}^v_j )_{ \ixd{a}{j} } - \sum_{ j = 1 }^{ l_2 } \kv_{ a_j b } \, ( \ms{z}^v_j )_{ c \ixd{a}{j} } \text{.}
\end{align}

Now, recalling \eqref{eql.aads_decomp_00}, along with \eqref{eql.aads_decomp_20}, we deduce
\begin{align}
\label{eql.aads_decomp_30} \ms{z}_{ c \ix{a} } &= \sum_{ i = 1 }^{ l_1 } \kv^e{}_c \, ( \ms{f}^\rho_i )_{ e \ix{a} } - \sum_{ j = 1 }^{ l_2 } \kv_{ a_j c } \, ( \ms{f}^v_j )_{ \ixd{a}{j} } \text{,} \\
\notag ( \ms{z}^\rho_i )_{ c d \ix{a} } &= \sum_{ \substack{ 1 \leq j \leq l_1 \\ j \neq i } } \kv^e{}_c \, ( \ms{f}^{ \rho, \rho }_{ i, j } )_{ e d \ix{a} } - \sum_{ j = 1 }^{ l_2 } \kv_{ a_j c } \, ( \ms{f}^{ \rho, v }_{ i, j } )_{ d \ixd{a}{j} } - \kv_{ b c } \, \ms{f}_{ \ix{a} } \text{,} \\
\notag ( \ms{z}^v_j )_{ c \ixd{a}{j} } &= \sum_{ i = 1 }^{ l_1 } \kv^e{}_c \, ( \ms{f}^{ \rho, v }_{ i, j } )_{ e \ixd{a}{j} } + \kv^e{}_c \, \ms{f}_{ \ixr{a}{j}{e} } - \sum_{ \substack{ 1 \leq i \leq l_2 \\ i \neq j } } \kv_{ a_i c } \, ( \ms{f}^{ v, v }_{ i, j } )_{ \ixd{a}{i,j} } \text{.}
\end{align}
Combining \eqref{eql.aads_decomp_0}, \eqref{eql.aads_decomp_21}, and the above, we conclude that
\begin{align}
\label{eql.aads_decomp_31} I_1 + I_2 + I_5 &= \sum_{ i = 1 }^{ l_1 } [ \kv^e{}_c \, \Dvm_b ( \ms{f}^\rho_i )_{ e \ix{a} } + \Dvm_b \kv^d{}_c \, ( \ms{f}^\rho_i )_{ e \ix{a} } ] - \sum_{ j = 1 }^{ l_2 } [ \kv_{ a_j c } \, \Dvm_b ( \ms{f}^v_j )_{ \ixd{a}{j} } + \Dvm_b \kv_{ a_j c } \, ( \ms{f}^v_j )_{ \ixd{a}{j} } ] \\
\notag &= \rho^{-1} \sum_{ i = 1 }^{ l_1 } \Dvm_b ( \ms{f}^\rho_i )_{ c \ix{a} } - \rho^{-1} \sum_{ j = 1 }^{ l_2 } \gv_{ a_j c } \, \Dvm_b ( \ms{f}^v_j )_{ \ixd{a}{j} } + \sum_{ i = 1 }^{ l_1 } \mi{S} ( \gv^{-1}, \Lv, \Dvm f^\rho_i )_{ c b \ix{a} } \\
\notag &\qquad + \sum_{ j = 1 }^{ l_2 } \mi{S} ( \Lv, \Dvm \ms{f}^v_j )_{ c b \ix{a} } + \sum_{ i = 1 }^{ l_1 } \mi{S} ( \gv^{-1}, \Dv \Lv, \ms{f}^\rho_i )_{ c b \ix{a} } + \sum_{ j = 1 }^{ l_2 } \mi{S} ( \Dv \Lv, \ms{f}^v_j )_{ c b \ix{a} } \text{.}
\end{align}
Similar computations using \eqref{eql.aads_decomp_22} also yield
\begin{align*}
I_3 &= \sum_{ i = 1 }^{ l_1 } \kv^d{}_b \, \Dvm_c ( \ms{f}^\rho_i )_{ d \ix{a} } + 2 \sum_{ 1 \leq i < j \leq l_1 } \kv^d{}_b \kv^e{}_c \, ( \ms{f}^{ \rho, \rho }_{ i, j } )_{ e d \ix{a} } - \sum_{ i = 1 }^{ l_1 } \sum_{ j = 1 }^{ l_2 } \kv^d{}_b \kv_{ a_j c } \, ( \ms{f}^{ \rho, v }_{ i, j } )_{ d \ixd{a}{j} } - l_1 \kv^d{}_b \kv_{ d c } \, \ms{f}_{ \ix{a} } \\
&= \rho^{-1} \sum_{ i = 1 }^{ l_1 } \Dvm_c ( \ms{f}^\rho_i )_{ b \ix{a} } - l_1 \rho^{-2} \gv_{ b c } \, \ms{f}_{ \ix{a} } + 2 \rho^{-2} \sum_{ 1 \leq i < j \leq l_1 } ( \ms{f}^{ \rho, \rho }_{ i, j } )_{ c b \ix{a} } - \rho^{-2} \sum_{ i = 1 }^{ l_1 } \sum_{ j = 1 }^{ l_2 } \gv_{ a_j c } \, ( \ms{f}^{ \rho, v }_{ i, j } )_{ b \ixd{a}{j} } \text{,} \\
&\qquad + \sum_{ i = 1 }^{ l_1 } \mi{S} ( \gv^{-1}, \Lv, \Dvm \ms{f}^\rho_i )_{ c b \ix{a} } + \rho^{-1} \mi{S} ( \Lv, \ms{f} )_{ c b \ix{a} } + \rho^{-1} \sum_{ 1 \leq i < j \leq l_1 } \mi{S} ( \gv^{-1}, \Lv, \ms{f}^{ \rho, \rho }_{ i, j } )_{ c b \ix{a} } \\
&\qquad + \mi{S} ( \gv^{-1}, \Lv, \Lv, \ms{f} )_{ c b \ix{a} } + \sum_{ 1 \leq i < j \leq l_1 } \mi{S} ( \gv^{-2}, \Lv, \Lv, \ms{f}^{ \rho, \rho }_{ i, j } )_{ c b \ix{a} } + \rho^{-1} \sum_{ i = 1 }^{ l_1 } \sum_{ j = 1 }^{ l_2 } \mi{S} ( \Lv, \ms{f}^{ \rho, v }_{ i, j } )_{ c b \ix{a} } \\
&\qquad + \rho^{-1} \sum_{ i = 1 }^{ l_1 } \sum_{ j = 1 }^{ l_2 } \mi{S} ( \gv, \gv^{-1}, \Lv, \ms{f}^{ \rho, v }_{ i, j } )_{ c b \ix{a} } + \sum_{ i = 1 }^{ l_1 } \sum_{ j = 1 }^{ l_2 } \mi{S} ( \gv^{-1}, \Lv, \Lv, \ms{f}^{ \rho, v }_{ i, j } )_{ c b \ix{a} } \text{,} \\
I_4 &= - \sum_{ j = 1 }^{ l_2 } \kv_{ a_j b } \, \Dvm_c ( \ms{f}^v_j )_{ \ixd{a}{j} } - \sum_{ i = 1 }^{ l_1 } \sum_{ j = 1 }^{ l_2 } \kv_{ a_j b } \kv^e{}_c \, ( \ms{f}^{ \rho, v }_{ i, j } )_{ e \ixd{a}{j} } - \sum_{ j = 1 }^{ l_2 } \kv_{ a_j b } \kv^e{}_c \, \ms{f}_{ \ixr{a}{j}{e} } \\
&\qquad + 2 \sum_{ 1 \leq i < j \leq l_2 } \kv_{ a_i c } \kv_{ a_j b } \, ( \ms{f}^{ v, v }_{ i, j } )_{ \ixd{a}{i,j} } \\
&= - \rho^{-1} \sum_{ j = 1 }^{ l_2 } \gv_{ a_j b } \, \Dvm_c ( \ms{f}^v_j )_{ \ixd{a}{j} } - \rho^{-2} \sum_{ i = 1 }^{ l_1 } \sum_{ j = 1 }^{ l_2 } \gv_{ a_j b } \, ( \ms{f}^{ \rho, v }_{ i, j } )_{ c \ixd{a}{j} } - \rho^{-2} \sum_{ j = 1 }^{ l_2 } \gv_{ a_j b } \, \ms{f}_{ \ixr{a}{j}{c} } \\
&\qquad + 2 \rho^{-2} \sum_{ 1 \leq i < j \leq l_2 } \gv_{ a_i c } \gv_{ a_j b } \, ( \ms{f}^{ v, v }_{ i, j } )_{ \ixd{a}{i,j} } + \sum_{ j = 1 }^{ l_2 } \mi{S} ( \Lv, \Dv \ms{f}^v_j )_{ c b \ix{a} } + \rho^{-1} \mi{S} ( \Lv, \ms{f} )_{ c b \ix{a} } \\
&\qquad + \mi{S} ( \gv^{-1}, \Lv, \Lv, \ms{f} )_{ c b \ix{a} } + \rho^{-1} \sum_{ 1 \leq i < j \leq l_2 } \mi{S} ( \gv, \Lv, \ms{f}^{ v, v }_{ i, j } )_{ c b \ix{a} } + \sum_{ 1 \leq i < j \leq l_2 } \mi{S} ( \Lv, \Lv, \ms{f}^{ v, v }_{ i, j } )_{ c b \ix{a} } \\
&\qquad + \rho^{-1} \sum_{ i = 1 }^{ l_1 } \sum_{ j = 1 }^{ l_2 } \mi{S} ( \Lv, \ms{f}^{ \rho, v }_{ i, j } )_{ c b \ix{a} } + \rho^{-1} \sum_{ i = 1 }^{ l_1 } \sum_{ j = 1 }^{ l_2 } \mi{S} ( \gv, \gv^{-1}, \Lv, \ms{f}^{ \rho, v }_{ i, j } )_{ c b \ix{a} } \\
&\qquad + \sum_{ i = 1 }^{ l_1 } \sum_{ j = 1 }^{ l_2 } \mi{S} ( \gv^{-1}, \Lv, \Lv, \ms{f}^{ \rho, v }_{ i, j } )_{ c b \ix{a} } \text{.}
\end{align*}

Finally, combining \eqref{eql.aads_decomp_11}, \eqref{eql.aads_decomp_31}, and the above, we obtain
\begin{align*}
\gv^{ b c } \, \nabla_{ b c } F_{ \ix{\rho} \ix{a} } &= \rho^{-l} \gv^{ b c } \nablam_{ b c } ( \rho^l \ms{f} )_{ \ix{a} } + 2 \rho^{-1} \left[ \sum_{ i = 1 }^{ l_1 } \gv^{ b c } \Dvm_b ( \ms{f}^\rho_i )_{ c \ix{a} } - \sum_{ j = 1 }^{ l_2 } \Dvm_{ a_j } ( \ms{f}^v_j )_{ \ixd{a}{j} } - ( n l_1 + l_2 ) \rho^{-2} \, \ms{f}_{ \ix{a} } \right] \\
&\qquad + 2 \rho^{-2} \left[ \sum_{ 1 \leq i < j \leq l_1 } \gv^{ b c } \, ( \ms{f}^{ \rho, \rho }_{ i, j } )_{ c b \ix{a} } - \sum_{ i = 1 }^{ l_1 } \sum_{ j = 1 }^{ l_2 } ( \ms{f}^{ \rho, v }_{ i, j } )_{ a_j \ixd{a}{j} } + \sum_{ 1 \leq i < j \leq l_2 } \gv_{ a_i a_j } \, ( \ms{f}^{ v, v }_{ i, j } )_{ \ixd{a}{i,j} } \right] \\
&\qquad + \sum_{ i = 1 }^{ l_1 } \mi{S} ( \gv^{-2}, \Lv, \Dvm f^\rho_i )_{ \ix{a} } + \sum_{ j = 1 }^{ l_2 } \mi{S} ( \gv^{-1}, \Lv, \Dvm f^v_j )_{ \ix{a} } + \sum_{ i = 1 }^{ l_2 } \mi{S} ( \gv^{-2}, \Dv \Lv, f^\rho_i )_{ \ix{a} } \\
&\qquad + \sum_{ j = 1 }^{ l_2 } \mi{S} ( \gv^{-1}, \Dv \Lv, f^v_j )_{ \ix{a} } + \rho^{-1} \mi{S} ( \gv^{-1}, \Lv, \ms{f} )_{ \ix{a} } + \mi{S} ( \gv^{-2}, \Lv, \Lv, \ms{f} )_{ \ix{a} } \\
&\qquad + \rho^{-1} \sum_{ 1 \leq i < j \leq l_1 } \mi{S} ( \gv^{-2}, \Lv, \ms{f}^{ \rho, \rho }_{ i, j } )_{ \ix{a} } + \sum_{ 1 \leq i < j \leq l_1 } \mi{S} ( \gv^{-3}, \Lv, \Lv, \ms{f}^{ \rho, \rho }_{ i, j } )_{ \ix{a} } \\
&\qquad + \rho^{-1} \sum_{ 1 \leq i < j \leq l_2 } \mi{S} ( \gv, \gv^{-1}, \Lv, \ms{f}^{ v, v }_{ i, j } )_{ \ix{a} } + \sum_{ 1 \leq i < j \leq l_2 } \mi{S} ( \gv^{-1}, \Lv, \Lv, \ms{f}^{ v, v }_{ i, j } )_{ \ix{a} } \\
&\qquad + \rho^{-1} \sum_{ i = 1 }^{ l_1 } \sum_{ j = 1 }^{ l_2 } \mi{S} ( \gv^{-1}, \Lv, \ms{f}^{ \rho, v }_{ i, j } )_{ \ix{a} } + \rho^{-1} \sum_{ i = 1 }^{ l_1 } \sum_{ j = 1 }^{ l_2 } \mi{S} ( \gv, \gv^{-2}, \Lv, \ms{f}^{ \rho, v }_{ i, j } )_{ \ix{a} } \\
&\qquad + \sum_{ i = 1 }^{ l_1 } \sum_{ j = 1 }^{ l_2 } \mi{S} ( \gv^{-2}, \Lv, \Lv, \ms{f}^{ \rho, v }_{ i, j } )_{ \ix{a} } \text{.}
\end{align*}
The last formula in \eqref{eq.aads_decomp} now follows from \eqref{eql.aads_decomp_10}, the above, and the fact that
\[
\Box F = \rho^2 ( \nabla_{ \rho \rho } F + \gv^{ b c } \nabla_{ b c } F ) \text{,} \qquad \rho^{-1} \Boxm ( \rho^l \ms{f} ) = \rho^2 [ \rho^{-l} \nablam_{ \rho \rho } ( \rho^l \ms{f} ) + \gv^{ b c } \rho^{-l} \nablam_{ b c } ( \rho^l \ms{f} ) ] \text{.}
\]

\subsection{Proof of Proposition \ref{thm.aads_weyl_pre}} \label{sec.aads_weyl_pre}

Throughout this proof, we assume all indices are raised and lowered using $g$.
First, the Bianchi identity and Proposition \ref{thm.aads_einstein} yield
\begin{align}
\label{eql.aads_weyl_pre_0} \nabla_\mu W_{ \alpha \beta \gamma \delta } + \nabla_\gamma W_{ \alpha \beta \delta \mu } + \nabla_\delta W_{ \alpha \beta \mu \gamma } &= 0 \text{,} \\
\notag \nabla^\mu W_{ \mu \beta \gamma \delta } = \nabla^\mu W_{ \mu \beta \gamma \delta } + \nabla_\gamma W_{ \mu \beta \delta }{}^\mu + \nabla_\delta W_{ \mu \beta }{}^\mu{}_\gamma &= 0 \text{,}
\end{align}
which immediately proves the first identity in \eqref{eq.aads_weyl_pre}.

Next, taking a divergence of the first part of \eqref{eql.aads_weyl_pre_0} (in the ``$\mu$" component) yields
\begin{equation}
\label{eql.aads_weyl_pre_1} 0 = \Box \mc{W}_{ \alpha \beta \gamma \delta } + \nabla^\mu{}_\gamma W_{ \alpha \beta \delta \mu } + \nabla^\mu{}_\delta W_{ \alpha \beta \mu \gamma } \text{.}
\end{equation}
Commuting derivatives and then applying the second part of \eqref{eql.aads_weyl_pre_0}, we obtain
\begin{align}
\label{eql.aads_weyl_pre_2} \nabla^\mu{}_\gamma W_{ \alpha \beta \delta \mu } &= - R^\lambda{}_\alpha{}^\mu{}_\gamma W_{ \lambda \beta \delta \mu } - R^\lambda{}_\beta{}^\mu{}_\gamma W_{ \alpha \lambda \delta \mu } - R^\lambda{}_\delta{}^\mu{}_\gamma W_{ \alpha \beta \lambda \mu } - R^\lambda{}_\mu{}^\mu{}_\gamma W_{ \alpha \beta \delta \lambda } \\
\notag &:= I_{1, 1} + I_{1, 2} + I_{1, 3} + I_{1, 4} \text{,} \\
\notag \nabla^\mu{}_\delta W_{ \alpha \beta \mu \gamma } &= - R^\lambda{}_\alpha{}^\mu{}_\delta W_{ \lambda \beta \mu \gamma } - R^\lambda{}_\beta{}^\mu{}_\delta W_{ \alpha \lambda \mu \gamma } - R^\lambda{}_\mu{}^\mu{}_\delta W_{ \alpha \beta \lambda \gamma } - R^\lambda{}_\gamma{}^\mu{}_\delta W_{ \alpha \beta \mu \lambda } \\
\notag &:= I_{2, 1} + I_{2, 2} + I_{2, 3} + I_{2, 4} \text{.}
\end{align}

Now, using Proposition \ref{thm.aads_einstein}, we expand
\begin{align*}
I_{1, 1} + I_{2, 2} &= - ( W^\lambda{}_\alpha{}^\mu{}_\gamma + g^\lambda{}_\gamma g_\alpha{}^\mu ) W_{ \lambda \beta \delta \mu } - ( W^\lambda{}_\beta{}^\mu{}_\delta + g^\lambda{}_\delta g_\beta{}^\mu ) W_{ \alpha \lambda \mu \gamma } \\
&= - 2 W_{ \alpha \delta \beta \gamma } + W^\lambda{}_\alpha{}^\mu{}_\gamma W_{ \lambda \beta \mu \delta } + W^\lambda{}_\beta{}^\mu{}_\delta W_{ \lambda \alpha \mu \gamma } \text{,} \\
I_{1, 2} + I_{2, 1} &= - ( W^\lambda{}_\beta{}^\mu{}_\gamma + g^\lambda{}_\gamma g_\beta{}^\mu ) W_{ \alpha \lambda \delta \mu } - ( W^\lambda{}_\alpha{}^\mu{}_\delta + g^\lambda{}_\delta g_\alpha{}^\mu ) W_{ \lambda \beta \mu \gamma } \\
&= - 2 W_{ \alpha \gamma \delta \beta } - W^\lambda{}_\beta{}^\mu{}_\gamma W_{ \lambda \alpha \mu \delta } - W^\lambda{}_\alpha{}^\mu{}_\delta W_{ \lambda \beta \mu \gamma } \text{,} \\
I_{1, 3} + I_{2, 4} &= - ( W^\lambda{}_\delta{}^\mu{}_\gamma + g^\lambda{}_\gamma g_\delta{}^\mu ) W_{ \alpha \beta \lambda \mu } - ( W^\lambda{}_\gamma{}^\mu{}_\delta + g^\lambda{}_\delta g_\gamma{}^\mu ) W_{ \alpha \beta \mu \lambda } \\
&= - 2 W_{ \alpha \beta \gamma \delta } - W^\lambda{}_\delta{}^\mu{}_\gamma W_{ \alpha \beta \lambda \mu } - W^\lambda{}_\gamma{}^\mu{}_\delta W_{ \alpha \beta \mu \lambda } \text{,} \\
I_{1, 4} + I_{2, 3} &= 2 n W_{ \alpha \beta \gamma \delta } \text{.}
\end{align*}
Summing the above and recalling the symmetries of $W$, we then obtain
\begin{align*}
\sum_{ i = 1 }^2 \sum_{ j = 1 }^4 I_{ i, j } &= 2 n W_{ \alpha \beta \gamma \delta } + W^\lambda{}_\alpha{}^\mu{}_\gamma W_{ \lambda \beta \mu \delta } + W^\lambda{}_\beta{}^\mu{}_\delta W_{ \lambda \alpha \mu \gamma } \\
&\qquad - W^\lambda{}_\beta{}^\mu{}_\gamma W_{ \lambda \alpha \mu \delta } - W^\lambda{}_\alpha{}^\mu{}_\delta W_{ \lambda \beta \mu \gamma } + W^{ \lambda \mu }{}_{ \gamma \delta } W_{ \alpha \beta \lambda \mu } \text{.}
\end{align*}
Combining \eqref{eql.aads_weyl_pre_1}, \eqref{eql.aads_weyl_pre_2}, and the above results in the remaining equation of \eqref{eq.aads_weyl_pre}.

\subsection{Proof of Proposition \ref{thm.aads_bianchi}} \label{sec.aads_bianchi}

Let us begin with the first two identities in \eqref{eq.aads_bianchi}.
Applying the first equation in \eqref{eq.aads_decomp} to $F := W$ (with permuted indices) yields
\begin{align*}
\rho^2 \nabla_\rho W_{ \rho a \rho b } &= \rho^{-2} \Dvm_\rho ( \rho^2 \wv^2_{ a b } ) \\
&= \Dvm_\rho \wv^2_{ a b } + 2 \rho^{-1} \wv^2_{ a b } \text{,} \\
\notag \rho^2 \nabla_\rho W_{ \rho a b c } &= \rho^{-2} \Dvm_\rho ( \rho^2 \wv^1_{ a b c } ) \\
&= \Dvm_\rho \wv^1_{ a b c } + 2 \rho^{-1} \wv^1_{ a b c } \text{.}
\end{align*}
Applying \eqref{eq.aads_metric} and the first identity in \eqref{eq.aads_weyl_pre} to the left-hand sides of the above yields
\begin{align}
\label{eql.aads_bianchi_0} \Dvm_\rho \wv^2_{ a b } + 2 \rho^{-1} \wv^2_{ a b } &= \gv^{ c d } \nabla_c ( \rho^2 W )_{ \rho b a d } \text{,} \\
\notag \Dvm_\rho \wv^1_{ a b c } + 2 \rho^{-1} \wv^1_{ a b c } &= - \gv^{ d e } \nabla_d ( \rho^2 W )_{ e a b c } \text{.}
\end{align}

Applying the second identity of \eqref{eq.aads_decomp} to the first equation in \eqref{eql.aads_bianchi_0}, we have
\begin{align}
\label{eql.aads_bianchi_1} \Dvm_\rho \wv^2_{ a b } + 2 \rho^{-1} \wv^2_{ a b } &= \gv^{ c d } \Dv_c \wv^1_{ b a d } - \rho^{-1} \gv^{ c d } \, \wv^0_{ d a c b } + \rho^{-1} \gv^{ c d } \gv_{ c d } \, \wv^2_{ a b } - \rho^{-1} \gv^{ c d } \gv_{ c a } \, \wv^2_{ d b } \\
\notag &\qquad + \sch{ \gv^{-2}, \Lv, \wv^0 }_{ a b } + \sch{ \gv^{-1}, \Lv, \wv^2 }_{ a b } \\
\notag &= \gv^{ c d } \Dv_c \wv^1_{ b a d } + n \rho^{-1} \wv^2_{ a b } + \sch{ \gv^{-2}, \Lv, \wv^0 }_{ a b } + \sch{ \gv^{-1}, \Lv, \wv^2 }_{ a b } \text{,}
\end{align}
where in the last step, we also used Definition \ref{def.aads_weyl} and the trace-free property of $W$ to obtain
\[
- \gv^{ c d } \, \wv^0_{ d a c b } = \wv^2_{ a b } \text{.}
\]
The first part of \eqref{eq.aads_bianchi} follows immediately from \eqref{eql.aads_bianchi_1}.

Similarly, applying \eqref{eq.aads_decomp} to the second part of \eqref{eql.aads_bianchi_0} yields
\begin{align}
\label{eql.aads_bianchi_2} \Dvm_\rho \wv^1_{ a b c } + 2 \rho^{-1} \wv^1_{ a b c } &= - \gv^{ d e } \Dv_d \wv^0_{ e a b c } + \rho^{-1} \gv^{ d e } \gv_{ d e } \, \wv^1_{ a b c } - \rho^{-1} \gv^{ d e } \gv_{ d a } \, \wv^1_{ e b c } \\
\notag &\qquad + \rho^{-1} \gv^{ d e } \gv_{ d b } \, \wv^1_{ c e a } - \rho^{-1} \gv^{ d e } \gv_{ d c } \, \wv^1_{ b e a } + \sch{ \gv^{-1}, \Lv, \wv^1 }_{ a b c } \\
\notag &= - \gv^{ d e } \Dv_d \wv^0_{ e a b c } + n \rho^{-1} \, \wv^1_{ a b c } + \sch{ \gv^{-1}, \Lv, \wv^1 }_{ a b c } \text{,}
\end{align}
where in the last step, we noted from Definition \ref{def.aads_weyl} and the symmetries of $W$ that
\[
- \wv^1_{ a b c } + \wv^1_{ c b a } - \wv^1_{ b c a } = 0 \text{.}
\]
The second identity in \eqref{eq.aads_bianchi} now follows from \eqref{eql.aads_bianchi_2}.

For the two remaining parts of \eqref{eq.aads_bianchi}, we again start with the first part of \eqref{eq.aads_decomp}:
\[
\Dvm_\rho \wv^1_{ a b c } + 2 \rho^{-1} \wv^1_{ a b c } = \rho^2 \nabla_\rho W_{ \rho a b c } \text{,} \qquad \Dvm_\rho \wv^0_{ a b c d } + 2 \rho^{-1} \wv^0_{ a b c d } = \rho^2 \nabla_\rho W_{ a b c d } \text{.}
\]
The right-hand sides of the above can then be expanded using the Bianchi identity for $W$:
\begin{align*}
\Dvm_\rho \wv^1_{ a b c } + 2 \rho^{-1} \wv^1_{ a b c } &= \nabla_b ( \rho^2 W )_{ \rho a \rho c } - \nabla_c ( \rho^2 W )_{ \rho a \rho b } \text{,} \\
\Dvm_\rho \wv^0_{ a b c d } + 2 \rho^{-1} \wv^0_{ a b c d } &= \nabla_a ( \rho^2 W )_{ \rho b c d } - \nabla_b ( \rho^2 W )_{ \rho a c d } \text{.}
\end{align*}
Applying \eqref{eq.aads_decomp} to each term on the right-hand side of the above, we then obtain
\begin{align*}
\Dvm_\rho \wv^1_{ a b c } + 2 \rho^{-1} \wv^1_{ a b c } &= \Dv_b \wv^2_{ a c } + \rho^{-1} \wv^1_{ c b a } + \rho^{-1} \wv^1_{ a b c } - \Dv_c \wv^2_{ a b } - \rho^{-1} \wv^1_{ b c a } - \rho^{-1} \wv^1_{ a c b } \\
&\qquad + \sch{ \gv^{-1}, \Lv, \wv^1 }_{ a b c } \\
&= \Dv_b \wv^2_{ a c } - \Dv_c \wv^2_{ a b } + 3 \rho^{-1} \wv^1_{ a b c } + \sch{ \gv^{-1}, \Lv, \wv^1 }_{ a b c } \text{,} \\
\Dvm_\rho \wv^0_{ a b c d } + 2 \rho^{-1} \wv^0_{ a b c d } &= \Dv_a \wv^1_{ b c d } + \rho^{-1} \wv^0_{ a b c d } - \rho^{-1} \gv_{ a c } \, \wv^2_{ b d } + \rho^{-1} \gv_{ a d } \, \wv^2_{ b c } - \Dv_b \wv^1_{ a c d } - \rho^{-1} \wv^0_{ b a c d } \\
&\qquad + \rho^{-1} \gv_{ b c } \, \wv^2_{ a d } - \rho^{-1} \gv_{ b d } \, \wv^2_{ a c } + \sch{ \gv^{-1}, \Lv, \wv^0 }_{ a b c d } + \sch{ \Lv, \wv^2 }_{ a b c d } \\
&= \Dv_a \wv^1_{ b c d } - \Dv_b \wv^1_{ a c d } + 2 \rho^{-1} \wv^0_{ a b c d } + \rho^{-1} \gv_{ a d } \, \wv^2_{ b c } + \rho^{-1} \gv_{ b c } \, \wv^2_{ a d } \\
&\qquad - \rho^{-1} \gv_{ a c } \, \wv^2_{ b d } - \rho^{-1} \gv_{ b d } \, \wv^2_{ a c } + \sch{ \gv^{-1}, \Lv, \wv^0 }_{ a b c d } + \sch{ \Lv, \wv^2 }_{ a b c d } \text{.}
\end{align*}
The last two identities of \eqref{eq.aads_bianchi} are now immediate consequences of the above.

\subsection{Proof of Proposition \ref{thm.aads_wave}} \label{sec.aads_wave}

Throughout, we assume all indices are with respect to $\varphi$ and $\varphi_\rho$, for some arbitrary coordinates $( U, \varphi )$ on $\mi{I}$.
Also, for convenience, we define $Q$ by
\begin{equation}
\label{eql.aads_wave_0} Q_{ \alpha \beta \gamma \delta } := g^{ \lambda \kappa } g^{ \mu \nu } ( 2 W_{ \lambda \alpha \mu \delta } W_{ \kappa \beta \nu \gamma } - 2 W_{ \lambda \alpha \mu \gamma } W_{ \kappa \beta \nu \delta } - W_{ \lambda \mu \gamma \delta } W_{ \alpha \beta \kappa \nu } ) \text{,}
\end{equation}
that is, the right-hand side of the wave equation in \eqref{eq.aads_weyl_pre}.
Then, by \eqref{eq.aads_metric} and \eqref{eq.aads_weyl}, we have
\begin{align}
\label{eql.aads_wave_1} Q_{ \rho b \rho d } &= g^{ \lambda \kappa } g^{ \mu \nu } ( 2 W_{ \lambda \rho \mu d } W_{ \kappa b \nu \rho } - 2 W_{ \lambda \rho \mu \rho } W_{ \kappa b \nu d } - W_{ \lambda \mu \rho d } W_{ \rho b \kappa \nu } ) \\
\notag &= \sch{ \gv^{-1}, \wv^2, \wv^2 } + \sch{ \gv^{-2}, \wv^1, \wv^1 } + \sch{ \gv^{-2}, \wv^0, \wv^2 } \text{,} \\
\notag Q_{ \rho b c d } &= g^{ \lambda \kappa } g^{ \mu \nu } ( 2 W_{ \lambda \rho \mu d } W_{ \kappa b \nu c } - 2 W_{ \lambda \rho \mu c } W_{ \kappa b \nu d } - W_{ \lambda \mu c d } W_{ \rho b \kappa \nu } ) \\
\notag &= \sch{ \gv^{-1}, \wv^1, \wv^2 } + \sch{ \gv^{-2}, \wv^0, \wv^1 } \text{,} \\
\notag Q_{ a b c d } &= g^{ \lambda \kappa } g^{ \mu \nu } ( 2 W_{ \lambda a \mu d } W_{ \kappa b \nu c } - 2 W_{ \lambda a \mu c } W_{ \kappa b \nu d } - W_{ \lambda \mu c d } W_{ a b \kappa \nu } ) \\
\notag &= \sch{ \wv^2, \wv^2 } + \sch{ \gv^{-1}, \wv^1, \wv^1 } + \sch{ \gv^{-2}, \wv^0, \wv^0 } \text{.}
\end{align}

We now apply \eqref{eq.aads_weyl} and the last equation in \eqref{eq.aads_decomp} to $W$ to obtain
\begin{align}
\label{eql.aads_wave_20} \rho^2 \Box W_{ \rho a \rho b } &= \rho^{-2} \Boxm ( \rho^2 \wv^2 )_{ a b } + 2 \rho \gv^{ c d } ( \Dv_c \wv^1_{ b d a } + \ms{D}_c \wv^1_{ a d b } ) - 0 - ( 2 n + 2 ) \, \wv^2_{ a b } \\
\notag &\qquad + 4 \, \wv^2_{ a b } + 2 \gv^{ c d } \, \wv^0_{ c a d b } + 0 + \ms{E}^2_{ a b } \\
\notag &= \rho^{-2} \Boxm ( \rho^2 \wv^2 )_{ a b } + 2 \rho \gv^{ c d } ( \Dv_c \wv^1_{ b d a } + \ms{D}_c \wv^1_{ a d b } ) - 2 n \, \wv^2_{ a b } + \ms{E}^2_{ a b } \text{,}
\end{align}
where we used that $-\ms{w}^2$ is the $\gv$-trace of $\ms{w}^0$, and where the error terms $\ms{E}^2$ are given by
\begin{align}
\label{eql.aads_wave_21} \ms{E}^2 &= \rho^2 \, \sch{ \gv^{-2}, \Lv, \Dv \wv^1 } + \rho^2 \, \sch{ \gv^{-2}, \Dv \Lv, \wv^1 } + \rho \, \sch{ \gv^{-1}, \Lv, \wv^2 } + \rho \, \sch{ \gv^{-2}, \Lv, \wv^0 } \\
\notag &\qquad + \rho^2 \, \sch{ \gv^{-2}, \Lv, \Lv, \wv^2 } + \rho^2 \, \sch{ \gv^{-3}, \Lv, \Lv, \wv^0 } + \rho \, \sch{ \gv, \gv^{-2}, \Lv, \wv^2 } \text{,}
\end{align}
Applying the second part of \eqref{eq.aads_weyl_pre} and the first part of \eqref{eq.aads_bianchi} to \eqref{eql.aads_wave_20} yields
\begin{align*}
\rho^{-2} \Boxm ( \rho^2 \wv^2 )_{ a b } + 4 ( n - 2 ) \wv^2_{ a b } &= \rho^2 \Box W_{ \rho a \rho b } + 4 \rho \, \Dvm_\rho \wv^2_{ a b } + 2 n \, \wv^2_{ a b } + \ms{E}^2_{ a b } \\
\notag &= \rho^2 Q_{ \rho a \rho b } + 4 \rho \, \Dvm_\rho \wv^2_{ a b } + \ms{E}^2_{ a b } \text{,}
\end{align*}
Moreover, expanding the left-hand side using \eqref{eq.aads_comm_rho} (with $p = 2$), we see that
\begin{equation}
\label{eql.aads_wave_22} \Boxm \wv^2_{ a b } + 2 ( n - 2 ) \wv^2_{ a b } = \ms{E}^2_{ a b } + \rho^2 Q_{ \rho a \rho b } \text{,}
\end{equation}
The first equation in \eqref{eq.aads_wave} now follows from \eqref{eql.aads_wave_1}, \eqref{eql.aads_wave_21}, and the above.

The remaining parts of \eqref{eq.aads_wave} are treated similarly.
Again, we first apply the last part of \eqref{eq.aads_decomp}:
\begin{align}
\label{eql.aads_wave_30} \rho^2 \Box W_{ \rho a b c } &= \rho^{-2} \Boxm ( \rho^2 \wv^1 )_{ a b c } + 2 \rho \gv^{ d e } \, \Dv_d \wv^0_{ e a b c } + 2 \rho \, ( \Dv_c \wv^2_{ a b } - \Dv_b \wv^2_{ a c } ) \\
\notag &\qquad - ( n + 3 ) \, \wv^1_{ a b c } + \ms{E}^1_{ a b c } \text{,} \\
\notag \rho^2 \Box W_{ a b c d } &= \rho^{-2} \Boxm ( \rho^2 \wv^0 )_{ a b c d } - 2 \rho \, ( \Dv_a \wv^1_{ b c d } - \Dv_b \wv^1_{ a c d } + \Dv_c \wv^1_{ d a b } - \Dv_d \wv^1_{ c a b } ) \\
\notag &\qquad - 4 \, \wv^0_{ a b c d } + 2 ( \gv_{ a c } \, \wv^2_{ b d } - \gv_{ a d } \, \wv^2_{ b c } - \gv_{ b c } \, \wv^2_{ a d } + \gv_{ b d } \, \wv^2_{ a c } ) + \ms{E}^0_{ a b c d } \text{,}
\end{align}
where we also made use of the symmetries of $W$, and where the corresponding error terms are
\begin{align}
\label{eql.aads_wave_31} \ms{E}^1 &= \rho^2 \, \sch{ \gv^{-2}, \Lv, \Dv \wv^0 } + \rho^2 \, \sch{ \gv^{-1}, \Lv, \Dv \wv^2 } + \rho^2 \, \sch{ \gv^{-2}, \Dv \Lv, \wv^0 } + \rho^2 \, \sch{ \gv^{-1}, \Dv \Lv, \wv^2 } \\
\notag &\qquad + \rho \, \sch{ \gv^{-1}, \Lv, \wv^1 } + \rho^2 \, \sch{ \gv^{-2}, \Lv, \Lv, \wv^1 } + \rho \, \sch{ \gv, \gv^{-2}, \Lv, \wv^1 } \text{,} \\
\notag \ms{E}^0 &= \rho^2 \, \sch{ \gv^{-1}, \Lv, \Dv \wv^1 } + \rho^2 \, \sch{ \gv^{-1}, \Dv \Lv, \wv^1 } + \rho \, \sch{ \gv^{-1}, \Lv, \wv^0 } \\
\notag &\qquad + \rho^2 \, \sch{ \gv^{-2}, \Lv, \Lv, \wv^0 } + \rho \, \sch{ \gv, \gv^{-1}, \Lv, \wv^2 } + \rho^2 \, \sch{ \gv^{-1}, \Lv, \Lv, \wv^2 } \text{.}
\end{align}
The first-order terms in the right-hand sides of \eqref{eql.aads_wave_30} can be further expanded using \eqref{eq.aads_bianchi}, and the terms involving $\Boxm$ in \eqref{eql.aads_wave_32} can be expanded using \eqref{eq.aads_comm_rho}:
\begin{align}
\label{eql.aads_wave_32} \rho^2 \Box W_{ \rho a b c } &= \rho^{-2} \Boxm ( \rho^2 \wv^1 )_{ a b c } - 4 \rho \, \Dvm_\rho \wv^1_{ a b c } + ( n - 5 ) \, \wv^1_{ a b c } + \ms{E}^1_{ a b c } \\
\notag &= \Boxm \wv^1_{ a b c } - ( n + 1 ) \, \wv^1_{ a b c } + \ms{E}^1_{ a b c } \text{,} \\
\notag \rho^2 \Box W_{ a b c d } &= \rho^{-2} \Boxm ( \rho^2 \wv^0 )_{ a b c d } - 4 \rho \, \Dvm_\rho \wv^0_{ a b c d } - 4 \, \wv^0_{ a b c d } \\
\notag &\qquad + 2 ( \gv_{ a d } \, \wv^2_{ b c } + \gv_{ b c } \, \wv^2_{ a d } - \gv_{ a c } \, \wv^2_{ b d } - \gv_{ b d } \, \wv^2_{ a c } ) + \ms{E}^0_{ a b c d } \\
\notag &= \Boxm \wv^0_{ a b c d } - 2n \, \wv^0_{ a b c d } + 2 ( \gv_{ a d } \, \wv^2_{ b c } + \gv_{ b c } \, \wv^2_{ a d } - \gv_{ a c } \, \wv^2_{ b d } - \gv_{ b d } \, \wv^2_{ a c } ) + \ms{E}^0_{ a b c d } \text{.}
\end{align}

The left-hand sides of \eqref{eql.aads_wave_32} can be expanded using \eqref{eq.aads_weyl_pre}, \eqref{eql.aads_wave_0}, and \eqref{eql.aads_wave_1}:
\begin{align*}
\rho^2 \Box W_{ \rho a b c } &= - 2 n \, \wv^1_{ a b c } + \rho^2 \, Q_{ \rho a b c } \text{,} \\
\rho^2 \Box W_{ a b c d } &= - 2 n \, \wv^0_{ a b c d } + \rho^2 \, Q_{ a b c d } \text{.}
\end{align*}
Combining the above with \eqref{eql.aads_wave_1} and \eqref{eql.aads_wave_32} results in the second equation in \eqref{eq.aads_wave}, as well as
\[
\Boxm \wv^0_{ a b c d } = - 2 ( \gv_{ a d } \, \wv^2_{ b c } + \gv_{ b c } \, \wv^2_{ a d } - \gv_{ a c } \, \wv^2_{ b d } - \gv_{ b d } \, \wv^2_{ a c } ) + \ms{E}^0_{ a b c d } + \rho^2 Q_{ a b c d } \text{.}
\]
Recalling \eqref{eq.aads_wstar}, \eqref{eql.aads_wave_22}, and the above, we then see that
\begin{align*}
\Boxm \wv^\star_{ a b c d } &= \ms{E}^0_{ a b c d } - \frac{1}{ n - 2 } ( \gv_{ a d } \ms{E}^2_{ b c } + \gv_{ b c } \ms{E}^2_{ a d } - \gv_{ a c } \ms{E}^2_{ b d } - \gv_{ b d } \ms{E}_{ a c } ) + \rho^2 Q_{ a b c d } \\
&\qquad - \frac{ \rho^2 }{ n - 2 } ( \gv_{ a d } Q_{ \rho b \rho c } + \gv_{ b c } Q_{ \rho a \rho d } - \gv_{ a c } Q_{ \rho b \rho d } - \gv_{ b d } Q_{ \rho a \rho c } ) \text{.}
\end{align*}
The last equation in \eqref{eq.aads_wave} now follows from \eqref{eql.aads_wave_1}, \eqref{eql.aads_wave_21}, \eqref{eql.aads_wave_31}, and the above.

\subsection{Proof of Proposition \ref{thm.aads_O_regular}} \label{sec.aads_O_regular}

The first two parts of \eqref{eq.aads_O_g} follow immediately from Theorem \ref{thm.aads_fg}, in particular from the boundary limits for $\gv$ and $\gv^{-1}$.
The last two parts of \eqref{eq.aads_O_g} also follow from Theorem \ref{thm.aads_fg}, using the boundary limits for $\mi{L}_\rho \gv$ and $\mi{L}_\rho^2 \gv$, along with Taylor's theorem.

Next, for the Weyl curvature components, the key is the formulation \eqref{eq.aads_connection} of the Einstein-vacuum equations.
From the boundary limits in Theorem \ref{thm.aads_fg}, we see that
\begin{equation}
\label{eql.aads_O_regular_1} \rho^{-1} \mi{L}_\rho \gv = \oo{ M - 2 }{ 1 } \text{,} \qquad - \mi{L}^2_\rho \gv + \rho^{-1} \mi{L}_\rho \gv = \oo{ M - 2 }{1} = \oo{ M - 3 }{ \rho } \text{.}
\end{equation}
Moreover, writing the components of $\ms{R}$ in terms of coordinate derivatives of $\gv$ yields
\begin{equation}
\label{eql.aads_O_regular_2} \ms{R} = \oo{ M - 2 }{1} \text{.}
\end{equation}
Combining \eqref{eq.aads_connection}, \eqref{eq.aads_O_g}, \eqref{eql.aads_O_regular_1}, and \eqref{eql.aads_O_regular_2} results in the asymptotics for $\wv^0$, $\wv^1$, and $\wv^2$ in \eqref{eq.aads_O_w}.
The asymptotics for $\wv^\star$ in \eqref{eq.aads_O_w} then follow from \eqref{eq.aads_wstar} and the above.

The asymptotics for $\mi{L}_\rho \wv^0$, $\mi{L}_\rho \wv^1$, and $\mi{L}_\rho \wv^2$ are proved similarly---we apply $\mi{L}_\rho$ to both sides of \eqref{eq.aads_connection}, and we then apply the boundary limits in Theorem \ref{thm.aads_fg} to each term on the right-hand side.
In addition, we make use of the following observations:
\begin{itemize}
\item For $\mi{L}_\rho \wv^1$, we also use \eqref{eq.aads_comm} to commute $\mi{L}_\rho$ with $\Dv$.

\item For $\mi{L}_\rho \wv^0$, we express the components of $\Rv$ in terms of derivatives of $\gv$, which yields
\[
\mi{L}_\rho \Rv = \oo{ M - 3 }{1} \text{.}
\]

\item For both $\mi{L}_\rho \wv^0$ and $\mi{L}_\rho \wv^2$, we also note that
\[
\mi{L}_\rho ( \rho^{-1} \Lv ) = \oo{ M - 3 }{1} \text{.}
\]
\end{itemize}
Finally, the asymptotics for $\mi{L}_\rho \wv^\star$ once again follow from \eqref{eq.aads_wstar} and the above.

\subsection{Proof of Proposition \ref{thm.sys_O_QB}} \label{sec.sys_O_QB}

The first part of \eqref{eq.sys_O_QB} is a consequence of the following identity, which holds with respect to any coordinate system $( U, \varphi )$ on $\mi{I}$:
\[
\gv^{ab} - \check{\gv}^{ab} = - \gv^{ac} ( \gv_{cd} - \check{\gv}_{cd} ) \check{\gv}^{bd} \text{,}
\]

Next, we can view \eqref{Qdef} as a system of differential equations for the components of $\ms{Q}$, with coefficients and sources given by $\gv^{-1}$, $\gv - \check{\gv}$, and $\Lv$.
By solving these equations, we can write the components of $\ms{Q} |_\sigma$, $\sigma \in ( 0, \rho_0 ]$, as integrals in $\rho$ (from $0$ to $\sigma$) of some combination of $\gv^{-1}$, $\gv - \check{\gv}$, and $\Lv$, along with a matrix exponential factor.
Combining the above with \eqref{eq.aads_O_g} results in the second part of \eqref{eq.sys_O_QB}.
(In particular, we obtain one extra power of $\rho$ from the integral with respect to $\rho$.)
The third part of \eqref{eq.sys_O_QB} then follows from \eqref{eq.aads_O_g}, \eqref{Bdef}, and the second part of \eqref{eq.sys_O_QB}.

It remains to prove \eqref{eq.sys_diff_Gamma}.
The first part of \eqref{eq.sys_diff_Gamma} is immediate, since it is clear from the first two parts of \eqref{eq.aads_Gamma} that $\smash{\Gamma^\alpha_{ \rho \rho }, \check{\Gamma}^\alpha_{ \rho \rho }}$ and $\smash{\Gamma^\rho_{ \alpha \rho }, \check{\Gamma}^\rho_{ \alpha \rho }}$ are identical.
Moreover, the second identity in \eqref{eq.sys_diff_Gamma} follows immediately from applying the fifth part of \eqref{eq.aads_Gamma}.
Similarly, for the third part of \eqref{eq.sys_diff_Gamma}, we use the third and fourth parts of \eqref{eq.aads_Gamma} in order to obtain
\[
( \Gamma - \check{\Gamma} )^a_{ \rho b } = ( \Gammav - \check{\Gammav} )^a_{ \rho b } = \frac{1}{2} ( \gv^{ c d } \Lv_{ a d } - \check{\gv}^{cd} \check{\Lv}_{ a d } ) \text{.}
\]
Recalling the asymptotics of \eqref{eq.aads_O_g} and the first part of \eqref{eq.sys_O_QB} yields, as desired,
\begin{align*}
( \Gammav - \check{\Gammav} )^a_{ \rho b } &= \frac{1}{2} ( \gv^{cd} - \check{\gv}^{cd} ) \Lv_{ad} + \frac{1}{2} \check{\gv}^{cd} ( \Lv_{ad} - \check{\Lv}_{ad} ) \\
&= \oo{ M - 2 }{ \rho; \gv^{-1} - \check{\gv}^{-1} }^a{}_b + \oo{M}{ 1; \Lv - \check{\Lv} }^a{}_b \\
&= \oo{ M - 2 }{ \rho; \gv - \check{\gv} }^a{}_b + \oo{M}{ 1; \Lv - \check{\Lv} }^a{}_b \text{.}
\end{align*}

For the last part of \eqref{eq.sys_diff_Gamma}, note the last part of \eqref{eq.aads_Gamma} implies
\[
( \Gamma - \check{\Gamma} )^c_{ a b } = ( \Gammav - \check{\Gammav} )^c_{ a b } \text{.}
\]
Expanding the Christoffel symbols, we then obtain, via a direct computation,
\begin{align*}
( \Gammav - \check{\Gammav} )^c_{ a b } &= \frac{1}{2} \gv^{cd} ( \partial_a \gv_{db} + \partial_b \gv_{da} - \partial_d \gv_{ab} ) - \frac{1}{2} \check{\gv}^{cd} ( \partial_a \check{\gv}_{db} + \partial_b \check{\gv}_{da} - \partial_d \check{\gv}_{ab} ) \\
&= \frac{1}{2} \check{\gv}^{cd} [ \Dv_a ( \gv - \check{\gv} )_{db} + \Dv_b ( \gv - \check{\gv} )_{da} - \Dv_d ( \gv - \check{\gv} )_{ab} ] \text{,}
\end{align*}
as desired (see also \cite[Appendix D]{wald:gr}).
The final equality in \eqref{eq.sys_diff_Gamma} now follows from \eqref{eq.aads_O_g}.

\subsection{Proof of Proposition \ref{thm.sys_transport}} \label{sec.sys_transport}

The first part of \eqref{eq.sys_transport} follows from the definitions \eqref{eq.aads_sff} of $\Lv$ and $\check{\Lv}$.
For the second part of \eqref{eq.sys_transport}, we combine the asymptotics \eqref{eq.aads_O_g} with the definition \eqref{Qdef}.

Next, we take the difference of two applications of the third part of \eqref{eq.aads_connection} to obtain
\[
\mi{L}_\rho ( \Lv - \check{\Lv} )_{ab} = \rho^{-1} ( \Lv - \check{\Lv} )_{ab} - 2 ( \wv^2 - \check{\wv}^2 )_{ab} + \frac{1}{4} ( \gv^{ c d } \Lv_{ a d } \Lv_{ b c } - \check{\gv}^{ c d } \check{\Lv}_{ a d } \check{\Lv}_{ b c } ) \text{.}
\]
The last term can be treated using \eqref{eq.aads_O_g} and the first part of \eqref{eq.sys_O_QB}, and we obtain
\[
\mi{L}_\rho ( \Lv - \check{\Lv} ) = \rho^{-1} ( \Lv - \check{\Lv} ) - 2 ( \wv^2 - \check{\wv}^2 ) + \oo{M - 2}{ \rho^2; \gv - \check{\gv} } + \oo{M - 2}{ \rho; \Lv - \check{\Lv} } \text{.}
\]
The third part of \eqref{eq.sys_transport} follows immediately from the above, since \eqref{eq.aads_O_w} and \eqref{Wdef} imply
\[
\wv^2 - \check{\wv}^2 = \ms{W}^2 + \oo{ M - 3 }{ \rho; \gv - \check{\gv} } + \oo{ M - 3 }{ \rho; \ms{Q} } \text{.}
\]

For the final identity in \eqref{eq.sys_transport}, we begin by applying $\mi{L}_\rho$ to \eqref{Bdef}, and then commuting $\Dv$ and $\mi{L}_\rho$ via \eqref{eq.aads_comm}.
The error term in \eqref{eq.aads_comm} can then be expressed as asymptotic terms using \eqref{eq.aads_O_g}:
\begin{align*}
\mi{L}_\rho \ms{B}_{cab} &= \Dv_c \mi{L}_\rho ( \gv - \check{\gv} )_{ab} - \Dv_a \mi{L}_\rho ( \gv - \check{\gv} )_{cb} - \Dv_b \mi{L}_\rho \ms{Q}_{ca} \\
&\qquad + \oo{M-3}{ \rho; \gv - \check{\gv} }_{cab} + \oo{M-3}{ \rho; \ms{Q} }_{cab} \text{.}
\end{align*}
For the first three terms on the right-hand side, we apply \eqref{eq.aads_sff}, \eqref{eq.aads_O_g}, and \eqref{Qdef} to obtain
\begin{align*}
\mi{L}_\rho \ms{B}_{cab} &= \Dv_c ( \Lv - \check{\Lv} )_{ab} - \Dv_a ( \Lv - \check{\Lv} )_{cb} - \frac{1}{2} \gv^{de} \Dv_b [ \Lv_{ad} ( \gv - \check{\gv} + \ms{Q} )_{ce} ] \\
&\qquad + \frac{1}{2} \gv^{de} \Dv_b [ \Lv_{cd} ( \gv - \check{\gv} + \ms{Q} )_{ae} ] + \oo{M-3}{ \rho; \gv - \check{\gv} }_{cab} + \oo{M-3}{ \rho; \ms{Q} }_{cab} \\
&= ( \Dv_c \Lv_{ab} - \Dv_a \Lv_{cb} ) - ( \Dv_c \check{\Lv}_{ab} - \Dv_a \check{\Lv}_{cb} ) - \frac{1}{2} \gv^{de} \Lv_{ad} \Dv_b ( \gv - \check{\gv} + \ms{Q} )_{ce} \\
&\qquad + \frac{1}{2} \gv^{de} \Lv_{cd} \Dv_b ( \gv - \check{\gv} + \ms{Q} )_{ae} + \oo{M-3}{ \rho; \gv - \check{\gv} }_{cab} + \oo{M-3}{ \rho; \ms{Q} }_{cab} \text{.}
\end{align*}
Applying the first part of \eqref{eq.aads_connection}, we then have
\begin{align}
\label{eql.sys_transport_1} \mi{L}_\rho \ms{B}_{cab} &= 2 ( \wv^1 - \check{\wv}^1 )_{bac} - ( \Dv - \check{\Dv} )_c \check{\Lv}_{ab} + ( \Dv - \check{\Dv} )_a \check{\Lv}_{cb} - \frac{1}{2} \gv^{de} \Lv_{ad} \Dv_b ( \gv - \check{\gv} + \ms{Q} )_{ce} \\
\notag &\qquad + \frac{1}{2} \gv^{de} \Lv_{cd} \Dv_b ( \gv - \check{\gv} + \ms{Q} )_{ae} + \oo{M-3}{ \rho; \gv - \check{\gv} }_{cab} + \oo{M-3}{ \rho; \ms{Q} }_{cab} \text{.}
\end{align}

For the terms in \eqref{eql.sys_transport_1} involving $\Dv - \check{\Dv}$, we expand using the last part of \eqref{eq.aads_Gamma}:
\begin{align*}
- ( \Dv - \check{\Dv} )_c \check{\Lv}_{ab} + ( \Dv - \check{\Dv} )_a \check{\Lv}_{cb} &= ( \Gammav - \check{\Gammav} )^d_{ca} \check{\Lv}_{db} + ( \Gammav - \check{\Gammav} )^d_{cb} \check{\Lv}_{ad} - ( \Gammav - \check{\Gammav} )^d_{ac} \check{\Lv}_{db} - ( \Gammav - \check{\Gammav} )^d_{ab} \check{\Lv}_{cd} \\
&= \frac{1}{2} \check{\gv}^{de} \check{\Lv}_{ad} [ \Dv_c ( \gv - \check{\gv} )_{eb} + \Dv_b ( \gv - \check{\gv} )_{ec} - \Dv_e ( \gv - \check{\gv} )_{bc} ] \\
&\qquad - \frac{1}{2} \check{\gv}^{de} \check{\Lv}_{cd} [ \Dv_a ( \gv - \check{\gv} )_{eb} + \Dv_b ( \gv - \check{\gv} )_{ea} - \Dv_e ( \gv - \check{\gv} )_{ba} ] \text{.}
\end{align*}
Combining the above with \eqref{eql.sys_transport_1}, we conclude that
\begin{align}
\label{eql.sys_transport_2} \mi{L}_\rho \ms{B}_{cab} &= 2 ( \wv^1 - \check{\wv}^1 )_{bac} + \frac{1}{2} ( \check{\gv}^{de} \check{\Lv}_{ad} - \gv^{de} \Lv_{ad} ) \Dv_b ( \gv - \check{\gv} )_{ce} \\
\notag &\qquad - \frac{1}{2} ( \check{\gv}^{de} \check{\Lv}_{cd} - \gv^{de} \Lv_{cd} ) \Dv_b ( \gv - \check{\gv} )_{ae} + \frac{1}{2} \check{\gv}^{de} \check{\Lv}_{ad} [ \Dv_c ( \gv - \check{\gv} )_{eb} - \Dv_e ( \gv - \check{\gv} )_{bc} ] \\
\notag &\qquad - \frac{1}{2} \check{\gv}^{de} \check{\Lv}_{cd} [ \Dv_a ( \gv - \check{\gv} )_{eb} - \Dv_e ( \gv - \check{\gv} )_{ba} ] - \frac{1}{2} \gv^{de} \Lv_{ad} \Dv_b \ms{Q}_{ce} + \frac{1}{2} \gv^{de} \Lv_{cd} \Dv_b \ms{Q}_{ae} \\
\notag &\qquad + \oo{M-3}{ \rho; \gv - \check{\gv} }_{cab} + \oo{M-3}{ \rho; \ms{Q} }_{cab} \\
\notag &= 2 ( \wv^1 - \check{\wv}^1 )_{bac} + \mc{I}_1 + \mc{I}_2 + \mc{I}_3 + \mc{I}_4 - \frac{1}{2} \gv^{de} \Lv_{ad} \Dv_b \ms{Q}_{ce} + \frac{1}{2} \gv^{de} \Lv_{cd} \Dv_b \ms{Q}_{ae} \\
\notag &\qquad + \oo{M-3}{ \rho; \gv - \check{\gv} }_{cab} + \oo{M-3}{ \rho; \ms{Q} }_{cab} \text{.}
\end{align}

For $\mc{I}_1$ and $\mc{I}_2$, we take as ``main terms" the differences $\gv^{-1} - \check{\gv}^{-1}$ and $\Lv - \check{\Lv}$.
(In particular, we treat $\Dv \gv$ and $\Dv \check{\gv}$ as coefficients of the form $\oo{M-1}{1}$.)
Recalling \eqref{eq.aads_O_g} and \eqref{eq.sys_O_QB} then yields
\begin{align}
\label{eql.sys_transport_11} \mc{I}_1 &= \oo{ M - 2 }{ \rho; \gv - \check{\gv} }_{cab} + \oo{ M - 1 }{ 1; \Lv - \check{\Lv} }_{cab} \text{,} \\
\notag \mc{I}_2 &= \oo{ M - 2 }{ \rho; \gv - \check{\gv} }_{cab} + \oo{ M - 1 }{ 1; \Lv - \check{\Lv} }_{cab} \text{.}
\end{align}
By similar computations, we also obtain
\begin{align}
\label{eql.sys_transport_12} \mc{I}_3 &= \frac{1}{2} \gv^{de} \Lv_{ad} [ \Dv_c ( \gv - \check{\gv} )_{eb} - \Dv_e ( \gv - \check{\gv} )_{bc} ] + \oo{ M - 2 }{ \rho; \gv - \check{\gv} }_{cab} + \oo{ M - 1 }{ 1; \Lv - \check{\Lv} }_{cab} \text{,} \\
\notag \mc{I}_4 &= - \frac{1}{2} \check{\gv}^{de} \check{\Lv}_{cd} [ \Dv_a ( \gv - \check{\gv} )_{eb} - \Dv_e ( \gv - \check{\gv} )_{ba} ] + \oo{ M - 2 }{ \rho; \gv - \check{\gv} }_{cab} + \oo{ M - 1 }{ 1; \Lv - \check{\Lv} }_{cab} \text{.}
\end{align}
Combining \eqref{eql.sys_transport_2}--\eqref{eql.sys_transport_12} and recalling \eqref{Bdef} results in
\begin{align*}
\mi{L}_\rho \ms{B}_{cab} &= 2 ( \wv^1 - \check{\wv}^1 )_{bac} + \frac{1}{2} \gv^{de} \Lv_{ad} \ms{B}_{ceb} - \frac{1}{2} \gv^{de} \Lv_{cd} \ms{B}_{aeb} + \oo{M-3}{ \rho; \gv - \check{\gv} }_{cab} \\
&\qquad + \oo{M-3}{ \rho; \ms{Q} }_{cab} + \oo{ M - 1 }{ 1; \Lv - \check{\Lv} }_{cab} \text{.}
\end{align*}
The above, along with Proposition \ref{thm.aads_O_regular} and \eqref{Wdef}, immediately yield the last part of \eqref{eq.sys_transport}.

It remains to prove the derivative transport equations \eqref{eq.sys_transport_deriv}.
These are in fact immediate consequences of the preceding \eqref{eq.sys_transport}, since \eqref{eq.aads_comm} and \eqref{eq.aads_O_g} imply that
\[
\mi{L}_\rho \Dv \ms{A} = \Dv \mi{L}_\rho \Dv \ms{A} + \oo{ M - 3 }{ \rho; \ms{A} }
\]
holds for any vertical tensor field $\ms{A}$.

\subsection{Proof of Proposition \ref{thm.sys_wave}} \label{sec.sys_wave}

We prove all three formulas at once.
Throughout, we fix
\[
( \wv, \check{\wv}, \ms{W} ) \in \{ ( \wv^2, \check{\wv}^2, \ms{W}^2 ), ( \wv^1, \check{\wv}^1, \ms{W}^1 ), ( \wv^\star, \check{\wv}^1, \ms{W}^\star ) \} \text{,}
\]
as well as the corresponding constant $\sigma_{ \wv }$:
\[
\sigma_{ \wv^2 } := 2 ( n - 2 ) \text{,} \qquad \sigma_{ \wv^1 } := n - 1 \text{,} \qquad \sigma_{ \wv^\star } := 0 \text{.}
\]
Then, from \eqref{Wdef}, we see that $\ms{W}$ satisfies the following wave equation:
\begin{align}
\label{eql.sys_wave_0} ( \Boxm + \sigma_\wv ) \ms{W}_{ \bar{a} } &= [ ( \Boxm + \sigma_\wv ) \wv - ( \check{\Boxm} + \sigma_\wv ) \check{\wv} ]_{ \bar{a} } - ( \Boxm - \check{\Boxm} ) \check{\wv}_{ \bar{a} } \\
\notag &\qquad - \gv^{bc} \sum_{ j = 1 }^d ( \Boxm + \sigma_\wv ) [ \check{\wv}_{ \bar{a}_j [b] } ( \gv - \check{\gv} + \ms{Q} )_{ a_j c } ] \\
\notag &= \mi{A}_0 + \mi{A}_1 + \mi{A}_2 \text{.}
\end{align}
Here and below, we index with respect to an arbitrary coordinate system $( U, \varphi )$ on $\mi{I}$.

To expand $\mi{A}_0$, we use that both $\wv$ and $\check{\wv}$ solve wave equations of the form \eqref{eq.aads_wave}, and we take the difference of these two equations.
In particular, we take the difference of each pair of corresponding terms on the right-hand side of the relevant equation in \eqref{eq.aads_wave}.
Like in earlier proofs, each of these can be written as a sum, with each term involving a geometric difference quantity along with other coefficients that can be controlled using Proposition \ref{thm.aads_O_regular} and \eqref{eq.sys_O_QB}.
For example, we can treat\footnote{On the left-hand side, the two schematic terms have the same algebraic form.}
\begin{align*}
\rho \, \sch{ \gv^{-1}, \Lv, \wv^2 } - \rho \, \sch{ \check{\gv}^{-1}, \check{\Lv}, \check{\wv}^2 } &= \rho \, \sch{ \gv^{-1} - \check{\gv}^{-1}, \Lv, \wv^2 } + \rho \, \sch{ \check{\gv}^{-1}, \Lv - \check{\Lv}, \wv^2 } \\
&\qquad + \rho \, \sch{ \check{\gv}^{-1}, \check{\Lv}, \wv^2 - \check{\wv}^2 } \\
&= \oo{ M-2 }{ \rho^2; \gv - \check{\gv} } + \oo{ M-2 }{ \rho; \Lv - \check{\Lv} } + \oo{ M-2 }{ \rho^2; \wv^2 - \check{\wv}^2 } \text{.}
\end{align*}
For differences involving a vertical derivative, one must also take into account the difference $\Dv - \check{\Dv}$ of connections; here, the key is to write this in terms of differences of Christoffel symbols and then apply the last identitiy in \eqref{eq.sys_diff_Gamma}.
One example of this is the following:
\begin{align*}
\rho^2 \, \sch{ \gv^{-1}, \Lv, \Dv \wv^1 } - \rho^2 \, \sch{ \check{\gv}^{-1}, \check{\Lv}, \check{\Dv} \check{\wv}^1 } &= \rho^2 \, \sch{ \gv^{-1} - \check{\gv}^{-1}, \Lv, \Dv \wv^1 } + \rho^2 \, \sch{ \check{\gv}^{-1}, \Lv - \check{\Lv}, \Dv \wv^1 } \\
&\qquad + \rho^2 \, \sch{ \check{\gv}^{-1}, \check{\Lv}, \Dv ( \wv^1 - \check{\wv}^1 ) } \\
&\qquad + \rho^2 \, \sch{ \check{\gv}^{-1}, \check{\Lv}, ( \Dv - \check{\Dv} ) \check{\wv}^1 } \\
&= \oo{M-3}{ \rho^3; \gv - \check{\gv} } + \oo{M-3}{ \rho^2; \Lv - \check{\Lv} } \\
&\qquad + \oo{M-2}{ \rho^3; \Dv ( \wv^1 - \check{\wv}^1 ) } + \oo{M-2}{ \rho^3; \Dv ( \gv - \check{\gv} ) } \text{.}
\end{align*}

After a diligent analysis of all such difference terms, we obtain
\begin{align}
\label{eql.sys_wave_1} \mi{A}_0 &= \oo{ M - 3 }{ \rho^2; \gv - \check{\gv} }_{ \bar{a} } + \oo{ M - 3 }{ \rho; \Lv - \check{\Lv} }_{ \bar{a} } + \oo{M-2}{ \rho^3; \Dv ( \gv - \check{\gv} ) }_{ \bar{a} } \\
\notag &\qquad + \oo{ M - 2 }{ \rho^2; \Dv ( \Lv - \check{\Lv} ) }_{ \bar{a} } + \sum_{ \ms{v} \in \{ \wv^0, \wv^1, \wv^2 \} } [ \oo{ M - 3 }{ \rho^2; \ms{v} - \check{\ms{v}} } + \oo{ M - 2 }{ \rho^3; \Dv ( \ms{v} - \check{\ms{v}} ) } ]_{ \bar{a} } \text{.}
\end{align}
In addition, using Proposition \ref{thm.aads_O_regular} and \eqref{Wdef}, we can then express, for $k \in \{ 1, 2 \}$,
\begin{align}
\label{eql.sys_wave_2} \ms{w}^k - \check{\ms{w}}^k &= \ms{W}^k + \oo{M-2}{ 1; \gv - \check{\gv} } + \oo{M-2}{ 1; \ms{Q} } \text{,} \\
\notag \Dv ( \ms{w}^k - \check{\ms{w}}^k ) &= \Dv \ms{W}^k + \oo{M-3}{ 1; \gv - \check{\gv} } + \oo{M-3}{ 1; \ms{Q} } \\
\notag &\qquad + \oo{M-2}{ 1; \Dv ( \gv - \check{\gv} ) } + \oo{M-2}{ 1; \Dv \ms{Q} } \text{.}
\end{align}
A similar computation can be done for $\ms{w}^0 - \check{\ms{w}}^0$ and $\ms{W}^\star$, but we also recall \eqref{eq.aads_wstar} and \eqref{eql.sys_wave_2}:
\begin{align}
\label{eql.sys_wave_3} \ms{w}^0 - \check{\ms{w}}^0 &= ( \ms{w}^\star - \check{\ms{w}}^\star ) + \sch{ \gv - \check{\gv}, \wv^2 } + \sch{ \check{\gv}^{-1}, \wv^2 - \check{\wv}^2 } \\
\notag &= \ms{W}^\star + \oo{M-2}{ 1; \gv - \check{\gv} } + \oo{M-2}{ 1; \ms{Q} } + \oo{M}{ 1; \ms{W}^2 } \text{,} \\
\notag \Dv ( \ms{w}^0 - \check{\ms{w}}^0 ) &= \Dv \ms{W}^\star + \oo{M-3}{ 1; \gv - \check{\gv} } + \oo{M-3}{ 1; \ms{Q} } + \oo{M-2}{ 1; \Dv ( \gv - \check{\gv} ) } \\
\notag &\qquad + \oo{M-2}{ 1; \Dv \ms{Q} } + \oo{M-1}{ 1; \ms{W}^2 } + \oo{M}{ 1; \Dv \ms{W}^2 } \text{.}
\end{align}
Combining \eqref{eql.sys_wave_1}--\eqref{eql.sys_wave_3}, we conclude that
\begin{align}
\label{eql.sys_wave_4} \mi{A}_0 &= \oo{ M - 3 }{ \rho^2; \gv - \check{\gv} }_{ \bar{a} } + \oo{ M - 3 }{ \rho^2; \ms{Q} } + \oo{ M - 3 }{ \rho; \Lv - \check{\Lv} }_{ \bar{a} } \\
\notag &\qquad + \oo{M-2}{ \rho^3; \Dv ( \gv - \check{\gv} ) }_{ \bar{a} } + \oo{M-2}{ \rho^3; \Dv \ms{Q} }_{ \bar{a} } + \oo{ M - 2 }{ \rho^2; \Dv ( \Lv - \check{\Lv} ) }_{ \bar{a} } \\
\notag &\qquad + \sum_{ \ms{V} \in \{ \ms{W}^\star, \ms{W}^1, \ms{W}^2 \} } [ \oo{ M - 3 }{ \rho^2; \ms{V} } + \oo{ M - 2 }{ \rho^3; \Dv \ms{V} } ]_{ \bar{a} } \text{.}
\end{align}

Next, for $\mi{A}_1$, we apply \eqref{eq.aads_boxm} twice---once with $\Boxm$ and once with $\check{\Boxm}$---and subtract the resulting equations in order to obtain (note that $\mi{L}_\rho$ is independent of the metric)
\begin{align*}
\mi{A}_1 &= - \rho^2 [ \gv^{bc} \Dv_{bc} - \check{\gv}^{bc} \check{\Dv}_{bc} ] \check{\wv}_{ \bar{a} } + \rho^2 \, [ \sch{ \gv^{-1}, \Lv, \mi{L}_\rho \check{\wv} } - \sch{ \check{\gv}^{-1}, \check{\Lv}, \mi{L}_\rho \check{\wv} } ]_{ \bar{a} } \\
&\qquad + \rho \, [ \sch{ \gv^{-1}, \Lv, \check{\wv} } - \sch{ \check{\gv}^{-1}, \check{\Lv}, \check{\wv} } ]_{ \bar{a} } + \rho^2 \, [ \sch{ \gv^{-1}, \mi{L}_\rho \Lv, \check{\wv} } - \sch{ \check{\gv}^{-1}, \mi{L}_\rho \check{\Lv}, \check{\wv} } ]_{ \bar{a} } \\
&\qquad + \rho^2 \, [ \sch{ \gv^{-2}, \Lv, \Lv, \check{\wv} } - \sch{ \check{\gv}^{-2}, \check{\Lv}, \check{\Lv}, \check{\wv} } ]_{ \bar{a} } \text{,}
\end{align*}
where each matching pair of schematic terms has the same algebraic form.
Similar to the previous treatment of $\mi{A}_0$, each schematic difference on the right-hand side of the above can be expanded and then controlled using Proposition \ref{thm.aads_O_regular} and \eqref{eq.sys_O_QB}; this then yields
\begin{align*}
\mi{A}_1 &= - \rho^2 [ \gv^{bc} \Dv_{bc} - \check{\gv}^{bc} \check{\Dv}_{bc} ] \check{\wv}_{ \bar{a} } + \oo{M-3}{ \rho^2; \gv - \check{\gv} }_{ \bar{a} } + \oo{M-3}{ \rho, \Lv - \check{\Lv} }_{ \bar{a} } \\
&\qquad + \oo{M-2}{ \rho^2; \mi{L}_\rho ( \Lv - \check{\Lv} ) }_{ \bar{a} } \text{.}
\end{align*}
In addition, we expand the last term of the above using the third identity in \eqref{eq.sys_transport}:
\begin{align}
\label{eql.sys_wave_10} \mi{A}_1 &= - \rho^2 [ \gv^{bc} \Dv_{bc} - \check{\gv}^{bc} \check{\Dv}_{bc} ] \check{\wv}_{ \bar{a} } + \oo{M-3}{ \rho^2; \gv - \check{\gv} }_{ \bar{a} } + \oo{M-3}{ \rho^3; \ms{Q} }_{ \bar{a} } \\
\notag &\qquad + \oo{M-3}{ \rho, \Lv - \check{\Lv} }_{ \bar{a} } + \oo{M-2}{ \rho^2; \ms{W}^2 }_{ \bar{a} } \\
\notag &= \mi{A}_{ 1, v } + \oo{M-3}{ \rho^2; \gv - \check{\gv} }_{ \bar{a} } + \oo{M-3}{ \rho^3; \ms{Q} }_{ \bar{a} } + \oo{M-3}{ \rho, \Lv - \check{\Lv} }_{ \bar{a} } + \oo{M-2}{ \rho^2; \ms{W}^2 }_{ \bar{a} } \text{.}
\end{align}

To handle $\mi{A}_{ 1, v }$, we begin by expanding
\begin{align*}
\mi{A}_{ 1, v } &= - \rho^2 \gv^{bc} ( \Dv_{bc} - \check{\Dv}_{bc} ) \check{\wv}_{ \bar{a} } - \rho^2 ( \gv^{bc} - \check{\gv}^{bc} ) \check{\Dv}_{bc} \check{\wv}_{ \bar{a} } \\
&= - \rho^2 \gv^{bc} \Dv_b ( \Dv - \check{\Dv} )_c \check{\wv}_{ \bar{a} } - \rho^2 \gv^{bc} ( \Dv - \check{\Dv} )_b \check{\Dv}_c \check{\wv}_{ \bar{a} } - \rho^2 ( \gv^{bc} - \check{\gv}^{bc} ) \check{\Dv}_{bc} \check{\wv}_{ \bar{a} } \text{.}
\end{align*}
The operator $\Dv - \check{\Dv}$ can be expressed in terms of differences of Christoffel symbols.
Applying again Proposition \ref{thm.aads_O_regular}, \eqref{eq.sys_O_QB}, and the last part of \eqref{eq.sys_diff_Gamma}, the above becomes
\begin{align}
\label{eql.sys_wave_11} \mi{A}_{ 1, v } &= \rho^2 \gv^{bc} \sum_{ j = 1 }^d \Dv_b [ ( \Gammav - \check{\Gammav} )^d_{ c a_j } \check{\wv}_{ \bar{a}_j [d] } ] + \oo{ M - 3 }{ \rho^2; \Dv ( \gv - \check{\gv} ) }_{ \bar{a} } + \oo{M-4}{ \rho^2; \gv - \check{\gv} }_{ \bar{a} } \\
\notag &= \rho^2 \gv^{bc} \sum_{ j = 1 }^d \Dv_b \{ \check{\gv}^{de} [ \Dv_c ( \gv - \check{\gv} )_{ e a_j } + \Dv_{ a_j } ( \gv - \check{\gv} )_{ e c } - \Dv_e ( \gv - \check{\gv} )_{ c a_j } ] \check{\wv}_{ \bar{a}_j [d] } \} \\
\notag &\qquad + \oo{M-4}{ \rho^2; \gv - \check{\gv} }_{ \bar{a} } + \oo{ M - 3 }{ \rho^2; \Dv ( \gv - \check{\gv} ) }_{ \bar{a} } \\
\notag &= \rho^2 \gv^{bc} \gv^{de} \sum_{ j = 1 }^d [ \Dv_{bc} ( \gv - \check{\gv} )_{ e a_j } + \Dv_{ b a_j } ( \gv - \check{\gv} )_{ e c } - \Dv_{be} ( \gv - \check{\gv} )_{ c a_j } ] \check{\wv}_{ \bar{a}_j [d] } \\
\notag &\qquad + \oo{M-4}{ \rho^2; \gv - \check{\gv} }_{ \bar{a} } + \oo{ M - 3 }{ \rho^2; \Dv ( \gv - \check{\gv} ) }_{ \bar{a} } \text{.}
\end{align}
(Notice in the last step, any term having less than two derivatives of $\gv - \check{\gv}$ can be included with the schematic terms.
In addition, $\check{\gv}^{de}$ can be converted to $\gv^{de}$ at the cost of an extra term involving $\gv - \check{\gv}$, which can also be absorbed into the schematic terms.)

Recalling the definitions \eqref{Qdef} and \eqref{Bdef}, then \eqref{eql.sys_wave_11} can be rewritten as
\begin{align*}
\mi{A}_{ 1, v } &= \rho^2 \gv^{bc} \gv^{de} \sum_{ j = 1 }^d [ \Dv_{bc} ( \gv - \check{\gv} )_{ e a_j } + \Dv_{ b c } \ms{Q}_{ a_j e } + \Dv_b \ms{B}_{ a_j e c } ] \check{\wv}_{ \bar{a}_j [d] } \\
&\qquad + \oo{M-4}{ \rho^2; \gv - \check{\gv} }_{ \bar{a} } + \oo{ M - 3 }{ \rho^2; \Dv ( \gv - \check{\gv} ) }_{ \bar{a} } \\
&= \rho^2 \gv^{bc} \gv^{de} \sum_{ j = 1 }^d \Dv_{bc} ( \gv - \check{\gv} + \ms{Q} )_{ a_j e } \check{\wv}_{ \bar{a}_j [d] } + \oo{ M - 2 }{ \rho^2; \Dv \ms{B} }_{ \bar{a} } \\
&\qquad + \oo{M-4}{ \rho^2; \gv - \check{\gv} }_{ \bar{a} } + \oo{ M - 3 }{ \rho^2; \Dv ( \gv - \check{\gv} ) }_{ \bar{a} } \text{.}
\end{align*}
Combining \eqref{eql.sys_wave_10} and the above, we conclude that
\begin{align}
\label{eql.sys_wave_12} \mi{A}_1 &= \rho^2 \gv^{bc} \gv^{de} \sum_{ j = 1 }^d \Dv_{bc} ( \gv - \check{\gv} + \ms{Q} )_{ a_j e } \check{\wv}_{ \bar{a}_j [d] } + \oo{M-4}{ \rho^2; \gv - \check{\gv} }_{ \bar{a} } \\
\notag &\qquad + \oo{M-3}{ \rho^3; \ms{Q} }_{ \bar{a} } + \oo{M-3}{ \rho, \Lv - \check{\Lv} }_{ \bar{a} } + \oo{ M - 3 }{ \rho^2; \Dv ( \gv - \check{\gv} ) }_{ \bar{a} } \\
\notag &\qquad + \oo{ M - 2 }{ \rho^2; \Dv \ms{B} }_{ \bar{a} } + \oo{M-2}{ \rho^2; \ms{W}^2 }_{ \bar{a} } \text{.}
\end{align}

Now, for the remaining term $\mi{A}_2$, we begin by expanding
\begin{align}
\label{eql.sys_wave_20} \mi{A}_2 &= - \gv^{de} \sum_{ j = 1 }^d \check{\wv}_{ \bar{a}_j [d] } \, \Boxm ( \gv - \check{\gv} + \ms{Q} )_{ a_j e } - \gv^{de} \sum_{ j = 1 }^d ( \Boxm + \sigma_\wv ) \check{\wv}_{ \bar{a}_j [d] } \, ( \gv - \check{\gv} + \ms{Q} )_{ a_j e } \\
\notag &\qquad - \rho^2 \gv^{de} \sum_{ j = 1 }^d \Dvm_\rho \check{\wv}_{ \bar{a}_j [d] } \, \Dvm_\rho ( \gv - \check{\gv} + \ms{Q} )_{ a_j e } - \rho^2 \gv^{bc} \gv^{de} \sum_{ j = 1 }^d \Dv_b \check{\wv}_{ \bar{a}_j [d] } \, \Dv_c ( \gv - \check{\gv} + \ms{Q} )_{ a_j e } \\
\notag &= \mi{A}_{ 2, 1 } + \mi{A}_{ 2, 2 } + \mi{A}_{ 2, 3 } + \mi{A}_{ 2, 4 } \text{.}
\end{align}
We now expand each of the terms on the right-hand side of \eqref{eql.sys_wave_20}.
First, by Proposition \ref{thm.aads_O_regular},
\begin{equation}
\label{eql.sys_wave_21} \mi{A}_{ 2, 4 } = \oo{ M - 3 }{ \rho^2; \Dv ( \gv - \check{\gv} ) }_{ \bar{a} } + \oo{ M - 3 }{ \rho^2; \Dv \ms{Q} }_{ \bar{a} } \text{.}
\end{equation}
Next, for $\mi{A}_{ 2, 3 }$, we recall \eqref{eq.aads_vertical_connection}, along with \eqref{Qdef} and the first part of \eqref{eq.sys_transport}:
\begin{align}
\label{eql.sys_wave_22} \mi{A}_{ 2, 3 } &= \oo{ M - 3 }{ \rho^2; \mi{L}_\rho ( \gv - \check{\gv} ) }_{ \bar{a} } + \oo{ M - 3 }{ \rho^2; \mi{L}_\rho \ms{Q} }_{ \bar{a} } + \oo{ M - 3 }{ \rho^3; \gv - \check{\gv} }_{ \bar{a} } + \oo{ M - 3 }{ \rho^3; \ms{Q} }_{ \bar{a} } \\
\notag &= \oo{ M - 3 }{ \rho^3; \gv - \check{\gv} }_{ \bar{a} } + \oo{ M - 3 }{ \rho^3; \ms{Q} }_{ \bar{a} } + \oo{ M - 3 }{ \rho^2; \Lv - \check{\Lv} }_{ \bar{a} } \text{.}
\end{align}
For $\mi{A}_{ 2, 2 }$, we begin by applying the wave equation \eqref{eq.aads_wave} satisfied by $\wv$.
By inspecting the right-hand of the appropriate equation and recalling Proposition \ref{thm.aads_O_regular}, we obtain
\begin{align*}
( \Boxm + \sigma_{ \wv } ) \wv &= \sum_{ k = 0 }^2 \oo{ M - 2 }{ \rho^3; \Dv \wv^k } + \sum_{ k = 0 }^2 \oo{ M - 3 }{ \rho^2; \wv^k } \\
&= \oo{ M - 3 }{ \rho^2 } \text{.}
\end{align*}
From the above, we then obtain the following asymptotics:
\begin{equation}
\label{eql.sys_wave_23} \mi{A}_{ 2, 2 } = \oo{ M - 3 }{ \rho^2; \gv - \check{\gv} }_{ \bar{a} } + \oo{ M - 3 }{ \rho^2; \ms{Q} }_{ \bar{a} } \text{.}
\end{equation}

For $\mi{A}_{ 2, 1 }$, we decompose $\Boxm$ using \eqref{eq.aads_boxm} and apply Proposition \ref{thm.aads_O_regular}:
\begin{align}
\label{eql.sys_wave_24} \mi{A}_{ 2, 1 } &= - \rho^2 \gv^{bc} \gv^{de} \sum_{ j = 1 }^d \check{\wv}_{ \bar{a}_j [d] } \, \Dv_{bc} ( \gv - \check{\gv} + \ms{Q} )_{ a_j e } + \oo{ M - 2 }{ \rho^2; \mi{L}_\rho^2 ( \gv - \check{\gv} + \ms{Q} ) }_{ \bar{a} } \\
\notag &\qquad + \oo{ M - 2 }{ \rho; \mi{L}_\rho ( \gv - \check{\gv} + \ms{Q} ) }_{ \bar{a} } + \oo{ M - 2 }{ \rho^2; \gv - \check{\gv} + \ms{Q} }_{ \bar{a} } \text{.}
\end{align}
For the asympototic terms on the right-hand side, we recall \eqref{Qdef} and \eqref{eq.sys_transport}:
\begin{align*}
\mi{L}_\rho ( \gv - \check{\gv} + \ms{Q} ) &= \Lv - \check{\Lv} + \oo{ M - 2 }{ \rho; \gv - \check{\gv} } + \oo{ M - 2 }{ \rho; \ms{Q} } \text{,} \\
\mi{L}_\rho^2 ( \gv - \check{\gv} + \ms{Q} ) &= -2 \ms{W}^2 + \oo{ M - 3 }{ 1; \gv - \check{\gv} } + \oo{ M - 3 }{ 1; \ms{Q} } + \oo{ M - 2 }{ \rho^{-1}; \Lv - \check{\Lv} } \text{.}
\end{align*}
Thus, combining \eqref{eql.sys_wave_24} and the above yields
\begin{align*}
\mi{A}_{ 2, 1 } &= - \rho^2 \gv^{bc} \gv^{de} \sum_{ j = 1 }^d \check{\wv}_{ \bar{a}_j [d] } \, \Dv_{bc} ( \gv - \check{\gv} + \ms{Q} )_{ a_j e } + \oo{ M - 3 }{ \rho^2; \gv - \check{\gv} }_{ \bar{a} } + \oo{ M - 3 }{ \rho^2; \ms{Q} }_{ \bar{a} } \\
&\qquad + \oo{ M - 2 }{ \rho; \Lv - \check{\Lv} }_{ \bar{a} } + \oo{ M - 2 }{ \rho^2; \ms{W}^2 }_{ \bar{a} } \text{,}
\end{align*}
Combining \eqref{eql.sys_wave_20}--\eqref{eql.sys_wave_23} and the above, we have that
\begin{align}
\label{eql.sys_wave_25} \mi{A}_2 &= - \rho^2 \gv^{bc} \gv^{de} \sum_{ j = 1 }^d \check{\wv}_{ \bar{a}_j [d] } \, \Dv_{bc} ( \gv - \check{\gv} + \ms{Q} )_{ a_j e } + \oo{ M - 3 }{ \rho^2; \gv - \check{\gv} }_{ \bar{a} } \\
\notag &\qquad + \oo{ M - 3 }{ \rho^2; \ms{Q} }_{ \bar{a} } + \oo{ M - 2 }{ \rho; \Lv - \check{\Lv} }_{ \bar{a} } + \oo{ M - 3 }{ \rho^2; \Dv ( \gv - \check{\gv} ) }_{ \bar{a} } \\
\notag &\qquad + \oo{ M - 3 }{ \rho^2; \Dv \ms{Q} }_{ \bar{a} } + \oo{ M - 2 }{ \rho^2; \ms{W}^2 }_{ \bar{a} } \text{.}
\end{align}

Finally, both \eqref{eq.sys_wave} and \eqref{Fdef} follow from \eqref{eql.sys_wave_0}, \eqref{eql.sys_wave_4}, \eqref{eql.sys_wave_12}, and \eqref{eql.sys_wave_25}.

\subsection{Proof of Proposition \ref{thm.sys_vanish}} \label{sec.sys_vanish}

The first part of \eqref{eq.sys_vanish} is simply the first equation in \eqref{eq.sys_transport}, while the second identity in \eqref{eq.sys_vanish} follows from the final part of \eqref{eq.aads_connection}, along with Proposition \ref{thm.aads_O_regular} and \eqref{eq.sys_O_QB}.
(Note this is just the third part of \eqref{eq.sys_transport}, except we do not renormalize $\wv^2 - \check{\wv}^2$.)

For the remaining identities, we begin by rewriting the first, third, and fourth parts of \eqref{eq.aads_bianchi} as
\begin{align}
\label{eql.sys_bianchi_0} \mi{L}_\rho ( \rho^{ 2-n } \wv^2 )_{ a b } &= \rho^{ 2-n } \gv^{ c d } \Dv_c \wv^1_{ b a d } + \rho^{ 2-n } \, \sch{ \gv^{-2}, \Lv, \wv^0 }_{ a b } + \rho^{ 2-n } \, \sch{ \gv^{-1}, \Lv, \wv^2 }_{ a b } \text{,} \\
\notag \mi{L}_\rho ( \rho^{-1} \wv^1 )_{ a b c } &= \rho^{-1} \Dv_b \wv^2_{ a c } - \rho^{-1} \Dv_c \wv^2_{ a b } + \rho^{-1} \, \sch{ \gv^{-1}, \Lv, \wv^1 }_{ a b c } \text{,} \\
\notag \mi{L}_\rho \wv^0_{ a b c d } &= \Dv_a \wv^1_{ b c d } - \Dv_b \wv^1_{ a c d } + \rho^{-1} \gv_{ a d } \wv^2_{ b c } + \rho^{-1} \gv_{ b c } \wv^2_{ a d } - \rho^{-1} \gv_{ a c } \wv^2_{ b d } - \rho^{-1} \gv_{ b d } \wv^2_{ a c } \\
\notag &\qquad + \sch{ \gv^{-1}, \Lv, \wv^0 }_{ a b c d } + \sch{ \Lv, \wv^2 }_{ a b c d } \text{.}
\end{align}
In particular, here we converted $\Dvm_\rho$ to $\mi{L}_\rho$ using \eqref{eq.aads_vertical_connection}.
We then take the difference of each of the equations in \eqref{eql.sys_bianchi_0}, applied once with respect to $\gv$ and once with respect to $\check{\gv}$, and we treat the resulting difference terms as in the proofs of Propositions \ref{thm.sys_transport} and \ref{thm.sys_wave}.
In particular, various geometric quantities can be controlled using Proposition \ref{thm.aads_O_regular} and \eqref{eq.sys_O_QB}; moreover, the difference $\Dv - \check{\Dv}$ can be controlled by $\Dv ( \gv - \check{\gv} )$ using the last part of \eqref{eq.sys_diff_Gamma}.

For example, from the first part of \eqref{eql.sys_bianchi_0}, we have
\begin{align*}
\mi{L}_\rho [ \rho^{2-n} ( \wv^2 - \check{\wv}^2 ) ]_{ab} &= \rho^{2-n} ( \gv^{cd} - \check{\gv}^{cd} ) \Dv_c \wv^1_{ b a d } + \rho^{ 2-n } \check{\gv}^{cd} \Dv_c ( \wv^1 - \check{\wv}^1 )_{ b a d } + \rho^{ 2-n } \check{\gv}^{ c d } ( \Dv - \check{\Dv} )_c \check{\wv}^1_{ b a d } \\
&\qquad + [ \rho^{ 2-n } \, \sch{ \gv^{-2}, \Lv, \wv^0 } - \rho^{ 2-n } \, \sch{ \check{\gv}^{-2}, \check{\Lv}, \check{\wv}^0 } ]_{ a b } \\
&\qquad + [ \rho^{ 2-n } \, \sch{ \gv^{-1}, \Lv, \wv^2 } - \rho^{ 2-n } \, \sch{ \gv^{-1}, \Lv, \wv^2 } ]_{ a b } \text{,} \\
&= \oo{M}{ \rho^{ 2-n }; \Dv ( \wv^1 - \check{\wv}^1 ) }_{ab} + \oo{ M - 4 }{ \rho^{ 3-n }; \gv - \check{\gv} }_{ab} \\
&\qquad + \oo{ M - 2 }{ \rho^{ 2-n }; \Lv - \check{\Lv} }_{ab} + \oo{ M - 3 }{ \rho^{ 3-n }; \Dv ( \gv - \check{\gv} ) }_{ab} \\
&\qquad + \oo{ M - 2 }{ \rho^{ 3-n }; \wv^0 - \check{\wv}^0 } + \oo{ M - 2 }{ \rho^{ 3-n }; \wv^2 - \check{\wv}^2 } \text{,}
\end{align*}
which yields the third equation in \eqref{eq.sys_vanish}.
Note that in the last step, we used the improved vanishing rate for $\wv^1$ in \eqref{eq.aads_O_w} when dealing with terms involving $\gv - \check{\gv}$ and $\Dv ( \gv - \check{\gv} )$.
The remaining parts of \eqref{eq.sys_vanish} are obtained analogously, once again using the improved rates for $\wv^1$ and $\wv^2$ in \eqref{eq.aads_O_w}.